\documentclass[11pt,letterpaper]{article}
\usepackage{thmtools}
\usepackage{thm-restate}
\usepackage{bm}
\usepackage{mathrsfs}           %
\usepackage{vmargin,fancyhdr}   %
\usepackage{tikz}
\usetikzlibrary{positioning,chains,fit,shapes,calc}
\usepackage{algorithm}
\usepackage{algpseudocode}

\usepackage{subcaption}
\usepackage{enumerate}

\usepackage{amsmath,amssymb, amsthm}    %
\usepackage{verbatim}           %
\usepackage{xspace}             %
\usepackage{graphicx,float}     %
\usepackage{ifthen,calc}        %
\usepackage{textcomp}           %
\usepackage{fancybox}           %
\usepackage{hhline}             %
\usepackage{float}              %

\usepackage[vflt]{floatflt}     %
\usepackage[small,compact]{titlesec}
\usepackage{setspace}
\usepackage{color}
\definecolor{ForestGreen}{rgb}{0.1333,0.5451,0.1333}
\usepackage{thumbpdf}
\usepackage[letterpaper, colorlinks,linkcolor=ForestGreen,citecolor=ForestGreen, bookmarks,bookmarksopen,bookmarksnumbered]
{hyperref}

\setpapersize{USletter}
\setmarginsrb{1in}{.5in}        %
{1in}{.5in}        %
{.25in}{.25in}     %
{.25in}{.5in}      %
\newcommand{\showccc}[0]{0}
\newcommand{\ccc}[2][nothing]{%
	\ifthenelse{\showccc=0}{}{
		\ensuremath{^{\Lsh\Rsh}}\marginpar{\raggedright\tiny\textsf{%
				\ifthenelse{\equal{#1}{nothing}}{}{\textbf{#1}\\}#2}}}}
\newcounter{hours}\newcounter{minutes}
\newcommand{\hhmm}{%
	\setcounter{hours}{\time/60}%
	\setcounter{minutes}{\time-\value{hours}*60}%
	\ifthenelse{\value{hours}<10}{0}{}\thehours:%
	\ifthenelse{\value{minutes}<10}{0}{}\theminutes}
\ifthenelse{\showccc=0}{\rhead{}}{\rhead{\today \ [\hhmm]}}
\newtheorem{proposition}{Proposition}
\newtheorem{corollary}{Corollary}
\newtheorem{definition}{Definition}

\newtheorem{remark}{Remark}
\newtheorem{lemma}{Lemma}
\newtheorem{fact}{Fact}

\newtheorem{observation}{Observation}

\usepackage{float}
\floatstyle{plain}\newfloat{myfig}{t}{figs}[section]
\floatname{myfig}{\textsc{Figure}}
\floatstyle{plain}\newfloat{myalg}{H}{algs}[section]
\floatname{myalg}{}
\newcommand{\defeq}{:=}
\newcommand{\norm}[1]{\left\lVert#1\right\rVert}
\newcommand{\inprod}[2]{\left\langle#1, #2\right\rangle}
\newcommand{\tvd}[2]{\norm{#1 - #2}_{\textup{TV}}}
\newcommand{\fwc}{f_{\textup{wc}}}
\newcommand{\oracle}{\mathcal{O}}
\newcommand{\pih}{\hat{\pi}}
\newcommand{\yh}{\hat{y}}
\newcommand{\xh}{\hat{x}}
\newcommand{\fcomp}{f_{\textup{oracle}}}
\newcommand{\eps}{\epsilon}
\newcommand{\lam}{\lambda}

\newcommand{\argmin}{\textup{argmin}} 
\newcommand{\R}{\mathbb{R}}

\newcommand{\N}{\mathbb{N}}

\newcommand{\half}{\frac{1}{2}}
\newcommand{\thalf}{\tfrac{1}{2}}

\newcommand{\E}{\mathbb{E}}
\newcommand{\Var}{\textup{Var}}

\newcommand{\Nor}{\mathcal{N}}

\newcommand{\xset}{\mathcal{X}}

\newcommand{\tP}{\widetilde{P}}
\newcommand{\tQ}{\widetilde{Q}}

\newcommand{\id}{\mathbf{I}}
\newcommand{\tO}{\widetilde{O}}
\newcommand{\tOh}[1]{\tilde{O}\left(#1\right)}

\newcommand{\Par}[1]{\left(#1\right)}
\newcommand{\Brack}[1]{\left[#1\right]}
\newcommand{\Brace}[1]{\left\{#1\right\}}
\newcommand{\covar}{\boldsymbol{\Sigma}}

\newcommand{\tx}{\tilde{x}}
\newcommand{\alg}{\mathcal{A}}

\newcommand{\prop}{\mathcal{P}}
\newcommand{\tran}{\mathcal{T}}

\newcommand{\pistart}{\pi_{\textup{start}}}
\newcommand{\oracleplus}{\oracle^{+}}
\newcommand{\hp}{\hat{p}}
\newcommand{\AlternateSample}{\texttt{AlternateSample}}
\newcommand{\XSample}{\texttt{XSample}}
\newcommand{\MRW}{\texttt{FiniteSum-MRW}}
\newcommand{\IMRW}{\texttt{Inefficient-MRW}}
\newcommand{\ty}{\tilde{y}}

\newcommand{\by}{\bar{y}}
\newcommand{\bya}{\bar{y}_\alpha}
\newcommand{\sya}{y^*_\alpha}
\newcommand{\bal}{\bar{\alpha}}
\newcommand{\hal}{\hat{\alpha}}

\newcommand{\zstart}{Z_{\textup{start}}}
\newcommand{\jac}{\mathbf{J}}

\newcommand{\csg}{\texttt{Composite-Sample}}
\newcommand{\cssm}{\texttt{Composite-Sample-Shared-Min}}
\newcommand{\sjd}{\texttt{Sample-Joint-Dist}}
\newcommand{\yor}{\texttt{YSample}}

\newcommand{\dkl}{d_{\textup{KL}}}

\definecolor{burntorange}{rgb}{0.8, 0.33, 0.0}

  \usepackage{nth}
  \usepackage{intcalc}

\usepackage{url}
\usepackage{multirow}

\begin{document}

	\begin{titlepage}
		\def\thepage{}
		\thispagestyle{empty}
		
		\title{Structured Logconcave Sampling with a Restricted Gaussian Oracle} 
		
		\date{}
		\author{
			Yin Tat Lee\thanks{University of Washington and Microsoft Research, {\tt yintat@uw.edu}}
			\and
			Ruoqi Shen\thanks{University of Washington, {\tt shenr3@cs.washington.edu}}
			\and
			Kevin Tian\thanks{Stanford University, {\tt kjtian@stanford.edu}}
		}
		
		\maketitle

\abstract{
We give algorithms for sampling several structured logconcave families to high accuracy.\footnote{We say a sampler is ``high-accuracy'' if its mixing time has polylogarithmic dependence on the target accuracy $\epsilon$.} We further develop a reduction framework, inspired by \emph{proximal point methods} in convex optimization, which bootstraps samplers for regularized densities to generically improve dependences on problem conditioning $\kappa$ from polynomial to linear. A key ingredient in our framework is the notion of a ``restricted Gaussian oracle'' (RGO) for $g: \R^d \rightarrow \R$, which is a sampler for distributions whose negative log-likelihood sums a quadratic (in a multiple of the identity) and $g$. By combining our reduction framework with our new samplers, we obtain the following bounds for sampling structured distributions to total variation distance $\eps$.
\begin{itemize}
	\item For composite densities $\exp(-f(x) - g(x))$, where $f$ has condition number $\kappa$ and convex (but possibly non-smooth) $g$ admits an RGO, we obtain a mixing time of $O(\kappa d \log^3\tfrac{\kappa d}{\epsilon})$, matching the state-of-the-art non-composite bound \cite{LeeST20}. No composite samplers with better mixing than general-purpose logconcave samplers were previously known.
	\item For logconcave finite sums $\exp(-F(x))$, where $F(x) = \tfrac{1}{n}\sum_{i \in [n]} f_i(x)$ has condition number $\kappa$, we give a sampler querying $\widetilde{O}(n + \kappa\max(d, \sqrt{nd}))$ gradient oracles\footnote{For convenience of exposition, the $\widetilde{O}$ notation hides logarithmic factors in the dimension $d$, problem conditioning $\kappa$, desired accuracy $\epsilon$, and summand count $n$ (when applicable). A first-order (gradient) oracle for $f:\R^d \rightarrow \R$ returns $(f(x), \nabla f(x))$ on input $x$, and a zeroth-order (value) oracle only returns $f(x)$.} to $\{f_i\}_{i \in [n]}$. No high-accuracy samplers with nontrivial gradient query complexity were previously known.
	\item For densities with condition number $\kappa$, we give an algorithm obtaining mixing time $O(\kappa d \log^2\tfrac{\kappa d}{\epsilon})$, improving \cite{LeeST20} by a logarithmic factor with a significantly simpler analysis. We also show a zeroth-order algorithm attains the same query complexity.
\end{itemize}
}
 		
	\end{titlepage}

	\pagenumbering{gobble}
	\newpage
	\setcounter{tocdepth}{2}
	{
		\hypersetup{linkcolor=black}
		\tableofcontents
	}
	\newpage
	\pagenumbering{arabic}

\section{Introduction}
\label{sec:intro}

Since its study was pioneered by the celebrated randomized convex body volume approximation algorithm of Dyer, Frieze, and Kannan \cite{DyerFK91}, designing samplers for logconcave distributions has been a very active area of research in theoretical computer science and statistics with many connections to other fields. In a generic form, the problem can be stated as: sample from a distribution whose negative log-density is convex, under various access models to the distribution. 

Developing efficient algorithms for sampling from \emph{structured} logconcave densities is a topic that has received significant recent interest due to its widespread practical applications. There are many types of structure which densities commonplace in applications may possess that are exploitable for improved runtimes. Examples of such structure include derivative bounds (``well-conditioned densities'') and various types of separability (e.g.\ ``composite densities'' corresponding to possibly non-smooth regularization or restrictions to a set, and ``logconcave finite sums'' corresponding to averages over multiple data points).\footnote{We make this terminology more precise in Section~\ref{ssec:notation}, which contains various definitions used in this paper.} Building an algorithmic theory for sampling these latter two families, which are not well-understood in the literature, is a primary motivation of this work.

There are strong parallels between the types of structured logconcave families garnering recent attention and the classes of convex functions known to admit efficient first-order optimization algorithms. Notably, gradient descent and its accelerated counterpart \cite{Nesterov83} are well-known to quickly optimize a well-conditioned function, and have become ubiquitous in both practice and theory. Similarly, various methods have been developed for efficiently optimizing non-smooth but structured composite objectives \cite{BeckT09} and well-conditioned finite sums \cite{Allen-Zhu17}.

Logconcave sampling and convex optimization are intimately related primitives (cf.\ e.g.\ \cite{BertsimasV04, AbernethyH16}), so it is perhaps unsurprising that there are analogies between the types of structure algorithm designers may exploit. Nonetheless, our understanding of the complexity landscape for sampling is quite a bit weaker in comparison to counterparts in the field of optimization; few lower bounds are known for the complexity of sampling tasks, and obtaining stronger upper bounds is an extremely active research area (contrary to optimization, where matching bounds exist in many cases). Moreover (and perhaps relatedly), the toolkit for designing logconcave samplers is comparatively lacking; for many important primitives in optimization, it is unclear if there are analogs in sampling, possibly impeding improved bounds. Our work broadly falls under the themes of (1) understanding which types of structured logconcave distributions admit efficient samplers, and (2) leveraging connections between optimization and sampling for algorithm design. We address these problems on two fronts, which constitute the primary technical contributions of this paper.

\begin{enumerate}
	\item We give a general reduction framework for bootstrapping samplers with mixing times with polynomial dependence on a conditioning measure $\kappa$ to mixing times with linear dependence on $\kappa$. The framework is heavily motivated by a perspective on \emph{proximal point methods} in structured convex optimization as instances of optimizing composite objectives, and leverages this connection via a surprisingly simple analysis (cf.\ Theorem~\ref{thm:fzerocomp}).
	\item We develop novel ``base samplers'' for composite logconcave distributions and logconcave finite sums (cf.\ Theorems~\ref{thm:mainclaim},~\ref{thm:base_finitesum}). The former is the first composite sampler with stronger guarantees than those known in the general logconcave setting. The latter constitutes the first high-accuracy finite sum sampler whose gradient query complexity improves upon the na\"ive strategy of querying full gradients of the negative log-density in each iteration. 
\end{enumerate}

Using our novel base samplers within our reduction framework, we obtain state-of-the-art samplers for all of the aforementioned structured families, i.e.\ well-conditioned, composite, and finite sum, as Corollaries~\ref{thm:kd},~\ref{thm:kdcomp}, and~\ref{thm:improved_finitesum}. We emphasize that even without our reduction technique, the guarantees of our base samplers for composite and finite sum-structured densities are the first of their kind. However, by boosting their mixing via our reduction, we obtain guarantees for these structured distribution families which are essentially the best one can hope for without a significant improvement in the most commonly studied well-conditioned regime (cf.\ discussion in Section~\ref{ssec:contribs}).

We now formally state our results in Section~\ref{ssec:contribs}, and situate them in the literature in Section~\ref{ssec:prev}. Section~\ref{ssec:technical} is a technical overview of our approaches for developing our base samplers for composite and finite sum-structured densities (Sections~\ref{ssec:compositeintro} and~\ref{ssec:finitesumintro}), as well as our proximal reduction framework (Section~\ref{ssec:overview}). Finally, Section~\ref{ssec:roadmap} gives a roadmap for the rest of the paper.

\subsection{Our results}\label{ssec:contribs}

Before stating our results, we first require the notion of a restricted Gaussian oracle, whose definition is a key ingredient in giving our reduction framework as well as our later composite samplers.

\begin{definition}[Restricted Gaussian oracle]
	\label{def:rgoracle}
	$\oracle(\lam, v)$ is a \emph{restricted Gaussian oracle} (RGO) for convex $g: \R^d \rightarrow \R$ if it returns
	\[\oracle(\lambda, v) \gets \textup{sample from the distribution with density} \propto \exp\left(-\frac{1}{2\lambda}\norm{x - v}_2^2 - g(x)\right).\]
\end{definition}

In other words, an RGO asks to sample from a multivariate Gaussian (with covariance a multiple of the identity), ``restricted'' by some convex function $g$. Intuitively, if we can reduce a sampling problem for the density $\propto \exp(-g)$ to calling an RGO a small number of times with a small value of $\lam$, each RGO subproblem could be much easier to solve than the original problem. This can happen for a variety of reasons, e.g.\ if the regularized density is extremely well-conditioned, or because it inherits concentration properties of a Gaussian. This idea of reducing a sampling problem to multiple subproblems, each implementing an RGO, underlies the framework of Theorem~\ref{thm:fzerocomp}. Because the idea of regularization by a large Gaussian component repeatedly appears in this paper, we make the following more specific definition for convenience, which lower bounds the size of the Gaussian.

\begin{definition}[$\eta$-RGO]\label{def:etargo}
	We say $\oracle(\lam, v)$ is an $\eta$-restricted Gaussian oracle ($\eta$-RGO) if it satisfies Definition~\ref{def:rgoracle} with the restriction that parameter $\lam$ is required to be always at most $\eta$ in calls to $\oracle$.
\end{definition}

Variants of our notion of an RGO have implicitly appeared previously \cite{CousinsV18, MouFWB19}, and efficient RGO implementation was a key subroutine in the fastest sampling algorithm for general logconcave distributions \cite{CousinsV18}. It also extends a similar oracle used in composite optimization, which we will discuss shortly. However, the explicit use of RGOs in a framework such as Theorem~\ref{thm:fzerocomp} is a novel technical innovation of our work, and we believe this abstraction will find further uses.

\paragraph{Proximal reduction framework.} In Section~\ref{sec:framework}, we prove correctness of our proximal reduction framework, whose guarantees are stated in the following Theorem~\ref{thm:fzerocomp}.

\begin{restatable}{theorem}{restatefzerocomp}\label{thm:fzerocomp}
	Let $\pi$ be a distribution on $\R^d$ with $\tfrac{d\pi}{dx}(x) \propto \exp(-\fcomp(x))$ such that $\fcomp$ is $\mu$-strongly convex, and let $\eps \in (0, 1)$. Let $\eta \leq \frac 1 \mu$, $T = \Theta(\frac{1}{\eta\mu}\log{\frac{d}{\eta\mu\epsilon}})$ for some $\beta \ge 1$, and $\oracle$ be a $\eta$-RGO for $\fcomp$. Algorithm~\ref{alg:alternatesample}, initialized at the minimizer of $\fcomp$, runs in $T$ iterations, each querying $\oracle$ a constant number of times, and obtains $\eps$ total variation distance to $\pi$.
\end{restatable}

In other words, if we can implement an $\eta$-RGO for a $\mu$-strongly convex function $\fcomp$ in time $\mathcal{T}_{\text{RGO}}$, we can sample from $\exp(-\fcomp)$ in time $\widetilde{O}(\tfrac{1}{\eta\mu} \cdot \mathcal{T}_{\text{RGO}})$. To highlight the power of this reduction framework, suppose there was an existing sampler $\alg$ for densities $\propto \exp(-f)$ with mixing time $\widetilde{O}(\kappa^{10}\sqrt{d})$, where $f: \R^d \rightarrow \R$ is $L$-smooth, $\mu$-strongly convex, and has condition number $\kappa = \tfrac{L}{\mu}$ (cf.\ Section~\ref{ssec:notation} for definitions).\footnote{No sampler with mixing time scaling as $\text{poly}(\kappa) \sqrt{d}$ is currently known.} Choosing $\eta = \tfrac{1}{L}$ and $\fcomp \gets f$ in Theorem~\ref{thm:fzerocomp} yields a sampler whose mixing time is $\widetilde{O}(\kappa \cdot \mathcal{T}_{\text{RGO}})$, where $\mathcal{T}_{\text{RGO}}$ is the cost of sampling from a density proportional to
\[\exp\Par{-\frac{L}{2}\norm{x - v}_2^2 - f(x)},\]
for some $v \in \R^d$. However, observe that this distribution has a negative log-density with constant condition number $\tfrac{L+ L}{L + \mu} \le 2$! By using $\alg$ as our RGO, we have $\mathcal{T}_{\text{RGO}} = \widetilde{O}(\sqrt{d})$, and the overall mixing time is $\widetilde{O}(\kappa \sqrt{d})$. Leveraging Theorem~\ref{thm:fzerocomp} in applications, we obtain the following new results, improving mixing of various ``base samplers'' which we bootstrap as RGOs for regularized densities. 

\paragraph{Well-conditioned densities.}

In \cite{LeeST20}, it was shown that a variant of Metropolized Hamiltonian Monte Carlo obtains a mixing time of $\tO(\kappa d \log^3 \tfrac{\kappa d}{\eps})$ for sampling a density on $\R^d$ with condition number $\kappa$. The analysis of \cite{LeeST20} was somewhat delicate, and required reasoning about conditioning on a nonconvex set with desirable concentration properties. In Section~\ref{ssec:kd}, we prove Corollary~\ref{thm:kd}, improving \cite{LeeST20} by roughly a logarithmic factor with a significantly simpler analysis.

\begin{restatable}{corollary}{restatekd}\label{thm:kd}
	Let $\pi$ be a distribution on $\R^d$ with $\tfrac{d\pi}{dx}(x) \propto \exp\Par{-f(x)}$ such that $f$ is $L$-smooth and $\mu$-strongly convex, and let $\eps \in (0, 1)$, $\kappa = \tfrac{L}{\mu}$. Assume access to $x^* = \argmin_{x \in \R^d} f(x)$. Algorithm~\ref{alg:alternatesample} with $\eta = \tfrac{1}{8Ld\log(\kappa)}$ using Algorithm~\ref{alg:xsample} as a restricted Gaussian oracle for $f$ uses $O(\kappa d\log\kappa\log\tfrac{\kappa d}{\eps})$ gradient queries in expectation, and obtains $\eps$ total variation distance to $\pi$.
\end{restatable}

We include Corollary~\ref{thm:kd} as a warmup for our more complicated results, as a way to showcase the use of our reduction framework in a slightly different way than the one outlined earlier. In particular, in proving Corollary~\ref{thm:kd}, we will choose a significantly smaller value of $\eta$, at which point a simple rejection sampling scheme implements each RGO with expected constant gradient queries. 

We give another algorithm matching Corollary~\ref{thm:kd} with a deterministic query complexity bound as Corollary~\ref{corr:zerokd}. The algorithm of Corollary~\ref{corr:zerokd} is interesting in that it is entirely a \emph{zeroth-order} algorithm, and does not require access to a gradient oracle. To our knowledge, in the well-conditioned optimization setting, no zeroth-order query complexities better than roughly $\sqrt{\kappa} d$ are known, e.g.\ simulating accelerated gradient descent with a value oracle; thus, our sampling algorithm has a query bound off by only $\tilde{O}(\sqrt{\kappa})$ from the best-known optimization algorithm. We are hopeful this result may help in the search for query lower bounds for structured logconcave sampling.

\paragraph{Composite densities with a restricted Gaussian oracle.} In Section~\ref{sec:composite}, we develop a sampler for densities on $\R^d$ proportional to $\exp(-f(x) - g(x))$, where $f$ has condition number $\kappa$ and $g$ admits a restricted Gaussian oracle $\oracle$. We state its guarantees here.

\begin{restatable}{theorem}{restatekkdcomp}
	\label{thm:mainclaim}
	Let $\pi$ be a distribution on $\R^d$ with $\tfrac{d\pi}{dx}(x) \propto \exp\Par{-f(x) - g(x)}$ such that $f$ is $L$-smooth and $\mu$-strongly convex, and let $\eps \in (0, 1)$. Let $\eta \le \tfrac{1}{32L\kappa d\log(\kappa/\eps)}$ (where $\kappa = \tfrac{L}{\mu}$), $T = \Theta(\tfrac{1}{\eta\mu}\log(\tfrac{\kappa d}{\eps}))$, and let $\oracle$ be a $\eta$-RGO for $g$. Further, assume access to the minimizer $x^* = \argmin_{x \in \R^d}\{f(x) + g(x)\}$. There is an algorithm which runs in $T$ iterations in expectation, each querying a gradient oracle of $f$ and $\oracle$ a constant number of times, and obtains $\eps$ total variation distance to $\pi$.
\end{restatable}

The assumption that the composite component $g$ admits an RGO can be thought of as a measure of ``simplicity'' of $g$. This mirrors the widespread use of a proximal oracle as a measure of simplicity in the context of composite optimization \cite{BeckT09}, which we now define.

\begin{definition}[Proximal oracle]
	\label{def:proximaloracle}
	$\oracle(\lam, v)$ is a \emph{proximal oracle} for convex $g: \R^d \rightarrow \R$ if it returns
	\[\oracle(\lambda, v) \gets \argmin_{x \in \R^d}\left\{\frac{1}{2\lambda}\norm{x - v}_2^2 + g(x)\right\}.\]
\end{definition}

Many regularizers $g$ in defining composite optimization objectives, which are often used to enforce a quality such as sparsity or ``simplicity'' in a solution, admit efficient proximal oracles. In particular, if the proximal oracle subproblem admits a closed form solution (or otherwise is computable in $O(d)$ time), the regularized objective can be optimized at essentially no asymptotic loss. It is readily apparent that our RGO (Definition~\ref{def:rgoracle}) is the extension of Definition~\ref{def:proximaloracle} to the sampling setting. In \cite{MouFWB19}, a variety of regularizations arising in practical applications including coordinate-separable $g$ (such as restrictions to a coordinate-wise box, e.g.\ for a Bayesian inference task where we have side information on the ranges of parameters) and $\ell_1$ or group Lasso regularized densities were shown to admit RGOs. Our composite sampling results achieve a similar ``no loss'' phenomenon for such regularizations, with respect to existing well-conditioned samplers.

By choosing the largest possible value of $\eta$ in Theorem~\ref{thm:mainclaim}, we obtain an iteration bound of $\tO(\kappa^2 d)$. In Section~\ref{ssec:comp}, we prove Corollary~\ref{thm:kdcomp}, which improves Theorem~\ref{thm:mainclaim} by roughly a $\kappa$ factor.

\begin{restatable}{corollary}{restatekdcomp}\label{thm:kdcomp}
	Let $\pi$ be a distribution on $\R^d$ with $\tfrac{d\pi}{dx}(x) \propto \exp(-f(x) - g(x))$ such that $f$ is $L$-smooth and $\mu$-strongly convex, and let $\eps \in (0, 1)$, $\kappa = \tfrac{L}{\mu}$. Assume access to $x^* = \argmin_{x \in \R^d}\{f(x) + g(x)\}$ and let $\oracle$ be a restricted Gaussian oracle for $g$. There is an algorithm (Algorithm~\ref{alg:alternatesample} using Theorem~\ref{thm:mainclaim} as a restricted Gaussian oracle) which runs in $O(\kappa d\log^3\tfrac{\kappa d}{\eps})$ iterations in expectation, each querying a gradient of $f$ and $\oracle$ a constant number of times, and obtains $\eps$ total variation distance to $\pi$.
\end{restatable}

To sketch the proof, choosing $\eta = \tfrac{1}{L}$ in Theorem~\ref{thm:fzerocomp} yields an algorithm running in $\tO(\tfrac{1}{\eta\mu}) = \tO(\kappa)$ iterations. In each iteration, the RGO subproblem asks to sample from the distribution whose negative log-density is $f(x) + g(x) + \tfrac L 2 \norm{x - v}_2^2$ for some $v \in \R^d$, so we can call Theorem~\ref{thm:mainclaim}, where the ``well-conditioned'' portion $f(x) + \tfrac{L}{2}\norm{x - v}_2^2$ has constant condition number. Thus, Theorem~\ref{thm:mainclaim} runs in $\tO(d)$ iterations to solve the subproblem, yielding the result. In fact, Corollary~\ref{thm:kdcomp} nearly matches Corollary~\ref{thm:kd} in the case $g = 0$ uniformly. Surprisingly, this recovers the runtime of \cite{LeeST20} without appealing to strong gradient concentration bounds (e.g.\ \cite{LeeST20}, Theorem 3.2).

\paragraph{Logconcave finite sums.} In Section~\ref{sec:finitesum}, we initiate the study of mixing times for sampling logconcave finite sums with polylogarithmic dependence on accuracy. We give the following result.

\begin{restatable}{theorem}{restatezerofs}
	\label{thm:base_finitesum}
	Let $\pi$ be a distribution on $\R^d$ with $\tfrac {d\pi}{dx}(x)\propto \exp(-F(x))$, where $F(x)=\tfrac 1 n \sum_{i=1}^n f_i(x)$ is $\mu$-strongly convex, $f_i$ is $L$-smooth and convex $\forall i \in [n]$, $\kappa =\tfrac{L}{\mu}$, and $\eps \in (0,1)$. Assume access to $x^* = \argmin_{x \in \R^d} F(x)$. Algorithm~\ref{alg:mrw} uses $O\left(\kappa^2 d\log^4 \frac {n\kappa d}\eps\right)$ value queries to summands $\{f_i\}_{i \in [n]}$, and obtains $\eps$ total variation distance to $\pi$.
\end{restatable}

For a zeroth-order algorithm, Theorem~\ref{thm:base_finitesum} serves as a surprisingly strong baseline as it nearly matches the previously best-known bound for zeroth-order well-conditioned sampling when $n = 1$; however, when e.g.\ $ \kappa\approx d$, the complexity bound is at least cubic. By using Theorem~\ref{thm:base_finitesum} within the framework of Theorem~\ref{thm:fzerocomp}, we obtain the following improved result.

\begin{restatable}[Improved first-order logconcave finite sum sampling]{corollary}{restatefirstfs}\label{thm:improved_finitesum}
	In the setting of Theorem~\ref{thm:base_finitesum}, Algorithm~\ref{alg:alternatesample} 
	using Algorithm~\ref{alg:mrw} and SVRG \cite{Johnson013} as a restricted Gaussian oracle for $F$ uses 
	\[O\Par{n\log\Par{\frac{n\kappa d}{\eps}} + \kappa \sqrt{nd} \log^{3.5}\Par{\frac{n\kappa d}{\eps}} + \kappa d \log^5\Par{\frac{n\kappa d}{\eps}}} = \tO\Par{n + \kappa\max\Par{d, \sqrt{nd}}}\]
	queries to first-order oracles for summands $\{f_i\}_{i \in [n]}$, and obtains $\eps$ total variation distance to $\pi$.
\end{restatable}

Corollary~\ref{thm:improved_finitesum} has several surprising properties. First, its bound when $n = 1$ gives yet another way of (up to polylogarithmic factors) recovering the runtime of \cite{LeeST20} without gradient concentration. Second, up to a $\tO(\max(1, \sqrt{\tfrac{n}{d}}))$ factor, it is essentially the best runtime one could hope for without an improvement when $n = 1$. This is in the sense that $\tO(\kappa d)$ is the best runtime for $n = 1$, and to our knowledge every efficient well-conditioned sampler requires minimizer access, i.e.\ $\tO(n)$ gradient queries \cite{WoodworthS16}. Interestingly, when $n = 1$, Algorithm~\ref{alg:mrw} can be significantly simplified, and becomes the standard Metropolized random walk \cite{DwivediCW018}; this yields Corollary~\ref{corr:zerokd}, an algorithm attaining the iteration complexity of Corollary~\ref{thm:kd} while only querying a value oracle for $f$.

\subsection{Previous work}
\label{ssec:prev}

Logconcave sampling is a problem that, within the theoretical computer science field, has its origins in convex body sampling (a problem it generalizes). A long sequence of developments have made significant advances in the general model, where only convexity is assumed about the negative log-density, and only value oracle access is given. In this prior work discussion, we focus on more structured cases where all or part of the negative log-density has finite condition number, and refer the reader to \cite{Vempala05,lovasz2006simulated,cousins2015bypassing} for an account on progress in the general case.

\paragraph{Well-conditioned densities.} Significant recent efforts in the machine learning and statistics communities focused on obtaining provable guarantees for well-conditioned distributions, starting from pioneering work of \cite{dalalyan2017theoretical}, and continued in e.g.\ \cite{ChengCBJ18, DalalyanR18, ChenV19, ChenDWY19, DwivediCW018, DurmusM19, DurmusMM19, LeeSV18, MouMWBJ19, shen2019randomized, LeeST20}. In this setting, many methods based on discretizations of continuous-time first-order processes (such as the Langevin dynamics) have been proposed. Typically, error guarantees come in two forms: either in the $2$-Wasserstein ($W_2$) distance, or in total variation (TV). The line \cite{DwivediCW018, ChenDWY19, LeeST20} has brought the gradient complexity for obtaining $\eps$ TV distance to $\tO(\kappa d)$ where $d$ is the dimension, by exploiting gradient concentration properties. For progress in complexities depending polynomially on $\eps^{-1}$, attaining $W_2$ guarantees (typically incomparable to TV bounds), we defer to \cite{shen2019randomized}, the state-of-the-art using $\tO(\kappa^{\frac{7}{6}}\eps^{-\frac{1}{3}} + \kappa\eps^{-\frac{2}{3}})$ queries to obtain $W_2$ distance $\eps \sqrt{d\mu^{-1}}$ from the target.\footnote{Here, $\sqrt{d\mu^{-1}}$ is the effective diameter; this accuracy measure allows for scale-invariant $W_2$ guarantees.} We note incomparable guarantees can be obtained by assuming higher derivative bounds (e.g.\ a Lipschitz Hessian); our work uses only the minimal assumption of bounded second derivatives.

\paragraph{Composite densities.} Recent works have studied sampling from densities of the form \eqref{eq:compositesampling}, or its specializations (e.g.\ restrictions to a convex set). Several \cite{Pereyra16, BrosseDMP17, Bernton18} are based on Moreau envelope or proximal regularization strategies, and demonstrate efficiency under more stringent assumptions on the structure of the composite function $g$, but under minimal assumptions obtain fairly large provable mixing times $\Omega(d^5)$. Proximal regularization algorithms have also been considered for non-composite sampling \cite{Wibisono19}. Another discretization strategy based on projections was studied by \cite{BubeckEL18}, but obtained mixing time $\Omega(d^7)$. Finally, improved algorithms for special constrained sampling problems have been proposed, such as simplex restrictions \cite{HsiehKRC18}.

Of particular relevance and inspiration to this work is \cite{MouFWB19}. By generalizing and adapting Metropolized HMC algorithms of \cite{DwivediCW018, ChenDWY19}, adopting a Moreau envelope strategy, and using (a stronger version of) the RGO access model, \cite{MouFWB19} obtained a runtime which in the best case scales as $\tOh{\kappa^2 d}$, similar to the guarantee of our base sampler in Theorem~\ref{thm:mainclaim}. However, this result required a variety of additional assumptions, such as access to the normalization factor of restricted Gaussians, Lipschitzness of $g$, warmness of the start, and various problem parameter tradeoffs. The general problem of sampling from \eqref{eq:compositesampling} under minimal assumptions more efficiently than general-purpose logconcave algorithms is to the best of our knowledge unresolved (even under restricted Gaussian oracle access), a novel contribution of our mixing time bound. Our results also suggest that the RGO is a natural notion of tractability for the composite sampling problem.

\paragraph{Logconcave finite sums.} Since \cite{welling2011bayesian} proposed the stochastic gradient Langevin dynamics, which at each step stochastically estimates the full gradient of the function, there has been a long line of work giving bounds for this method and other similar algorithms \cite{dalalyan2019user,gao2018global,salim2019stochastic,barkhagen2018stochastic,nemeth2019stochastic}. These convergence rates depend heavily on the variance of the stochastic estimates. Inspired by variance-reduced methods in convex optimization, samplers based on low-variance estimators have also been proposed  \cite{dubey2016variance,durmus2016stochastic,bierkens2019zig,baker2019control,nagapetyan2017true,chen2017convergence,zou2018subsampled,chatterji2018theory}. Although our reduction-based approach is not designed specifically for solving problems of finite sum structure, our speedup can be viewed as due to a lower variance estimator implicitly defined through the oracle subproblems of Theorem~\ref{thm:fzerocomp} via repeated re-centering.

In Table \ref{table:runtime}, we list prior runtimes \cite{zou2018subsampled,chatterji2018theory} for sampling logconcave finite sums; note these results additionally require bounded higher derivatives (with the exception of the $\kappa^4$ dependence), obtain guarantees only in Wasserstein distance, and have polynomial dependences on $\eps^{-1}$. On the other hand, our reduction-based approach obtains total variation bounds with linear dependence on $\kappa$ and polylogarithmic dependence on $\eps^{-1}$. Our bound also simultaneously matches the state-of-the-art bound when $n = 1$, a feature not shared by prior stochastic algorithms. To our knowledge, no previous nontrivial\footnote{As mentioned previously, one can always compute the full $\nabla F$ in every iteration in a well-conditioned sampler.} bounds were known in the high-accuracy regime before our work.

\begin{table}
	\begin{centering}
		\begin{tabular}{|c|c|c|}
			\hline 
			\multirow{2}{*}{Method} & \multicolumn{2}{c|}{Gradient oracle complexity}\tabularnewline
			\cline{2-3} \cline{3-3} 
			& $W_{2}\leq\epsilon$, $\mu = 1$ & $W_{2}\leq\epsilon\sqrt{d\mu^{-1}}$\tabularnewline
			\hline 
			SAGA-LD \cite{chatterji2018theory} & $n+\frac{\kappa^{1.5}\sqrt{d}+\kappa d+Md}{\epsilon}+\frac{\kappa d^{4/3}}{\epsilon^{2/3}}$ & $n+\frac{\kappa^{1.5}+\kappa\sqrt{d}+M\sqrt{d}}{\epsilon}+\frac{\kappa d^{2/3}}{\epsilon^{2/3}}$\tabularnewline
			\hline 
			SVRG-LD \cite{chatterji2018theory}& $n + \frac{\kappa^{1.5}\sqrt{d} + \kappa d + Md}{\eps} + \frac{\kappa d^{4/3}}{\eps^{2/3}}$ & $n+\frac{\kappa^{3}}{\epsilon^{2}}+\frac{\kappa^{1.5}+M\sqrt{d}}{\epsilon}$\tabularnewline
			\hline 
			CV-ULD \cite{chatterji2018theory}& $n+\frac{\kappa^{4}d^{1.5}}{\epsilon^{3}}$ & $n+\frac{\kappa^{4}}{\epsilon^{3}}$\tabularnewline
			\hline 
			SVRG-LD \cite{zou2018subsampled} & $n+\frac{\kappa^{1.5}\sqrt{d}+Md}{\epsilon}+\frac{\kappa\sqrt{nd}}{\epsilon}$ & $n+\frac{\kappa^{1.5}+M\sqrt{d}}{\epsilon}+\frac{\kappa\sqrt{n}}{\epsilon}$\tabularnewline
			\hline 
			State-of-the-art, $n = 1$ \cite{shen2019randomized} & $\frac{\kappa^{7/6}d^{1/6}}{\epsilon^{1/3}}+\frac{\kappa d^{1/3}}{\epsilon^{2/3}}$ & $\frac{\kappa^{7/6}}{\epsilon^{1/3}}+\frac{\kappa}{\epsilon^{2/3}}$\tabularnewline
			\hline 
		\end{tabular}
		\par\end{centering}
	\vspace{0.3cm}
	\begin{centering}
		\begin{tabular}{|c|c|}
			\hline 
			\multirow{1}{*}{Method} & \multicolumn{1}{c|}{Gradient oracle complexity (TV $\leq\epsilon$)}\tabularnewline
			\hline 
			Corollary~\ref{thm:improved_finitesum} & $n + \kappa d + \kappa\sqrt{nd}$\tabularnewline
			\hline 
			State-of-the-art, $n=1$ \cite{LeeST20} & $\kappa d$\tabularnewline
			\hline 
		\end{tabular}
		\par\end{centering}
	\caption{\label{table:runtime}Complexity of sampling from $e^{-F(x)}$ where $F(x)=\frac{1}{n}\sum_{i\in[n]}f_{i}(x)$ on $\R^d$ is $\mu$-strongly convex, each $f_{i}$ is convex and $L$-smooth, and $\kappa = \tfrac{L}{\mu}$. For relevant lines, $M$ is the Lipschitz constant of the Hessian $\nabla^2 F$, which our algorithm has no dependence on. Complexity is measured in terms of the number of calls to $f_i$ or $\nabla f_i$ for summands $\{f_i\}_{i \in [n]}$. We hide $\text{polylog}(\tfrac{n\kappa d}{\eps})$ factors for simplicity.}
\end{table}

\paragraph{Preliminary version.} A preliminary version of this work, containing the results of Section~\ref{sec:composite}, appeared as \cite{ShenTL20}. The preliminary version also contained an experimental evaluation of the algorithm in Section~\ref{sec:composite} for the task of sampling a (non-diagonal covariance) multivariate Gaussian restricted to a box, and demonstrated the efficacy of our method in comparison to general-purpose logconcave samplers (i.e.\ the hit-and-run method \cite{LovaszV06a}). The focus of the present version is giving theoretical guarantees for structured logconcave sampling tasks, so we omit empirical evaluations, and defer an evaluation of the new methods developed in this paper to interesting follow-up work.

\subsection{Technical overview}\label{ssec:technical}

\subsubsection{Composite logconcave sampling}\label{ssec:compositeintro}

We study the problem of approximately sampling from a distribution $\pi$ on $\R^d$, with density 
\begin{equation}\label{eq:compositesampling}\frac{d\pi(x)}{dx} \propto \exp\left(-f(x) - g(x)\right).\end{equation}
Here, $f: \R^d \rightarrow \R$ is assumed to be ``well-behaved'' (i.e.\ has finite condition number), and $g: \R^d \rightarrow \R$ is a convex, but possibly non-smooth function. This problem generalizes the special case of sampling from $\exp(-f(x))$ for well-conditioned $f$, simply by letting $g$ vanish. Even the specialization of \eqref{eq:compositesampling} where $g$ indicates a convex set (i.e.\ is $0$ inside the set, and $\infty$ outside) is not well-understood; existing mixing time bounds for this restricted setting are large polynomials in $d$ \cite{BrosseDMP17, BubeckEL18}, and are typically weaker than guarantees in the general logconcave setting \cite{LovaszV06a, LovaszV06b}. This is in contrast to the convex optimization setting, where first-order methods readily generalize to solve problem families such as $\min_{x \in \xset} f(x)$, where $\xset \subseteq \R^d$ is a convex set, as well as its generalization
\begin{equation}\label{eq:compositeopt}\min_{x \in \R^d} f(x) + g(x),\text{ where } g: \R^d \rightarrow \R \text{ is convex and admits a proximal oracle.}\end{equation}
We defined proximal oracles in Definition~\ref{def:proximaloracle}; in short, they are prodecures which minimize the sum of a quadratic and $g$. Definition~\ref{def:proximaloracle} is desirable as many natural non-smooth composite objectives arising in learning settings, such as the Lasso \cite{Tibshirani96} and elastic net \cite{ZouH05}, admit efficient proximal oracles. It is clear that the definition of a proximal oracle implies it can also handle arbitrary sums of linear functions and quadratics, as the resulting function can be rewritten as the sum of a constant and a single quadratic. The seminal work \cite{BeckT09} extends fast gradient methods to solve \eqref{eq:compositeopt} via proximal oracles, and has prompted many follow-up studies. 

Motivated by the success of the proximal oracle framework in convex optimization, we study sampling from the family \eqref{eq:compositesampling} through the lens of RGOs, a natural extension of Definition~\ref{def:proximaloracle}. The main result of Section~\ref{sec:composite} is a ``base'' algorithm efficiently sampling from \eqref{eq:compositesampling}, assuming access to an RGO for $g$. We now survey the main components of this algorithm. 

\textbf{Reduction to shared minimizers.} We first observe that without loss of generality, $f$ and $g$ share a minimizer: by shifting $f$ and $g$ by linear terms, i.e.\ $\tilde{f}(x) \defeq f(x) - \inprod{\nabla f(x^*)}{x}$, $\tilde{g}(x) \defeq g(x) + \inprod{\nabla f(x^*)}{x}$, where $x^*$ minimizes $f + g$, first-order optimality implies both $\tilde{f}$ and $\tilde{g}$ are minimized by $x^*$. Moreover, implementation of a first-order oracle for $\tilde{f}$ and an RGO for $\tilde{g}$ are immediate without additional assumptions. This modification becomes crucial for our later developments, and we hope this simple observation, reminiscent of ``variance reduction'' techniques in stochastic optimization \cite{Johnson013}, is broadly applicable to improving algorithms for the sampling problem induced by \eqref{eq:compositesampling}.

\textbf{Beyond Moreau envelopes: expanding the space.} A typical approach in convex optimization in handling non-smooth objectives $g$ is to instead optimize its \emph{Moreau envelope}, defined by
\begin{equation}\label{eq:moreauenvelope}g^\eta(y) \defeq \min_{x \in \R^d} \left\{g(x) + \frac{1}{2\eta}\norm{x - y}_2^2\right\}.\end{equation}
Intuitively, the envelope $g^\eta$ trades off function value with proximity to $y$; a standard exercise shows that $g^\eta$ is smooth (has a Lipschitz gradient), with smoothness depending on $\eta$, and moreover that computing gradients of $g^\eta$ reduces to calling a proximal oracle (Definition~\ref{def:proximaloracle}). It is natural to extend this idea to the composite sampling setting, e.g.\ via sampling from the density
\[\exp\left(-f(x) - g^\eta(x)\right).\]
However, a variety of complications prevent such strategies from obtaining rates comparable to their non-composite, well-conditioned counterparts, including difficulty in bounding closeness of the resulting distribution, as well as biased drifts of the sampling process due to error in gradients.

Our approach departs from this smoothing strategy in a crucial way, inspired by Hamiltonian Monte Carlo (HMC) methods \cite{Kramers40, Neal11}. HMC can be seen as a discretization of the ubiquitous Langevin dynamics, on an expanded space. In particular, discretizations of Langevin dynamics simulate the stochastic differential equation $\tfrac{dx_t}{dt} = -\nabla f(x_t) + \sqrt{2}\tfrac{dW_t}{dt}$, where $W_t$ is Brownian motion. HMC methods instead simulate dynamics on an extended space $\R^d \times \R^d$, via an auxiliary ``velocity'' variable which accumulates gradient information. This is sometimes interpreted as a discretization of the underdamped Langevin dynamics \cite{ChengCBJ18}. HMC often has desirable stability properties, and expanding the dimension via an auxiliary variable has been used in algorithms obtaining the fastest rates in the well-conditioned logconcave sampling regime \cite{shen2019randomized, LeeST20}. Inspired by this phenomenon, we consider the density on $\R^d \times \R^d$
\begin{equation}\label{eq:expandspace}\frac{d\pih}{dz}(z) \defeq \exp\left(-f(y) - g(x) - \frac{1}{2\eta}\norm{x - y}_2^2\right) \text{ where } z = (x, y).\end{equation}
Due to technical reasons, the family of distributions we use in our final algorithms are of slightly different form than \eqref{eq:expandspace}, but this simplification is useful to build intuition. Note in particular that the form of \eqref{eq:expandspace} is directly inspired by \eqref{eq:moreauenvelope}, where rather than maximizing over $x$, we directly expand the space. The idea is that for small enough $\eta$ and a set on $x$ of large measure, smoothness of $f$ will guarantee that the marginal of \eqref{eq:expandspace} on $x$ will concentrate $y$ near $x$, a fact we make rigorous. To sample from \eqref{eq:compositesampling}, we then show that a rejection filter applied to a sample $x$ from the marginal of \eqref{eq:expandspace} will terminate in constant steps. Consequently, it suffices to develop a fast sampler for \eqref{eq:expandspace}.

\textbf{Alternating sampling with an oracle.} The form of the distribution \eqref{eq:expandspace} suggests a natural strategy for sampling from it: starting from a current state $(x_k, y_k)$, we iterate
\begin{enumerate}
	\item Sample $y_{k + 1} \sim \exp\left(-f(y) - \tfrac{1}{2\eta}\norm{x_k - y}_2^2\right)$.
	\item Sample $x_{k + 1} \sim \exp\left(-g(x) - \frac{1}{2\eta}\norm{x - y_{k + 1}}_2^2\right)$, via a restricted Gaussian oracle.
\end{enumerate}
When $f$ and $g$ share a minimizer, taking a first-order approximation in the first step, i.e.\ sampling $y_{k + 1} \sim \exp(-f(x_k) - \inprod{\nabla f(x_k)}{y - x_k} - \tfrac{1}{2\eta}\norm{y - x_k}_2^2)$, can be shown to generalize the $\texttt{Leapfrog}$ step of HMC updates. However, for $\eta$ very small (as in our setting), we observe the first step itself reduces to the case of sampling from a distribution with constant condition number, performable in $\tilde{O}(d)$ gradient calls by e.g.\ Metropolized HMC \cite{DwivediCW018, ChenDWY19, LeeST20}. Moreover, it is not hard to see that this ``alternating marginal'' sampling strategy preserves the stationary distribution exactly, so no filtering is necessary. Directly bounding the conductance of this random walk, for small enough $\eta$, leads to an algorithm running in $\tOh{\kappa^2 d^2}$ iterations, each calling an RGO once, and a gradient oracle for $f$ roughly $\tOh{d}$ times. This latter guarantee is by an appeal to known bounds \cite{ChenDWY19, LeeST20} on the mixing time in high dimensions of Metropolized HMC for a well-conditioned distribution, a property satisfied by the $y$-marginal of \eqref{eq:expandspace} for small $\eta$. 

\textbf{Stability of Gaussians under bounded perturbations.} To obtain our tightest runtime result, we use that $\eta$ is chosen to be much smaller than $L^{-1}$ to show structural results about distributions of the form \eqref{eq:expandspace}, yielding tighter concentration for bounded perturbations of a Gaussian (i.e.\ the Gaussian has covariance $\tfrac{1}{\eta}\id$, and is restricted by $L$-smooth $f$ for $\eta \ll L^{-1}$). To illustrate, let
\[\frac{d\prop_x(y)}{dy} \propto \exp\left(-f(y) - \frac{1}{2\eta}\norm{y - x}_2^2\right)\]
and let its mean and mode be $\by_x$, $y^*_x$. It is standard that $\norm{\by_x - y^*_x}_2 \le \sqrt{d\eta}$, by $\eta^{-1}$-strong logconcavity of $\prop_x$. Informally, we show that for $\eta \ll L^{-1}$ and $x$ not too far from the minimizer of $f$, we can improve this to $\norm{\by_x - y^*_x}_2 = O(\sqrt{\eta})$; see Proposition~\ref{prop:min_perturb} for a precise statement. 

Using our structural results, we sharpen conductance bounds, improve the warmness of a starting distribution, and develop a simple rejection sampling scheme for sampling the $y$ variable in expected constant gradient queries. Our proofs are continuous in flavor and based on gradually perturbing the Gaussian and solving a differential inequality; we believe they may of independent interest. These improvements lead to an algorithm running in $\tOh{\kappa^2d}$ iterations; ultimately, we use our reduction framework, stated in Theorem~\ref{thm:fzerocomp}, to improve this dependence to $\tOh{\kappa d}$.

\subsubsection{Logconcave finite sums}\label{ssec:finitesumintro}

We initiate the algorithmic study of the following task in the high-accuracy regime: sample $x \sim \pi$ within total variation distance $\eps$, where $\tfrac{d\pi}{dx}(x) \propto \exp(-F(x))$ and 
\begin{equation}\label{eq:finitesum} F(x) = \frac{1}{n}\sum_{i \in [n]} f_i(x),\end{equation}
all $f_i: \R^d \rightarrow \R$ are convex and $L$-smooth, and $F$ is $\mu$-strongly convex. We call such a distribution $\pi$ a (well-conditioned) \emph{logconcave finite sum}. 

In applications (where summands correspond to points in a dataset, e.g.\ in Bayesian linear and logistic regression tasks \cite{DwivediCW018}), querying $\nabla F$ may be prohibitively expensive, so a natural goal is to obtain bounds on the number of required queries to summands $\nabla f_i$ for $i \in [n]$. This motivation also underlies the development of stochastic gradient methods in optimization, a foundational tool in modern statistics and data processing. Na\"ively, one can complete the task by using existing samplers for well-conditioned distributions and querying the full gradient $\nabla F$ in each iteration, resulting in a summand gradient query complexity of $\tO(n\kappa d)$ \cite{LeeST20}. Many recent works, inspired from recent developments in the complexity of optimizing a well-conditioned finite sum, have developed subsampled gradient methods for the sampling problem. However, to our knowledge, all such guarantees depend polynomially on the accuracy $\eps$ and are measured in the $2$-Wasserstein distance; in the high-accuracy, total variation case no nontrivial query complexity is currently known.

We show in Section~\ref{sec:finitesum} that given access to the minimizer of $F$, a simple zeroth-order algorithm which queries $\tO(\kappa^2 d)$ values of summands $\{f_i\}_{i \in [n]}$ succeeds (i.e.\ it never requires a full value or gradient query of $F$). The algorithm is essentially the Metropolized random walk proposed in \cite{DwivediCW018} for the $n = 1$ case with a cheaper subsampled filter step. Notably, because the random walk is conducted with respect to $F$, we cannot efficiently query the function value at any point; nonetheless, by randomly sampling to compute a nearly-unbiased estimator of the rejection probability, we do not incur too much error. This random walk was shown in \cite{ChenDWY19} to mix in $\tO(\kappa^2 d)$ iterations; we implement each step to sufficient accuracy using $\tO(1)$ function evaluations.

It is natural to ask if first-order information can be used to improve this query complexity, perhaps through ``variance reduction'' techniques (e.g.\ \cite{Johnson013}) developed for stochastic optimization. The idea behind variance reduction is to recenter gradient estimates in phases, occasionally computing full gradients to improve the estimate quality. One fundamental difficulty which arises from using variance reduction in high-accuracy sampling is that the resulting algorithms are not \emph{stateless}. By this, we mean that the variance-reduced estimates depend on the history of the algorithm, and thus it is difficult to ascertain correctness of the stationary distribution. We take a different approach to achieve a linear query dependence on the conditioning $\kappa$, described in the following section.

\subsubsection{Proximal point reduction framework}\label{ssec:overview}

To motivate Theorem~\ref{thm:fzerocomp}, we first recast existing ``proximal point'' reduction-based optimization methods through the lens of composite optimization, and subsequently show that similar ideas underlying our composite sampler in Section~\ref{ssec:compositeintro} yield an analagous ``proximal point reduction framework'' for sampling. Chronologically, our composite sampler (originally announced in \cite{ShenTL20}) predates our reduction framework, which was then inspired by the perspective given here. We hope these insights prove fruitful for further development of proximal approaches to sampling.

\paragraph{Proximal point methods as composite optimization.} Proximal point methods are a well-studied primitive in optimization, developed by \cite{Rockafellar76}; cf.\ \cite{ParikhB14} for a modern survey. The principal idea is that to minimize convex $F:\R^d \rightarrow \R$, it suffices to solve a sequence of subproblems
\begin{equation}\label{eq:proxpoint}x_{k + 1} \gets \argmin_{x \in \R^d}\Brace{F(x) + \frac{1}{2\lam}\norm{x - x_k}_2^2}.\end{equation}
Intuitively, by tuning the parameter $\lam \ge 0$, we trade off how regularized the subproblems \eqref{eq:proxpoint} are with how rapidly the overall method converges. Smaller values of $\lam$ result in larger regularization amounts which are amenable to algorithms for minimizing well-conditioned objectives. 

For optimizing functions of the form \eqref{eq:finitesum} via stochastic gradient estimates to $\eps$ error, \cite{Johnson013, DefazioBL14, SchmidtRB17} developed variance-reduced methods obtaining a query complexity of $\tO(n + \kappa)$. To match a known lower bound of $\tO(n + \sqrt{n\kappa})$ due to \cite{WoodworthS16}, two works \cite{LinMH15, FrostigGKS15} appropriately applied instances of accelerated proximal point methods \cite{Guler92} with careful analyses of how accurately subproblems \eqref{eq:proxpoint} needed to be solved. These algorithms black-box called the $\tO(n + \kappa)$ runtime as an oracle to solve the subproblems \eqref{eq:proxpoint} for an appropriate choice of $\lam$, obtaining an accelerated rate.\footnote{We note that an improved runtime without extraneous logarithmic factors was later obtained by \cite{Allen-Zhu17}.} To shed some light on this acceleration procedure, we adopt an alternative view on proximal point methods.\footnote{This perspective can also be found in the lecture notes \cite{Lee18}.} Consider the following known composite optimization result.

\begin{proposition}[Informal statement of \cite{BeckT09}]\label{prop:beckt}
Let $f: \R^d \rightarrow \R$ be $L$-smooth and $\mu$-strongly convex, and $g: \R^d \rightarrow \R$ admit a proximal oracle $\oracle(\lam, v)$ (cf.\ Definition~\ref{def:proximaloracle}). There is an algorithm taking $\tO(\sqrt{\kappa})$ iterations for $\kappa = \tfrac{L}{\mu}$ to find an $\eps$-approximate minimizer to $f + g$, each querying $\nabla f$ and $\oracle$ a constant number of times. Further, $\lam = \tfrac{1}{L}$ in all calls to $\oracle$.
\end{proposition}

Ignoring subtleties of the error tolerance of $\oracle$, we show how to use an instance of Proposition~\ref{prop:beckt} to recover the $\tO(n + \sqrt{n\kappa})$ query complexity for optimizing \eqref{eq:finitesum}. Let $f(x) = \tfrac{\mu}{2}\norm{x}_2^2$, and $g = F - f$. For any $\Lambda \ge \mu$, $f$ is both $\mu$-strongly convex and $\Lambda$-smooth. Moreover, note that all calls to the proximal oracle $\oracle$ for $g$ require solving subproblems of the form
\begin{equation}\label{eq:regsubprob}\argmin_{x \in \R^d}\Brace{F(x) - \frac{\mu}{2}\norm{x}_2^2 + \frac{\Lambda}{2}\norm{x - v}_2^2}.\end{equation}
The upshot of choosing a smoothness bound $\Lambda \ge \mu$ is that the regularization amount in \eqref{eq:regsubprob} increases, improving the conditioning of the subproblem, which is $\Lambda$-strongly convex and $L + \Lambda$-smooth. The algorithm of e.g.\ \cite{Johnson013} solves each subproblem \eqref{eq:regsubprob} in $\tO(n + \tfrac{L + \Lambda}{\Lambda})$ gradient queries, leading to an overall query complexity (for Proposition~\ref{prop:beckt}) of
\[\tO\Par{\sqrt{\frac{\Lambda}{\mu}} \cdot \Par{n + \frac{L}{\Lambda}}}.\]
Optimizing over $\Lambda \ge \mu$, i.e.\ taking $\Lambda = \max(\mu, \tfrac{L}{n})$, yields the desired bound of $\tO(n + \sqrt{n\kappa})$.

\paragraph{Applications to sampling.} In Sections~\ref{sec:composite} and~\ref{sec:finitesum}, we develop samplers for structured families with quadratic dependence on the conditioning $\kappa$. The proximal point approach suggests a strategy for accelerating these runtimes. Namely, if there is a framework which repeatedly calls a sampler for a regularized density (analogous to calls to \eqref{eq:proxpoint}), one could trade off the regularization with the rate of the outer loop. Fortunately, in the spirit of interpreting proximal point methods as composite optimization, the composite sampler of Section~\ref{sec:composite} itself meets these reduction framework criteria. 

We briefly recall properties of our composite sampler here. Let $\pi$ be a distribution on $\R^d$ with $\tfrac{d\pi}{dx}(x) \propto \exp(-\fwc(x) - \fcomp(x))$,\footnote{To disambiguate, we sometimes also use the notation $\fwc + \fcomp$ rather than $f + g$ in defining instances of our reduction framework or composite samplers, when convenient in the context.} where $\fwc$ is well-conditioned (has finite condition number $\kappa$) and $\fcomp$ admits an RGO, which solves subproblems of the form
\begin{equation}\label{eq:rgosubp}\oracle(\eta, v) \sim \text{the density proportional to } \exp\Par{-\frac{1}{2\eta}\norm{x - v}_2^2 - \fcomp(x)}.\end{equation}
The algorithm of Section~\ref{sec:composite} only calls $\oracle$ with a fixed $\eta$; note the strong parallel between the RGO subproblem and the proximal oracle of Proposition~\ref{prop:beckt}. For a given value of $\eta \ge 0$, our composite sampler runs in $\tO(\tfrac{1}{\eta\mu})$ iterations, each requiring a call to $\oracle$. Smaller $\eta$ improve the conditioning of the negative log-density of subproblem \eqref{eq:rgosubp}, but increase the overall iteration count, yielding a range of trade-offs. The algorithm of Section~\ref{sec:composite} has an upper bound requirement on $\eta$ (cf.\ Theorem~\ref{thm:mainclaim}); in Section~\ref{sec:framework}, we observe that this may be lifted when $\fwc = 0$ uniformly, allowing for a full range of choices. Moreover, the analysis of the composite sampler becomes much simpler when $\fwc = 0$, as in Theorem~\ref{thm:fzerocomp}. We give the framework as Algorithm~\ref{alg:alternatesample}, as well as a full (fairly short) convergence analysis. By trading off the regularization amount with the cost of implementing \eqref{eq:rgosubp} via bootstrapping base samplers, we obtain a host of improved runtimes.

Beyond our specific applications, the framework we provide has strong implications as a generic reduction from mixing times scaling polynomially in $\kappa$ to improved methods scaling linearly in $\kappa$. This is akin to the observation in \cite{LinMH15} that accelerated proximal point methods generically improve $\text{poly}(\kappa)$ dependences to $\sqrt{\kappa}$ dependences for optimization. We are optimistic this insight will find further implications in the logconcave sampling literature.

\subsection{Erratum, and a word of warning for $o(d)$ mixing}\label{ssec:bug}

The initial version of this paper, presented at COLT 2021, had an incorrect proof of Theorem~\ref{thm:fzerocomp}. This was due to our reliance on the average conductance (``spectral profile'') technique of \cite{LovaszK99} for bounding mixing. Roughly speaking, the mistake was caused by a misunderstanding that for stationary measures satisfying $\mu$-log isoperimetry (for example, $\mu$-strongly logconcave densities) and with transition distributions of $\Delta$-close points having constant overlap, \cite{LovaszK99} provides mixing time bounds of the form (where $\beta$ is a warmness parameter of the starting distribution)
\begin{equation}\label{eq:avcond}\int_{\frac 1 \beta}^{\frac 1 2}\frac{1}{s \Phi(s)^2} ds \lesssim \frac{1}{\mu\Delta^2}\int_{\frac 1 \beta}^{\frac 1 2}\frac{1}{s\log(s)} ds \approx \frac{1}{\mu \Delta^2} \log\log \beta .\end{equation}
Here, $\Phi(s)$ is the $s$-conductance of the Markov chain, which can typically be lower bounded by $\Omega(\sqrt{\mu \log(s)} \Delta)$ under a stationary density exhibiting log-isoperimetry. However, the trivial bound $\Phi(s) \le 1$  demonstrates that there is an additive $\log(\beta)$ term in \eqref{eq:avcond}. This is a bottleneck towards mixing times scaling as $o(d)$ for distributions where only an $\exp(\Omega(d))$-warm start is feasible; in particular, the conductance actually scales as $\min(1, \Omega(\sqrt{\mu \log(s)} \Delta))$, causing this additive term. In settings where $\mu\Delta^2 \ge d^{-1}$ (such as our reductions, where this term often scales as a condition number of the problem), this additive term $\log(\beta) = \Omega(d)$ may dominate. This observation (and the fix) came out of conversations with Sinho Chewi; we are immensely greatful for his help.

For the particular structure of the algorithm in Theorem~\ref{thm:fzerocomp}, we are able to give an alternative analysis going through $W_2$ convergence bounds, preserving the correctness of the theorem. However, this bottleneck is a general phenomenon which may cause future attempts to use Metropolized algorithms from exponentially warm starts to be stuck at $\Omega(d)$ iterations, which merits further investigation. We write this section as a word of warning to future researchers aiming at sublinear dimension dependences in Metropolized algorithms, and as a suggested open research direction.

\subsection{Roadmap}\label{ssec:roadmap}

We give notations and preliminaries in Section~\ref{sec:prelims}. In Section~\ref{sec:framework} we give our framework for bootstrapping fast regularized samplers, and prove its correctness (Theorem~\ref{thm:fzerocomp}). Assuming the ``base samplers'' of Theorems~\ref{thm:mainclaim} and~\ref{thm:base_finitesum}, in Section~\ref{sec:improve} we apply our reduction to obtain all of our strongest guarantees, namely Corollaries~\ref{thm:kd},~\ref{thm:kdcomp}, and~\ref{thm:improved_finitesum}. We then prove Theorems~\ref{thm:mainclaim} and~\ref{thm:base_finitesum} in Sections~\ref{sec:composite} and~\ref{sec:finitesum}. 

\section{Preliminaries}
\label{sec:prelims}

\subsection{Notation}\label{ssec:notation}

\paragraph{General notation.} For $d \in \N$, $[d]$ refers to the set of naturals $1 \le i \le d$; $\norm{\cdot}_2$ is the Euclidean norm on $\R^d$ when $d$ is clear from context. $\Nor(\mu, \covar)$ is the multivariate Gaussian of specified mean and variance, $\id$ is the identity of appropriate dimension when clear from context, and $\preceq$ is the Loewner order on symmetric matrices.

\paragraph{Functions.} We say twice-differentiable function $f: \R^d \rightarrow \R$ is $L$-smooth and $\mu$-strongly convex if $\mu\id \preceq \nabla^2 f(x) \preceq L\id$ for all $x \in \R^d$; it is well-known that $L$-smoothness implies that $f$ has an $L$-Lipschitz gradient, and that for any $x, y \in \R^d$,
\[f(x) + \inprod{\nabla f(x)}{y - x} + \frac{\mu}{2}\norm{y - x}_2^2 \le f(y) \le f(x) + \inprod{\nabla f(x)}{y - x} + \frac{L}{2}\norm{y - x}_2^2.\]
If $f$ is $L$-smooth and $\mu$-strongly convex, we say it has a condition number $\kappa \defeq \tfrac{L}{\mu}$. We call a zeroth-order oracle, or ``value oracle'', an oracle which returns $f(x)$ on any input point $x \in \R^d$; similarly, a first-order oracle, or ``gradient oracle'', returns both the value and gradient $(f(x), \nabla f(x))$.

\paragraph{Distributions.} We call distribution $\pi$ on $\R^d$ logconcave if $\tfrac{d\pi}{dx}(x) = \exp(-f(x))$ for convex $f$; $\pi$ is $\mu$-strongly logconcave if $f$ is $\mu$-strongly convex. For $A \subseteq \R^d$, $A^c$ is its complement, and we let $\pi(A) \defeq \int_{x \in A} d\pi(x)$. We say distribution $\rho$ is $\beta$-warm with respect to $\pi$ if $\tfrac{d\pi}{d\rho}(x) \le \beta$ everywhere, and define the total variation $\tvd{\pi}{\rho}\defeq \sup_{A \subseteq \R^d} \pi(A) - \rho(A)$. We will frequently use the fact that $\tvd{\pi}{\rho}$ is also the probability that $x \sim \pi$ and $x' \sim \rho$ are unequal under the best coupling of $(\pi, \rho)$; this allows us to ``locally share randomness'' when comparing two random walk procedures. We define the expectation $\E_\pi$ and variance $\Var_\pi$ with respect to distribution $\pi$ in the standard way,
\[\E_\pi[h(x)] \defeq \int h(x) d\pi(x),\; \Var_\pi[h(x)] \defeq \E_\pi\left[(h(x))^2\right] - \left(\E_\pi[h(x)]\right)^2.\]

\paragraph{Structured distributions.} This work considers two types of distributions with additional structure, which we call \emph{composite logconcave densities} and \emph{logconcave finite sums}. A composite logconcave density has the form $\exp(-f(x) -g(x))$, where both $f$ and $g$ are convex. In this context throughout, $f$ will either be uniformly $0$ or have a finite condition number (be ``well-conditioned''), and $g$ will represent a ``simple'' but possibly non-smooth function, as measured by admitting an RGO (cf.\ Definition~\ref{def:rgoracle}). We will sometimes refer to the components as $f$ and $g$ as $\fwc$ and $\fcomp$ respectively, to disambiguate when the functions $f$ and $g$ are already defined in context. In our reduction-based approaches, we have additional structure on the parameter $\lam$ which an RGO is called with (cf.\ Definition~\ref{def:etargo}). Specifically, in our instances typically $\lam^{-1} \gg L$ (or some other ``niceness'' parameter associated with the negative log-density); this can be seen as heavily regularizing the negative log-density, and often makes the implementation simpler.

Finally, a logconcave finite sum has density of the form $\exp(-F(x))$ where $F(x) = \tfrac{1}{n}\sum_{i \in [n]} f_i(x)$. When treating such densities, we make the assumption that each constituent summand $f_i$ is $L$-smooth and convex, and the overall function $F$ is $\mu$-strongly convex. We measure complexity of algorithms for logconcave finite sums by gradient queries to summands, i.e.\ $\nabla f_i(x)$ for some $i \in [n]$.

\paragraph{Optimization.} Throughout this work, we are somewhat liberal with assuming access to minimizers to various functions (namely, the negative log-densities of target distributions). We give a more thorough discussion of this assumption in Appendix~\ref{app:xstar}, but note here that for all function families we consider (well-conditioned, composite, and finite sum), efficient first-order methods exist for obtaining high accuracy minimizers, and this optimization query complexity is never the leading-order term in any of our algorithms assuming polynomially bounded initial error.

\subsection{Technical facts}

We will repeatedly use the following results.

\begin{fact}[Gaussian integral]\label{fact:gaussz}
For any $\lam \ge 0$ and $v \in \R^d$, 
\[\int \exp\Par{-\frac{1}{2\lam}\norm{x - v}_2^2} dx = (2\pi\lam)^{\frac{d}{2}}.\]
\end{fact}

\begin{fact}[\cite{DwivediCW018}, Lemma 1]\label{fact:subgauss}
Let $\pi$ be a $\mu$-strongly logconcave distribution, and let $x^*$ minimize its negative log-density. Then, for $x \sim \pi$ and any $\delta \in [0, 1]$, with probability at least $1 - \delta$,
\[\norm{x - x^*}_2 \le \sqrt{\frac{d}{\mu}} \Par{2 + 2\max\Par{\sqrt[4]{\frac{\log(1/\delta)}{d}}, \sqrt{\frac{\log(1/\delta)}{d}}}}.\]
\end{fact}

\begin{fact}[\cite{Harge04}, Theorem 1.1]
	\label{fact:convexshrink}
	Let $\pi$ be a $\mu$-strongly logconcave density. Let $d\gamma_\mu(x)$ be the Gaussian density with covariance matrix $\mu^{-1}\id$. For any convex function $h$,
	\[\E_\pi[h(x - \E_\pi[x])] \le \E_{\gamma_\mu}[h(x - \E_{\gamma_\mu}[x])]. \]
\end{fact}

\begin{fact}[\cite{DurmusM19}, Theorem 1]\label{fact:distxstar}
	Let $\pi$ be a $\mu$-strongly logconcave distribution, and let $x^*$ minimize its negative log-density. Then, $\E_\pi[\norm{x - x^*}_2^2] \le \tfrac{d}{\mu}$.
\end{fact}

\section{Proximal reduction framework}
\label{sec:framework}

The reduction framework of Theorem~\ref{thm:fzerocomp} can be thought of as a specialization of a more general composite sampler which we develop in Section~\ref{sec:composite}, whose guarantees are reproduced here.

\restatekkdcomp*

Our main observation, elaborated on more formally for specific applications in Section~\ref{sec:improve}, is that a variety of structured logconcave densities have negative log-densities $\fcomp$, where we can implement an efficient restricted Gaussian oracle for $\fcomp$ via calling an existing sampling method. Crucially, in these instantiations we use the fact that the distributions which $\oracle$ is required to sampled from are heavily regularized (restricted by a quadratic with large leading coefficient) to obtain fast samplers. We further note that the upper bound requirement on $\eta$ in Theorem~\ref{thm:mainclaim} can be lifted when the ``well-conditioned'' component is uniformly $0$. Instead of setting $f = 0$ and $g = \fcomp$ in Theorem~\ref{thm:mainclaim}, and refining the analysis for this special case to tolerate arbitrary $\eta$, we provide a self-contained proof here. This particular structure (the composite setting where $\fwc$ is uniformly zero and $\fcomp$ is strongly convex) admits significant simplifications from the more general case, so using slightly different proof techniques, we are able to obtain stronger convergence guarantees for this particular problem allowing for mixing in fewer than $d$ iterations from a feasible start (see Section~\ref{ssec:bug}).

\restatefzerocomp*

For simplicity of notation, we replace $\fcomp$ in the statement of Theorem~\ref{thm:fzerocomp} with $g$ throughout just this section. Let $\pi$ be a density on $\R^d$ with $\tfrac{d\pi}{dx}(x) \propto \exp(-g(x))$ where $g$ is $\mu$-strongly convex (but possibly non-smooth), and let $\oracle$ be a restricted Gaussian oracle for $g$. Consider the joint distribution $\pih$ supported on an expanded space $z = (x, y) \in \R^d \times \R^d$ with density, for some $\eta > 0$,
\[\frac{d\pih}{dz}(z) \propto \exp\Par{-g(x) - \frac{1}{2\eta}\norm{x - y}_2^2}.\]
Note that the $x$-marginal of $\pih$ is precisely $\pi$, so it suffices to sample from the $x$-marginal. We consider a simple alternating Markov chain for sampling from $\pih$, described in the following Algorithm~\ref{alg:alternatesample}.

\begin{algorithm}[ht!]\caption{$\AlternateSample(g, \eta, T)$}
	\label{alg:alternatesample}
	\textbf{Input:} $\mu$-strongly convex $g: \R^d \rightarrow \R$, $\eta > 0$, $T \in \N$, $x_0 = \min_x g(x)$. 
	\begin{algorithmic}[1]
		\For{$k \in [T]$}
		\State Sample $y_k \sim \pi_{x_{k - 1}}$, where for all $x \in \R^d$, $\tfrac{d\pi_x}{dy}(y) \propto \exp\Par{-\tfrac{1}{2\eta}\norm{x - y}_2^2}$.
		\State Sample $x_k \sim \pi_{y_k}$, where for all $y \in \R^d$, $\tfrac{d\pi_y}{dx}(x) \propto \exp\Par{-g(x) -\tfrac{1}{2\eta}\norm{x - y}_2^2}$.
		\EndFor
		\State \Return $x_T$
	\end{algorithmic}
\end{algorithm}
By observing that the distributions $\pi_x$ and $\pi_y$ in the above method are precisely the marginal distributions of $\pih$ with one variable fixed, it is straightforward to see that $\pih$ is a stationary distribution of the process. We make this formal in the following lemma.

\begin{lemma}[Alternating marginal sampling]
	\label{lem:alternate_exact}
	Let $\pih$ be a density on two blocks $(x, y)$. Sample $(x, y) \sim \pih$, and then sample $\tx \sim \pih(\cdot, y)$, $\ty \sim \pih(\tx, \cdot)$. Then, the distribution of $(\tx, \ty)$ is $\pih$. Moreover, the alternating marginal sampling Markov chain on either marginal is reversible.
\end{lemma}
\begin{proof}
	The density of the resulting distribution at $(\tx, y)$ is proportional to the product of the (marginal) density at $y$ and the conditional distribution of $\tx \mid y$, which by definition is $\pih$. Therefore, $(\tx, y)$ is distributed as $\pih$, and the argument for $\ty$ follows symmetrically. To see reversibility on the $x$ marginal, it suffices to note that the probability we move from $x$ to $x'$ is proportional to
	\[\int_{y} \pih(x, y) \pih(x', y) dy,\]
	which is a symmetric function of $x$ and $x'$. A similar argument holds for the $y$ marginals.
\end{proof}

We also state a simple observation about alternating schemes such as Algorithm~\ref{alg:alternatesample}, which will be useful later. Let $\prop_x$ be the density of $y_k$ after one step of the above procedure starting from $x_{k - 1} = x$, and let $\tran_x$ be the resulting density of $x_k$.

\begin{observation}\label{observe:pvst}
	For any two points $x$, $x' \in \R^d$, $\tvd{\tran_x}{\tran_{x'}} \le \tvd{\prop_x}{\prop_{x'}}$.
\end{observation}
\begin{proof}
	This follows by the coupling characterization of total variation (see e.g.\ Chapter 5 of \cite{LevinPW09}). Per the optimal coupling of $y \sim \prop_x$ and $y' \sim \prop_{x'}$, whenever the total variation sets $y = y'$ in Line 2 of $\AlternateSample$, we can couple the resulting distributions in Line 3 as well.
\end{proof}

In order to prove Theorem~\ref{thm:kd}, we first show that the random walk in Algorithm~\ref{alg:alternatesample} converges rapidly in the 2-Wasserstein distance (denoted $W_2$ in this section). 
\begin{lemma}
\label{lem:w2converge}
Let $\pi_{0}$ be the starting distribution of $x$ in Algorithm~\ref{alg:alternatesample}. Let $\pi_{k}$ be the distribution of $x_k$ and $\pi$ be the $x$-marginal of $\pih$. For all $k \ge 0$,
\begin{align*}
	W_2^2(\pi_{{k+1}}, \pi) \leq \frac{1}{(1+\eta\mu)^2}W^2_2(\pi_{{k}}, \pi).
\end{align*}
Hence, for any $\eta \leq \frac 1 \mu$, in $T'=O\Par{\frac{1}{\eta\mu}\log{\frac{d}{\mu\Delta}}}$ iterations, the random walk mixes to
\begin{align*}
	W_2(\pi_{{T'}}, \pi) \leq \Delta.
\end{align*} 
\end{lemma}
\begin{proof}
	Let $\Gamma_{x_k}$ be the optimal coupling between $x_k\sim  \pi_{k}$ and $\hat{x}\sim\pi$ according to the $W_2$ distance. Coupling the Gaussian random variable generating $y_{k+1}\sim\pi_{x_k}$ and $\yh\sim \pi_{\hat{x}}$ gives a coupling $\Gamma_{y_{k+1}}$ between $y_{k+1}$ and $\yh$ such that \begin{align}\label{eq:xycouple1}\E_{\Gamma_{y_{k+1}}}\Brack{\norm{y_{k+1} - \yh}_2^2} = \E_{\Gamma_{x_k}}\Brack{\norm{x_{k} - \xh}_2^2}.
	\end{align}
	Then,  let $\pi_y$ be the distribution of $x_{k+1}$ in a run of Line 3 of Algorithm~\ref{alg:alternatesample} starting from $y_{k +1} = y$, and $\pi_{\yh}$ be the distribution of $\hat{x}$ in Line 3 starting from $\yh$, respectively. Since $\pi_{\yh}$ is $\mu + \frac 1 \eta$ strongly log-concave, $\pi_{\yh}$ satisfies a log-Sobolev inequality with constant $\mu + \frac 1 \eta$ (Theorem 2 of \cite{otto2000generalization}). Hence, 
	\begin{align*} 
		W_2^2(\pi_y, \pi_{\yh}) & \leq \frac{2}{\mu +\frac 1 \eta} \dkl(\pi_y \| \pi_{\yh})\\
		& \leq \frac{1}{\Par{\mu + \frac 1 \eta}^2}  \E_{\pi_y}\Brack{\norm{\nabla \log \frac{\pi_y}{\pi_{\yh}}}_2^2}\\
		& \leq \frac 1 {(1+\eta\mu)^2}\norm{y-\yh}^2_2.
	\end{align*}
    The first step used the Talagrand transportation inequality (Theorem 1 of \cite{otto2000generalization}). The second step used the log-Sobolev inequality. The third step used
    \begin{align}
    	\label{eq:fisher_bound}
    	\nabla \log \frac{\pi_y(x)}{\pi_{\yh}(x)} &=\nabla \log \frac{\exp(-g(x)-\frac 1 {2\eta} \norm{x-y}_2^2)\int_{x'} \exp(-g(x)-\frac 1 {2\eta} \norm{x-\yh}_2^2)dx'}{\exp(-g(x)-\frac 1 {2\eta} \norm{x-\yh}_2^2)\int_{x'} \exp(-g(x)-\frac 1 {2\eta} \norm{x-y}_2^2)dx'} \notag\\ &= \frac 1 {2\eta} \nabla \Par{\norm{x-\yh}_2^2 - \norm{x-y}_2^2} = \frac 1 \eta (y - \yh).
    \end{align}
	Taking expectation over $\Gamma_{y_{k+1}}$ and using \eqref{eq:xycouple1} shows that 
	\begin{align*}
	W_2^2(\pi_{{k+1}}, \pi) \leq \frac{1}{(1+\eta\mu)^2}W^2_2(\pi_{k}, \pi).
	\end{align*}
	Algorithm~\ref{alg:alternatesample} starts from the distribution $\pi_0 = \delta_{x^*} $, where $x^* = \min_x g(x)$. Since $\pi$ is $\mu$-strongly logconcave, we have (see e.g.\ Proposition 1 of \cite{durmus2019high})
	\begin{align*}
		W_2^2 (\pi_{0}, \pi)  = \E_{\pih}\Brack{\norm{x^* - x}^2} \leq \frac d \mu.
	\end{align*}
    Then, for $\eta < \frac 1 \mu$, $\frac 1 {1 + \eta \mu} \leq 1- \frac{\eta \mu } 2$, so after $T' = O(\frac 1 {\eta \mu} \log \frac d {\mu \Delta}) $ iterations, $W_2(\pi_{{T'}}, \pi) \leq \Delta$.
\end{proof}
Next, we bound the KL divergence between the output of Algorithm~\ref{alg:alternatesample} and the target distribution $\pi$. We need the following standard lemma regarding KL divergences of marginal distributions.
\begin{lemma}
	\label{lem:expectkl}
	Let $P_z$ and $Q_z$ be distributions supported on $\mathcal{X}$ indexed by $z$, a random variable distributed as $\pi_z$. Let $\tP$ be the joint distribution of $(x,z)$ for $x\sim P_z$ and $z\sim \pi_z$, and $\tQ$ be the joint distribution of $(x,z)$ as $x\sim Q_z$ and $z\sim \pi_z$. Let $P$ and $Q$ be the marginal distribution of $\tP$ and $\tQ$ on $x$, averaged over $z$. Then,
	\begin{align*}
		\dkl(P\|Q) \leq \E_{z\sim\pi_z} \Brack{\dkl(P_z \| Q_z)}.
	\end{align*}
\end{lemma}
\begin{proof}
	By the definition of $\dkl$,
	\begin{align*}
	\dkl(\tP \| \tQ)  & = \E_{(x,z)\sim \tP} \Brack{\log \frac{\tP(x,z)}{\tQ(x,z)}} \\
	& = \E_{z\sim \pi_z} \Brack{\E_{x\sim P_z}\Brack{\log \frac{\tP(x,z)}{\tQ(x,z)}}} \\
	& = \E_{z\sim \pi_z} \Brack{\E_{x\sim P_z}\Brack{\log \frac{P_z(x)}{Q_z(x)}}} \\
	& = \E_{z\sim \pi_z}\Brack{\dkl(P_z\|Q_z)}.
	\end{align*}
   Finally, by the data processing inequality,
   	\begin{align*}
   	\dkl(P\|Q) \leq \dkl(\tP \| \tQ)=\E_{z\sim\pi_z} \Brack{\dkl(P_z \| Q_z)}.
   \end{align*}
\end{proof}
The following lemma shows that a 2-Wasserstein distance bound on the distribution at iteration $k$ implies a KL divergence bound on iteration $k+1$.
\begin{lemma}
	\label{lem:klconverge}
	Let $\pi_{{k}}$ be the distribution of $x_k$ for some $k$ such that $W_2(\pi_{k}, \pi)\leq \Delta$ and $ \pi$ be the x-marginal of $\pih$. Then,
	\begin{align*}
		\dkl(\pi_{{k+1}} \| \pi) \leq  \frac{\Delta^2}{2\eta}.
	\end{align*}
\end{lemma}
\begin{proof}
	As in Lemma~\ref{lem:w2converge}, let $\Gamma_{x_k}$ be the optimal coupling between $x_k\sim  \pi_{k}$ and $\hat{x}\sim\pi$, which yields a coupling $\Gamma_{y_{k+1}}$ between $y_{k+1}$ and $\yh$ such that \begin{align}\label{eq:xycouple}\E_{\Gamma_{y_{k+1}}}\Brack{\norm{y_{k+1} - \yh}_2^2} = \E_{\Gamma_{x_k}}\Brack{\norm{x_{k} - \xh}_2^2} \le \Delta^2.
	\end{align}
	Then,
	\begin{align*}
		\dkl(\pi_{k+1}\|\pi) &\leq \E_{(y_{k+1}, \yh) \sim \Gamma_{y_k}} \Brack{\dkl(\pi_{y_{k+1}} \| \pi_{\yh})} \\
		&\leq \frac{1}{2\eta^2\Par{\mu + \frac 1 \eta}}\E_{(y_{k+1}, \yh) \sim \Gamma_{y_{k+1}}} \Brack{\norm{y_{k+1} - \yh}_2^2} \leq \frac{\Delta^2}{2\eta}.
	\end{align*}
    The first inequality followed from Lemma~\ref{lem:klconverge} by taking $P=\pi_{k+1}$, $Q=\pi$ and $z = (y_{k+1}, y)$. The second inequality used the log-Sobolev inequality and \eqref{eq:fisher_bound}. The last inequality used \eqref{eq:xycouple}.
\end{proof}

Finally, putting the pieces together, Theorem~\ref{thm:fzerocomp} follows from Lemma~\ref{lem:w2converge} and Lemma~\ref{lem:klconverge}. 
\begin{proof}[Proof of Theorem~\ref{thm:fzerocomp}]
 By Lemma~\ref{lem:w2converge} and Lemma~\ref{lem:klconverge}, there is $T=O\Par{\frac{1}{\eta\mu}\log{\frac{d}{\eta\mu\epsilon}}}$ so that $\dkl(\pi_T\|\pi) \leq  2\epsilon^2$.
By Pinsker's inequality, 
$$\tvd{\pi_T}{\pi} \leq \sqrt{\frac 1 2 \dkl(\pi_T\|\pi)} = \epsilon.$$
\end{proof}

We note that Theorem~\ref{thm:fzerocomp} is robust to a small amount of error tolerance in the sampler $\oracle$. Specifically, if $\oracle$ has tolerance $\tfrac{\eps}{2T}$, then by calling Theorem~\ref{thm:fzerocomp} with desired accuracy $\tfrac{\eps}{2}$ and adjusting constants appropriately, the cumulative  error incurred by all calls to $\oracle$ is within the total requisite bound (formally, this can be shown via the coupling characterization of total variation). We defer a more formal elaboration on this inexactness argument to Appendix~\ref{app:xstar} and the proof of Proposition~\ref{prop:sjdguarantee}. 	%

\section{Tighter runtimes for structured densities}
\label{sec:improve}

In this section, we use applications of Theorem~\ref{thm:fzerocomp} to obtain simple analyses of novel state-of-the-art high-accuracy runtimes for the well-conditioned densities studied in \cite{DwivediCW018, ChenDWY19, LeeST20}, as well as the composite and finite sum densities studied in this work. We will assume the conclusions of Theorems~\ref{thm:mainclaim} and~\ref{thm:base_finitesum} respectively in deriving the results of Sections~\ref{ssec:comp} and~\ref{ssec:improved_fs}.

\subsection{Well-conditioned logconcave sampling: proof of Corollary~\ref{thm:kd}}\label{ssec:kd}

In this section, let $\pi$ be a distribution on $\R^d$ with density proportional to $\exp(-f(x))$, where $f$ is $L$-smooth and $\mu$-strongly convex (and $\kappa = \tfrac{L}{\mu}$) and has pre-computed minimizer $x^*$. We will instantiate Theorem~\ref{thm:fzerocomp} with $\fcomp(x) = f(x)$, and choose $\eta = \tfrac{1}{8Ld\log(\kappa)}$. We now require an $\eta$-RGO $\oracle$ for $\fcomp = f$ to use in Theorem~\ref{thm:fzerocomp}.

Our implementation of $\oracle$ is a rejection sampling scheme. We use the following helpful guarantee. 

\begin{lemma}[Rejection sampling]\label{lem:reject}
Let $\pi$, $\pih$ be distributions on $\R^d$ with $\tfrac{d\pi}{dx}(x) \propto p(x)$, $\tfrac{d\pih}{dx}(x) \propto \hp(x)$. Suppose for some $C \ge 1$ and all $x \in \R^d$, $\tfrac{p(x)}{\hp(x)} \le C$. The following is termed ``rejection sampling'': repeat independent runs of the following procedure until a point is outputted.
\begin{enumerate}
	\item Draw $x \sim \pih$.
	\item With probability $\tfrac{p(x)}{C\hp(x)}$, output $x$.
\end{enumerate}
Rejection sampling terminates in $\tfrac{C\int \hp(x)dx}{\int p(x)dx}$ runs in expectation, and the output distribution is $\pi$.
\end{lemma}
\begin{proof}
The second claim follows from Bayes' rule which implies the conditional density of the output point is proportional to $\hat{p}(x) \cdot \tfrac{p(x)}{C\hat{p(x)}} \propto p(x)$, so the distribution is $\pi$. To see the first claim, the probability any sample outputs is
\[\int_x \frac{p(x)}{C\hat{p}(x)} d\pih(x) = \frac{1}{C}\int_x \frac{\int_x p(x) dx}{\int_x \hp(x)dx} d\pi(x) = \frac{\int_x p(x) dx}{C\int_x \hp(x)dx}.\]
The conclusion follows by independence and linearity of expectation.
\end{proof}

We further state a concentration bound shown first in \cite{LeeST20} regarding the norm of the gradient of a point drawn from a logsmooth distribution.

\begin{proposition}[Logsmooth gradient concentration, Corollary 3.3, \cite{LeeST20}]\label{prop:gradconc}
Let $\pi$ be a distribution on $\R^d$ with $\tfrac{d\pi}{dx}(x) \propto \exp(-f(x))$ where $f$ is convex and $L$-smooth. With probability at least $1 - \kappa^{-d}$, 
\begin{equation}\label{eq:gradbound}\norm{\nabla f(x)}_2 \le 3\sqrt{L} d \log \kappa \text{ for } x \sim \pi.\end{equation}
\end{proposition}

By the requirements of Theorem~\ref{thm:fzerocomp}, the restricted Gaussian oracle $\oracle$ only must be able to draw samples from densities of the form, for some $y \in \R^d$,
\begin{equation}\label{eq:xsample}\exp\Par{-\fcomp(x) - \frac{1}{2\eta}\norm{x - y}_2^2} = \exp\Par{-f(x) - 4Ld\log\kappa\norm{x - y}_2^2}.\end{equation}

We will use the following Algorithm~\ref{alg:xsample} to implement $\oracle$.

\begin{algorithm}[ht!]\caption{$\XSample(f, y, \eta)$}
	\label{alg:xsample}
	\textbf{Input:} $L$-smooth, $\mu$-strongly convex $f:\R^d \rightarrow \R$, $y \in \R^d$, $\eta > 0$ 
	\begin{algorithmic}[1]
		\If{$\norm{\nabla f(y)}_2 \le 3\sqrt{L}d\log\kappa$}
		\While{\textbf{true}}
		\State Draw $x \sim \Nor(y - \nabla f(y), \eta\id)$
		\State $\tau \sim \text{Unif}[0, 1]$
		\If{$\tau \le \exp(f(y) + \inprod{\nabla f(y)}{x - y} - f(x))$}
		\State \Return $x$
		\EndIf
		\EndWhile
		\EndIf
		\State Use \cite{ChenDWY19} to sample $x$ from \eqref{eq:xsample} to total variation distance $\tfrac{\eps}{\Theta(\kappa d^2\log^3(\frac{\kappa d}{\eps}))}$ using $O(d\log\tfrac{\kappa d}{\eps})$ queries to $\nabla f$ (Theorem 1, \cite{ChenDWY19}, where \eqref{eq:xsample} has constant condition number)
		\State \Return $x$
	\end{algorithmic}
\end{algorithm}

\begin{lemma}\label{lem:xsample}
Let $\eta = \tfrac{1}{8Ld\log(\kappa)}$, and suppose $y$ satisfies the bound in \eqref{eq:gradbound}, i.e.\ $\norm{\nabla f(y)}_2 \le 3\sqrt{L}d\log\kappa$. Then, Line 3 of Algorithm~\ref{alg:xsample} runs an expected 2 times, and Algorithm~\ref{alg:xsample} samples exactly from \eqref{eq:xsample}, whenever the condition of Line 1 is met.
\end{lemma}
\begin{proof}
Note that when the assumption of Line 1 is met, Algorithm~\ref{alg:xsample} is an instantiation of rejection sampling (Lemma~\ref{lem:reject}) with
\begin{align*}p(x) &= \exp\Par{-f(x) - \frac{1}{2\eta}\norm{x - y}_2^2},\\ \hp(x) &= \exp\Par{-f(y) - \inprod{\nabla f(y)}{x - y} - \frac{1}{2\eta}\norm{x - y}_2^2}.\end{align*}
By convexity, we may take $C = 1$. Next, by applying Fact~\ref{fact:gaussz} twice and $L$-smoothness of $\fcomp$,
\begin{align*}
\int_x p(x) dx &\ge \int_x \exp\Par{-f(y) - \inprod{\nabla f(y)}{x - y} - \frac{1 + \eta L}{2\eta}\norm{x - y}_2^2} dx \\
&= \exp\Par{-f(y) +\frac{\eta}{2(1 + \eta L)}\norm{\nabla f(y)}_2^2} \int_x \exp\Par{-\frac{1 + \eta L}{2\eta}\norm{x - y + \frac{\eta}{1 + \eta L}\nabla f(y)}_2^2} dx\\
&= \exp\Par{-f(y) +\frac{\eta}{2(1 + \eta L)}\norm{\nabla f(y)}_2^2}\Par{\frac{2\pi\eta}{1 + \eta L}}^{\frac{d}{2}},\\
\int_x \hp(x)dx &= \exp\Par{-f(y)+\frac{\eta}{2}\norm{\nabla f(y)}_2^2}(2\pi\eta)^{\frac{d}{2}},
\end{align*}
which implies the desired bound (recalling Lemma~\ref{lem:reject} and our assumed bound on $\norm{\nabla f(y)}_2$)
\begin{align*}
\frac{\int \hp(x)dx}{\int p(x)dx} &\le \exp\Par{\Par{\frac{\eta}{2} - \frac{\eta}{2(1 + \eta L)}}\norm{\nabla f(y)}_2^2}(1 + \eta L)^{\frac{d}{2}} \\
&\le 1.5\exp\Par{\frac{\eta^2 L}{2(1 + \eta L)}\norm{\nabla f(y)}_2^2} \le 2.
\end{align*}
\end{proof}

We are now equipped to prove our main result concerning well-conditioned densities.

\restatekd*
\begin{proof}
By applying Theorem~\ref{thm:fzerocomp} with the chosen $\eta$, and noting that the cumulative error due to all calls to Line 10 cannot amount to more than $\tfrac{\eps}{2}$ total variation error throughout the algorithm, it suffices to show that Algorithm~\ref{alg:xsample} uses $O(1)$ gradient queries each iteration in expectation. This happens whenever the condition in Line 1 is met via Lemma~\ref{lem:xsample}, so we must show Line 10 is executed with probability $O((d\log\frac{\kappa d}{\eps})^{-1})$.

To show this, note that combining Proposition~\ref{prop:gradconc} with the warmness of the start $x_0$ in Algorithm~\ref{alg:xsample}, this event occurs with probability at most $\kappa^{-\frac{d}{2}}$ in the first iteration.\footnote{Formally, Line 2 of Algorithm~\ref{alg:alternatesample} has $y_1 \sim \Nor(x_0, \eta\id)$, but by smoothness $\norm{\nabla f(y_1)}_2 \le \norm{\nabla f(x_0)}_2 + L\norm{x - y}_2$ and $L\norm{x - y}_2 \le \tO(L\sqrt{\eta})$ with high probability, adding a negligible constant to the bound of Proposition~\ref{prop:gradconc}.} Since warmness is monotonically decreasing\footnote{This is a standard fact in the literature, and can be seen as follows: each transition step in the chain is a convex combination of warm point masses, preserving warmness.} throughout using an exact oracle in Algorithm~\ref{alg:alternatesample}, and the total error accumulated due to Line 10 throughout the algorithm is $O((d\log\frac{\kappa d}{\eps})^{-1})$, we have the desired conclusion.
\end{proof}

We show a bound nearly-matching Corollary~\ref{thm:kd} using only value access to $f$, and with a deterministic iteration complexity (rather than an expected one), as Corollary~\ref{corr:zerokd} in Section~\ref{ssec:improved_fs}.

\subsection{Composite logconcave sampling: proof of Corollary~\ref{thm:kdcomp}}\label{ssec:comp}

In this section, let $\pi$ be a distribution on $\R^d$ with density proportional to $\exp(-f(x) - g(x))$, where $f$ is $L$-smooth and $\mu$-strongly convex (and $\kappa = \tfrac{L}{\mu}$), and $g$ is convex and admits a restricted Gaussian oracle $\oracle$. Without loss of generality, we assume that $f$ and $g$ share a minimizer $x^*$ which we have pre-computed; if this is not the case, we can redefine $f(x) \gets f(x) - \inprod{\nabla f(x^*)}{x}$ and $g(x) \gets g(x) + \inprod{\nabla f(x^*)}{x}$; see Section~\ref{ssec:sharedmin} for this reduction.

We will instantiate Theorem~\ref{thm:fzerocomp} with $\fcomp = f + g$, which is a $\mu$-strongly convex function. Our main result of this section follows directly from Theorem~\ref{thm:fzerocomp} and using Theorem~\ref{thm:mainclaim} as the required oracle $\oracle$, stated more precisely in the following.

\begin{comment}
\begin{lemma}\label{lem:warmstartcomp}
Define the distribution $\pistart$ on $\R^d$ with density
\[\frac{d\pistart}{dx}(x) \propto \exp\Par{-f(x^*) - \frac{L}{2}\norm{x - x^*}_2^2 - g(x)}.\]
Then, $\pistart$ is a $\beta$-warm distribution for $\pi$ with $\beta = \kappa^{\frac{d}{2}}$.
\end{lemma}
\begin{proof}
We wish to show that for all $x \in \R^d$,
\[\frac{d\pistart}{d\pi}(x) = \frac{\exp\Par{-f(x^*) - \frac{L}{2}\norm{x - x^*}_2^2 - g(x)}}{\exp\Par{-f(x) - g(x)}} \cdot \frac{\int \exp(-f(x) - g(x)) dx}{\int \exp\Par{-f(x^*) - \frac{L}{2}\norm{x - x^*}_2^2 -g(x)} dx} \le \beta.\]
First, by smoothness and the assumption that $x^*$ minimizes $f$ (e.g.\ $\nabla f(x^*) = 0$), $-f(x^*) - \tfrac{L}{2}\norm{x - x^*}_2^2 \le -f(x)$, so it suffices to bound the ratio of the normalization constants. Note that
\[\frac{\int \exp(-f(x) - g(x)) dx}{\int \exp\Par{-f(x^*) - \frac{L}{2}\norm{x - x^*}_2^2 -g(x)} dx} \le \frac{\int \exp(-f(x^*) - \frac{\mu}{2}\norm{x - x^*}_2^2 - g(x)) dx}{\int \exp\Par{-f(x^*) - \frac{L}{2}\norm{x - x^*}_2^2 -g(x)} dx}.\]
Finally, applying Proposition~\ref{prop:normalizationratio}, where we recall $\tfrac{\mu}{2}\norm{x - x^*}_2^2 + g(x)$ has minimizer $x^*$ and is $\mu$-strongly convex, we have the desired
\[\frac{d\pistart}{d\pi}(x) \le \frac{\int \exp(-f(x^*) - \frac{\mu}{2}\norm{x - x^*}_2^2 - g(x)) dx}{\int \exp\Par{-f(x^*) - \frac{L}{2}\norm{x - x^*}_2^2 -g(x)} dx} \le \kappa^{\frac{d}{2}}.\]
\end{proof}
\end{comment}

\restatekdcomp*
\begin{proof}
As discussed at the beginning of this section, assume without loss that $f$ and $g$ both are minimized by $x^*$. We apply the algorithm of Theorem~\ref{thm:fzerocomp} with $\eta = \tfrac{1}{L}$ to the $\mu$-strongly convex function $f + g$, which requires one call to $\oracle$ to implement. Thus, the iteration count parameter in Theorem~\ref{thm:fzerocomp} is $T = O(\kappa\log\tfrac{\kappa d}{\eps})$.

Recall that we chose $\eta = \tfrac{1}{L}$. To bound the total complexity of this algorithm, it suffices to give an $\eta$-RGO $\oracleplus$ for sampling from distributions with densities of the form, for some $y \in \R^d$,
\[\exp\Par{-f(x) - g(x) - \frac{1}{2\eta}\norm{x - y}_2^2} = \exp\Par{-f(x) - g(x) - \frac{L}{2}\norm{x - y}_2^2}\]
to total variation distance $\tfrac{\eps}{\Theta(T)}$ (see discussion at the end of Section~\ref{sec:framework}). To this end, we apply Theorem~\ref{thm:mainclaim} with the well-conditioned component $f(x) + \tfrac{L}{2}\norm{x - y}_2^2$, the composite component $g(x)$, and the largest possible choice of $\eta$. Note that we indeed have access to a restricted Gaussian oracle for $g$ (namely, $\oracle$), and this choice of well-conditioned component is $2L$-smooth and $L$-strongly convex, so its condition number is a constant. Thus, Theorem~\ref{thm:mainclaim} requires $O(d\log^2\tfrac{\kappa d}{\eps})$ calls to $\oracle$ and gradients of $f$ to implement the desired $\oracleplus$ on any query $y$ (where we note $\tfrac{\eps}{\Theta(T)} = \tfrac{1}{\textup{poly}(\kappa, d, \eps^{-1})}$). Combining these complexity bounds yields the desired conclusion.
\end{proof}

\subsection{Sampling logconcave finite sums: proof of Corollary~\ref{thm:improved_finitesum}}
\label{ssec:improved_fs}

In this section, let $\pi$ be a distribution on $\R^d$ with density proportional to $\exp(-F(x))$, where $F(x) = \tfrac{1}{n} \sum_{i \in [n]} f_i(x)$ is $\mu$-strongly convex, and for all $i \in [n]$, $f_i$ is $L$-smooth (and $\kappa = \tfrac L \mu$). We will instantiate Theorem~\ref{thm:fzerocomp} with $\fcomp(x) = F(x)$, and Theorem~\ref{thm:base_finitesum} as an $\eta$-RGO for some choice of $\eta$.

More precisely, Theorem~\ref{thm:base_finitesum} shows that given access to the minimizer $x^*$, only zeroth-order access to the summands of $F$ is necessary to obtain the iteration bound. In order to obtain the minimizer to high accuracy however, variance reduced stochastic gradient methods (e.g.\ \cite{Johnson013}) require $\Omega(n + \kappa)$ gradient queries, which amounts to $\Omega((n + \kappa)d)$ function evaluations. We state a convenient corollary of Theorem~\ref{thm:base_finitesum} which removes the requirement of accessing $x^*$, via an optimization pre-processing step using the method of \cite{Johnson013} (see further discussion in Appendix~\ref{app:xstar}). This is useful to us in proving Theorem~\ref{thm:improved_finitesum} because in the sampling tasks required by the RGO, the minimizer changes (and thus must be recomputed every time).

\begin{corollary}[First-order logconcave finite sum sampling]\label{corr:fs_xstar}
	In the setting of Theorem~\ref{thm:base_finitesum}, using \cite{Johnson013} to precompute the minimizer $x^*$ and running Algorithm~\ref{alg:mrw} uses $O(n\log\tfrac{\kappa d}{\eps}+ \kappa^2 d\log^4\tfrac{n\kappa d}{\eps})$ first-order oracle queries to summands $\{f_i\}_{i \in [n]}$ and obtains $\eps$ total variation distance to $\pi$.
\end{corollary}

We now apply the reduction framework developed in Section~\ref{sec:prelims} to our Algorithm~\ref{alg:mrw} to obtain an improved query complexity for sampling from logconcave finite sums. 

\restatefirstfs*

\begin{proof}
	We apply Theorem~\ref{thm:fzerocomp} with $\mu$-strongly convex $\fcomp = F(x)$, using Algorithm~\ref{alg:mrw} as the required $\eta$-RGO $\oracle$ for sampling from distributions with densities of the form
	\[	\exp\left( -{F}(x)-\frac 1 \eta \norm{x-y}^2_2 \right)\]
	for some $y \in \R^d$, to total variation $\tfrac {\eps}{\Theta(T)}$ (see Section~\ref{sec:framework}) for $T$ the iteration bound of Algorithm~\ref{alg:alternatesample}. We apply Theorem~\ref{thm:base_finitesum}  to the function $\widetilde{F}(x) ={F}(x)+\frac 1 \eta \norm{x-y}_2^2$; we can express this in finite sum form by adding $\tfrac{1}{\eta}\norm{x - y}_2^2$ to every constituent function, and the effect on gradient oracles is $\tfrac{1}{\eta}(x - y)$. Note $\widetilde{F}$ has condition number $O(1 + \eta L)$. For a given $\eta$, the overall complexity is
	\[\frac{\log \frac{\kappa d}{\eps}}{\eta\mu} \Par{n\log\Par{\frac{n\kappa d}{\eps}} + d\log^4\Par{\frac{n\kappa d}{\eps}} + (\eta L)^2 d \log^4\Par{\frac{n\kappa d}{\eps}}}\]
	Here, the inner loop complexity uses Corollary~\ref{corr:fs_xstar} to also find the minimizer (for warm starts), and the outer loop complexity is by Theorem~\ref{thm:fzerocomp}. The result follows by optimizing over $\eta$, namely picking $\eta = \max(\tfrac 1 L, \sqrt{\frac{n}{L^2 d \log^3(n\kappa d /\eps)}})$, and that Algorithm~\ref{alg:alternatesample} always must have at least one iteration.
\end{proof}

Note the only place that Corollary~\ref{thm:improved_finitesum} used gradient evaluations was in determining minimizers of subproblems, via the first step of Corollary~\ref{corr:fs_xstar}. Consider now the $n = 1$ case. By running e.g.\ accelerated gradient descent for smooth and strongly convex functions, it is well-known \cite{Nesterov83} that we can obtain a minimizer in $\tO(\sqrt{\kappa})$ iterations, each querying a gradient oracle, where $\kappa$ is the condition number. By smoothness, we can approximate every coordinate of the gradient to arbitrary precision using $2$ function evaluations, so this is a $\tO(\sqrt{\kappa} d)$ value oracle complexity.

Finally, for every optimization subproblem in Corollary~\ref{thm:improved_finitesum} where $\eta = (L\cdot \text{polylog}\frac{\kappa d}{\eps})^{-1}$, the condition number is a constant, which amounts to a $\tO(d)$ value oracle complexity for computing a minimizer. This is never the dominant term compared to Theorem~\ref{thm:base_finitesum}, yielding the following conclusion.

\begin{corollary}\label{corr:zerokd}
	In the setting of Corollary~\ref{thm:kd}, Algorithm~\ref{alg:alternatesample} using Algorithm~\ref{alg:mrw} as a restricted Gaussian oracle uses $O(\kappa d \log^2 \tfrac{\kappa d}{\eps})$ value queries and obtains $\eps$ total variation distance to $\pi$.
\end{corollary}

We note that the polylogarithmic factor is significantly improved when compared to Corollary~\ref{thm:improved_finitesum} by removing the random sampling steps in Algorithm~\ref{alg:mrw}. A precise complexity bound of the resulting Metropolized random walk, a zeroth-order algorithm mixing in $O(\kappa^2 d \log\tfrac{\kappa d}{\eps})$ for a logconcave distribution with condition number $\kappa$, is given as Theorem 2 of \cite{ChenDWY19}.

Finally, in the case $n \ge 1$, we also exhibit an improved query complexity in terms of an entirely zeroth-order sampling algorithm which interpolates with Corollary~\ref{corr:zerokd} (up to logarithmic factors). By trading off the $\tO(nd + \kappa d)$ zeroth-order complexity of minimizing a finite sum function \cite{Johnson013}, and the $\tO(\kappa^2 d)$ zeroth-order complexity of sampling, we can run Theorem~\ref{thm:fzerocomp} for the optimal choice of $\eta = \tO(\tfrac{\sqrt{n}}{L})$. The overall zeroth-order complexity can be seen to be $\tO(nd + \sqrt{n}\kappa d)$. 	%

\section{Composite logconcave sampling with a restricted Gaussian oracle}
\label{sec:composite}

In this section, we provide our ``base sampler'' for composite logconcave densities as Algorithm~\ref{alg:csg}, and give its guarantees by proving Theorem~\ref{thm:mainclaim}. Throughout, fix distribution $\pi$ with density 
\begin{equation}\label{eq:pidef}\begin{aligned}\frac{d\pi}{dx}(x) \propto \exp\left(-f(x) - g(x)\right),\text{where }f:\R^d \rightarrow \R\text{ is }L\text{-smooth, } \mu\text{-strongly convex,} \\
\text{and }g: \R^d \rightarrow \R \text{ admits a restricted Gaussian oracle } \oracle.\end{aligned}\end{equation}
We will define $\kappa \defeq \frac L \mu$, and assume that we have precomputed $x^* \defeq \argmin_{x \in \R^d}\left\{f(x) + g(x)\right\}$. Our algorithm proceeds in stages following the outline in Section~\ref{ssec:compositeintro}.

\begin{enumerate}
	\item $\csg$ is reduced to $\cssm$, which takes as input a distribution with negative log-density $f + g$, where $f$ and $g$ share a minimizer; this reduction is given in Section~\ref{ssec:sharedmin}, and the remainder of the section handles the shared-minimizer case.
	\item The algorithm $\cssm$ is a rejection sampling scheme built on top of sampling from a joint distribution $\pih$ on $(x, y) \in \R^d \times \R^d$ whose $x$-marginal approximates $\pi$. We give this reduction in Section~\ref{ssec:outerloop}. 
	\item The bulk of our analysis is for $\sjd$, an alternating marginal sampling algorithm for sampling from $\pih$. To implement marginal sampling, it alternates calls to $\oracle$ and a rejection sampling algorithm $\yor$. We prove its correctness in Section~\ref{ssec:alternate}. 
\end{enumerate}

We put these pieces together in Section~\ref{ssec:proofmainclaim} to prove Theorem~\ref{thm:mainclaim}. We remark that for simplicity, we will give the algorithms corresponding to the largest value of step size $\eta$ in the theorem statement; it is straightforward to modify the bounds to tolerate smaller values of $\eta$, which will cause the mixing time to become correspondingly larger (in particular, the value of $K$ in Algorithm~\ref{alg:sjd}).

\begin{algorithm}[ht!]\caption{$\csg(\pi, x^*, \eps)$}
	\label{alg:csg}
	\textbf{Input:} Distribution $\pi$ of form \eqref{eq:pidef}, $x^*$ minimizing negative log-density of $\pi$, $\eps \in [0, 1]$. \\
	\textbf{Output:} Sample $x$ from a distribution $\pi'$ with $\tvd{\pi'}{\pi} \le \eps$.
	\begin{algorithmic}[1]
		\State $\tilde{f}(x) \gets f(x) - \inprod{\nabla f(x^*)}{x}$, $\tilde{g}(x) \gets g(x) + \inprod{\nabla f(x^*)}{x}$
		\State \Return $\cssm(\pi, \tilde{f}, \tilde{g}, x^*, \eps)$
	\end{algorithmic}
\end{algorithm}

\begin{algorithm}[ht!]\caption{$\cssm(\pi, f, g, x^*, \eps)$}
	\label{alg:cssm}
	\textbf{Input:} Distribution $\pi$ of form \eqref{eq:pidef}, where $f$ and $g$ are both minimized by $x^*$, $\eps \in [0, 1]$. \\
	\textbf{Output:} Sample $x$ from a distribution $\pi'$ with $\tvd{\pi'}{\pi} \le \eps$.
	\begin{algorithmic}[1]
		\While {\textbf{true}}
		\State Define the set
		\begin{equation}\label{eq:omegadef} \Omega \defeq \left\{x \mid \norm{x - x^*}_2 \le 4\sqrt{\frac{d\log(288\kappa/\eps)}{\mu}}\right\}\end{equation}
		\State $x \gets \sjd(f, g, x^*, \oracle, \tfrac{\eps}{18})$
		\If{$x \in \Omega$}
		\State $\tau \sim \text{Unif}[0, 1]$
		\State $y \gets \yor(f, x, \eta)$
		\State $\alpha \gets \exp\left(f(y) - \inprod{\nabla f(x)}{y - x} - \tfrac{L}{2}\norm{y - x}_2^2 + g(x) + \tfrac{\eta L^2}{2}\norm{x - x^*}_2^2\right)$
		\State $\hat{\theta} \gets \exp\left(-f(x) - g(x) + \tfrac{\eta}{2(1 + \eta L)}\norm{\nabla f(x)}_2^2\right)(1 + \eta L)^{\frac{d}{2}}\alpha$
		\If{$\tau \le \tfrac{\hat{\theta}}{4}$}
		\State \Return $x$
		\EndIf
		\EndIf
		\EndWhile 
	\end{algorithmic}
\end{algorithm}
\begin{algorithm}[ht!]\caption{$\sjd(f, g, x^*, \eta, \oracle, \delta)$}
	\label{alg:sjd}
	\textbf{Input:} $f$, $g$ of form \eqref{eq:pidef} both minimized by $x^*$, $\delta \in [0, 1]$, $\eta > 0$, $\oracle$ restricted Gaussian oracle for $g$.\\
	\textbf{Output:} Sample $x$ from a distribution $\pih'$ with $\tvd{\pih'}{\pih} \le \delta$, where we overload $\pih$ to mean the marginal of \eqref{eq:pihdef} on the $x$ variable.
	\begin{algorithmic}[1]
		\State $\eta \gets \tfrac{1}{32 L\kappa d\log(16\kappa/\delta)}$
		\State Let $\pih$ be the density with
		\begin{equation}\label{eq:pihdef}\frac{d\pih}{dx}(z) \propto \exp\left(-f(y) - g(x) - \frac{1}{2\eta}\norm{y - x}_2^2 - \frac{\eta L^2}{2}\norm{x - x^*}_2^2\right) \end{equation}
		\State Call $\oracle$ to sample $x_0 \sim \pistart$, for
		\begin{equation}\label{eq:pistartdef}\frac{d\pistart(x)}{dx} \propto \exp\left(-\frac{L + \eta L^2}{2}\norm{x - x^*}_2^2 - g(x)\right)\end{equation}
		\State $K \gets \frac{2^{26}\cdot100}{\eta\mu}\log\left(\frac{d\log(16\kappa)}{4\delta}\right)$ (see Remark~\ref{rem:comments})
		\For{$k \in [K]$}
		\State Call $\yor\left(f, x_{k - 1}, \eta, \tfrac{\delta}{2Kd\log(\frac{d\kappa}{\delta})}\right)$ to sample $y_k \sim \pi_{x_{k - 1}}$ (Algorithm~\ref{alg:yor}), for
		\begin{equation}\label{eq:pixdef}\frac{d\pi_x}{dy}(y) \propto \exp\left(-f(y) - \frac{1}{2\eta}\norm{y - x}_2^2\right)\end{equation}
		\State Call $\oracle$ to sample $x_k \sim \pi_{y_k}$, for
		\begin{equation}\label{eq:piydef}\frac{d\pi_y}{dx}(x) \propto \exp\left(-g(x) - \frac{1}{2\eta}\norm{y - x}_2^2 - \frac{\eta L^2}{2}\norm{x - x^*}_2^2\right)\end{equation}
		\EndFor
		\State \Return $x_K$
	\end{algorithmic}
\end{algorithm}

\subsection{Reduction from $\csg$ to $\cssm$}
\label{ssec:sharedmin}

Correctness of $\csg$ is via the following properties.

\begin{restatable}{proposition}{restatecsgcorrectness}
	\label{prop:csgcorrectness}
	Let $\tilde{f}$ and $\tilde{g}$ be defined as in $\csg$. 
	\begin{enumerate}
		\item The density $\propto \exp(-f(x) - g(x))$ is the same as the density $\propto \exp(-\tilde{f}(x) - \tilde{g}(x))$.
		\item Assuming first-order (function and gradient evaluation) access to $f$, and restricted Gaussian oracle access to $g$, we can implement the same accesses to $\tilde{f}$, $\tilde{g}$ with constant overhead.
		\item $\tilde{f}$ and $\tilde{g}$ are both minimized by $x^*$.
	\end{enumerate}
\end{restatable}

\begin{proof}
	For $f$ and $g$ with properties as in \eqref{eq:pidef}, with $x^*$ minimizing $f + g$, define the functions 
	\[\tilde{f}(x) \defeq f(x) - \inprod{\nabla f(x^*)}{x},\; \tilde{g}(x) \defeq g(x) + \inprod{\nabla f(x^*)}{x},\]
	and observe that $\tilde{f}+ \tilde{g} = f + g$ everywhere. This proves the first claim. Further, implementation of a first-order oracle for $\tilde{f}$ and a restricted Gaussian oracle for $\tilde{g}$ are immediate assuming a first-order oracle for $f$ and a restricted Gaussian oracle for $g$, showing the second claim; any quadratic shifted by a linear term is the sum of a quadratic and a constant. We now show $\tilde{f}$ and $\tilde{g}$ have the same minimizer. By strong convexity, $\tilde{f}$ has a unique minimizer; first-order optimality shows that
	\[\nabla \tilde{f}(x^*) = \nabla f(x^*) - \nabla f(x^*) = 0,\]
	so this unique minimizer is $x^*$. Moreover, optimality of $x^*$ for $f + g$ implies that for all $x \in \R^d$,
	\[\inprod{\partial g(x^*) + \nabla f(x^*)}{x^* - x} \le 0.\]
	Here, $\partial g$ is a subgradient. This shows first-order optimality of $x^*$ for $\tilde{g}$ also, so $x^*$ minimizes $\tilde{g}$. 
\end{proof}

\subsection{Reduction from $\cssm$ to $\sjd$}
\label{ssec:outerloop}

$\cssm$ is a rejection sampling scheme, which accepts samples from subroutine $\sjd$ in the high-probability region $\Omega$ defined in \eqref{eq:omegadef}. We give a general analysis for approximate rejection sampling in Appendix~\ref{sssec:approxreject}, and Appendix~\ref{sssec:pipih} bounds relationships between distributions $\pi$ and $\pih$, defined in \eqref{eq:pidef} and \eqref{eq:pihdef} respectively (i.e.\ relative densities and normalization constant ratios). Combining these pieces proves the following main claim.

\begin{restatable}{proposition}{restatecssmcorrectness}
	\label{prop:cssmcorrectness}
	Let $\eta = \tfrac{1}{32 L\kappa d\log(288\kappa/\eps)}$, and assume $\sjd(f, g, x^*, \oracle, \delta)$ samples within $\delta$ total variation of the $x$-marginal on \eqref{eq:pihdef}. $\cssm$ outputs a sample within total variation $\eps$ of \eqref{eq:pidef} in an expected $O(1)$ calls to $\sjd$.
\end{restatable}

\subsection{Implementing $\sjd$}
\label{ssec:alternate}

$\sjd$ alternates between sampling marginals in the joint distribution $\pih$, as seen by definitions \eqref{eq:pixdef}, \eqref{eq:piydef}. We showed that marginal sampling attains the correct stationary distribution as Lemma~\ref{lem:alternate_exact}. We bound the conductance of the induced walk on iterates $\{x_k\}$ by combining an isoperimetry bound with a total variation guarantee between transitions of nearby points in Appendix~\ref{sssec:sjdconduct}. Finally, we give a simple rejection sampling scheme $\yor$ as Algorithm~\ref{alg:yor} for implementing the step \eqref{eq:pixdef}. Since the $y$-marginal of $\pih$ is a bounded perturbation of a Gaussian (intuitively, $f$ is $L$-smooth and $\eta^{-1} \gg L$), we show in a high probability region that rejecting from the sum of a first-order approximation to $f$ and the Gaussian succeeds in $2$ iterations. 

\begin{remark}\label{rem:comments}
	For simplicity of presentation, we were conservative in bounding constants throughout; in practice, we found that the constant in Line 4 is orders of magnitude too large (a constant $< 10$ sufficed), which can be found as Section 4 of \cite{ShenTL20}. Several constants were inherited from prior analyses, which we do not rederive to save on redundancy.
\end{remark}

We now give a complete guarantee on the complexity of $\sjd$.

\begin{restatable}{proposition}{restatesjdguarantee}\label{prop:sjdguarantee}
	$\sjd$ outputs a point with distribution within $\delta$ total variation distance from the $x$-marginal of $\pih$. The expected number of gradient queries per iteration is constant.
\end{restatable}

\subsection{Putting it all together: proof of Theorem~\ref{thm:mainclaim}}
\label{ssec:proofmainclaim}

We show Theorem~\ref{thm:mainclaim} follows from the guarantees of Propositions~\ref{prop:csgcorrectness},~\ref{prop:cssmcorrectness}, and~\ref{prop:sjdguarantee}. Formally, Theorem~\ref{thm:mainclaim} is stated for an arbitrary value of $\eta$ which is upper bounded by the value in Line 1 of Algorithm~\ref{alg:sjd}; however, it is straightforward to see that all our proofs go through for any smaller value. By observing the value of $K$ in $\sjd$, we see that the number of total iterations in each call to $\sjd$ $O\left(\tfrac{1}{\eta\mu}\log(\tfrac{\kappa d}{\eps})\right) = O\left(\kappa^2 d \log^2\left(\tfrac{\kappa d}{\delta}\right)\right).$
Proposition~\ref{prop:sjdguarantee} also shows that every iteration, we require an expected constant number of gradient queries and calls to $\oracle$, the restricted Gaussian oracle for $g$, and that the resulting distribution has $\delta$ total variation from the desired marginal of $\pih$. Next, Proposition~\ref{prop:cssmcorrectness} implies that the number of calls to $\sjd$ in a run of $\cssm$ is bounded by a constant, the choice of $\delta$  is $\Theta(\eps)$, and the resulting point has total variation $\eps$ from the original distribution $\pi$. Finally, Proposition~\ref{prop:csgcorrectness} shows sampling from a general distribution of the form \eqref{eq:compositesampling} is reducible to one call of $\cssm$, and the requisite oracles are implementable. 	%

\section{Logconcave finite sums}
\label{sec:finitesum}
In this section, we provide our ``base sampler'' for logconcave finite sums as Algorithm~\ref{alg:mrw}, and give its guarantees by proving Theorem~\ref{thm:base_finitesum}. Throughout, fix distribution $\pi$ with density
\begin{align*}
\frac{d\pi}{dx}(x) \propto \exp(-F(x)), \text{ where } F(x)=\frac 1 n \sum_{i \in [n]} f_i(x) \text{ is } \mu\text{-strongly convex,} \\
\text{ and for all } i \in [n],\; f_i \text{ is } L\text{-smooth}.
\end{align*}
We will define $\kappa \defeq \frac{L}{\mu}$, and assume that we have precomputed $x^* \defeq \argmin_{x \in \R^d} \{F(x)\}$. We will also assume explicitly that $\nabla f_i(x^*) = 0$ for all $i \in [n]$ throughout this section (i.e.\ all $f_i$ are minimized at the same point); this is without loss of generality, by a similar argument as in Proposition~\ref{prop:csgcorrectness}.

\begin{algorithm}[ht!]\caption{$\MRW(F, h, x_0,p, K)$}
	\label{alg:mrw}
	\textbf{Input:} $F(x) = \frac{1}{n}\sum_{i\in[n]}f_i(x)$, step size $h > 0$, initial $x_0$, $p \in [0, 1]$, iteration count $K \in \N$
	\begin{algorithmic}[1]
		\For{$0 \le k < K$}
		\State Draw $\xi_k \sim \mathcal{N}(0,\id)$
		\State $y_{k+1} \gets x_k + \sqrt{2h}\xi_k$
		\State Draw $S_k \subseteq [n]$ by including each $i\in S_k$ independently with probability $p$
		\State For each $i\in [n]$, 
		\[
		\gamma_{k}^{(i)}\gets\begin{cases}\frac{1}{p}\left(\sqrt{\exp\left(-\frac{1}{n}f_{i}(y_{k+1})+\frac{1}{n}f_{i}(x_{k})\right)}-1\right)+1 & i \in S_k \\ 1 & i\not\in S_k\end{cases}
		\]
		\State $\gamma_k \gets \prod_{i=1}^{n}\gamma_{k}^{(i)} $, $\tau \sim \text{Unif}[0, 1]$
		\If{$\tau \leq \tfrac{3}{4}\gamma_k$ and $|S_k| \le 2pn$} 
		\State $x_{k+1}\gets y_{k+1}$
		\Else
		\State {$x_{k+1}\gets x_{k}$}
		\EndIf
		\EndFor
		\State \Return $x_K$.
	\end{algorithmic}
\end{algorithm}
Algorithm~\ref{alg:mrw} is the zeroth-order Metropolized random walk of \cite{DwivediCW018} with an efficient, but biased, filter step; the goal of our analysis is to show this bias does not incur significant error.

\subsection{Approximate Metropolis-Hastings}

We first recall the following well-known fact underlying Metropolis-Hastings (MH) filters.
\begin{proposition}\label{prop:reversible}
Consider a random walk on $\R^d$ with proposal distributions $\{\prop_x\}_{x \in \R^d}$ and acceptance probabilities $\{\alpha(x, x')\}_{x, x' \in \R^d}$ conducted as follows: at a current point $x$,
\begin{enumerate}
	\item Draw a point $x' \sim \prop_x$.
	\item Move the random walk to $x'$ with probability $\alpha(x, x')$, else stay at $x$.
\end{enumerate}
Suppose $\prop_x(x') = \prop_{x'}(x)$ for all pairs $x, x' \in \R^d$, and further $\tfrac{d\pi}{dx}(x) \alpha(x, x') = \tfrac{d\pi}{dx}(x') \alpha(x', x)$.
Then, $\pi$ is a stationary distribution for the random walk.
\end{proposition}
\begin{proof}
This follows because the walk satisfies detailed balance (reversibility) with respect to $\pi$.
\end{proof}

We propose an algorithm that applies a variant of the Metropolis-Hastings filter to a Gaussian random walk. Specifically, we define the following algorithm, which we call $\IMRW$.
\begin{definition}[$\IMRW$]\label{def:mmhf}
Consider the following random walk for some step size $h > 0 $: for each iteration $k$ at a current point $x_k \in \R^d$,
\begin{enumerate}
	\item Set $y_{k + 1} \gets x_k + \sqrt{2h}\xi$, where $\xi \sim \Nor(0, \id)$. 
	\item  $x_{k + 1} \gets y_{k + 1}$ with probability $\alpha(x_k, y_{k + 1})$ (otherwise, $x_{k + 1} \gets x_k$), where
	\begin{equation}\label{eq:alphaxy}
	\alpha(x,y) = 
	\begin{cases}
	1 & \sqrt{\frac {\exp(-F(y))}{\exp(-F(x))} }> \frac 4 3,\\ 
	\frac 3 4  \sqrt{\frac {\exp(-F(y))}{\exp(-F(x))} }  &  \frac 3 4 \leq\sqrt{\tfrac {\exp(-F(y))}{\exp(-F(x))} } \leq \frac 4 3, \\
	\frac {\exp(-F(y))}{\exp(-F(x))}  & \sqrt{\frac {\exp(-F(y))}{\exp(-F(x))} }  < \frac 3 4.
	\end{cases}
	\end{equation}
\end{enumerate}
\end{definition}

\begin{lemma}
Distribution $\pi$ with $\tfrac{d\pi}{dx}(x) \propto \exp(-F(x))$ is stationary for $\IMRW$.
\end{lemma}

\begin{proof}
Without loss of generality, assume that $\pi$ has been normalized so that $\tfrac{d\pi}{dx}(x) = \exp(-F(x))$. We apply Proposition~\ref{prop:reversible}, dropping subscripts in the following. It is clear that $\prop_x(y) = \prop_y(x)$ for any $x, y$, so it suffices to check the second condition. When $\tfrac 3 4\le \sqrt{\tfrac {\exp(-F(y))}{\exp(-F(x))} } \le  \tfrac 4 3$, this follows from
\[\frac{d\pi}{dx}(x) \alpha(x, x') = \frac 3 4\sqrt{\exp(-F(x)-F(y))} = \frac{d\pi}{dx}(x') \alpha(x', x).\]
The other case is similar (as it is a standard Metropolis-Hastings filter).
\end{proof}

In Algorithm~\ref{alg:mrw}, we implement an approximate version of the modified MH filter in Definition~\ref{def:mmhf}, where we always assume the pair $x$, $y$ are in the second case of \eqref{eq:alphaxy}. In Lemma~\ref{lem:con_holds}, we show that if a certain boundedness condition holds, then Algorithm~\ref{alg:mrw} approximates $\IMRW$ well. We then show that the output distributions of $\IMRW$ and our Algorithm~\ref{alg:mrw} have small total variation distance in Lemma~\ref{lem:algos_close}.

\begin{lemma}
	\label{lem:con_holds}
	Suppose that in an iteration $0 \le k < K$ of Algorithm~\ref{alg:mrw}, the following three conditions hold for some parameters $R_x$, $C_\xi$, $C_x \in \R_{\ge 0}$:
	\begin{enumerate}
		\item $\norm{x_k - x^*}_2 \le R_x$.
		\item $\norm{\xi_k}_2 \le C_\xi\sqrt{d}$.
		\item For all $i \in [n]$, $|\nabla f_i(x_k)^\top \xi_k| \le C_x\norm{\nabla f_i(x_k)}_2$.
	\end{enumerate}
 Then, for any
 \begin{equation}\label{eq:hbound}
 h\leq \frac 1 {98C_x^2L^2R_x^2 + 7LC_\xi^2 d},\end{equation}
 $\tfrac 3 4 \leq \sqrt{\tfrac {\exp(-F(y_{k + 1}))}{\exp(-F(x_k))} } \leq \tfrac 4 3$.
	Moreover, we have 
	$\E \left[ \gamma_k \right] =  \sqrt{\frac {\exp(-F(y_{k+1}))}{\exp(-F(x_k))}}$, and when $|S_k| \leq 2pn$, $\gamma_k \leq \tfrac 4 3$.
\end{lemma}
\begin{proof}
	We first show $\E \left[ \gamma_k \right] = \sqrt{\frac {\exp(-F(y_{k+1}))}{\exp(-F(x_k))}}$. Since each $i \in S_k$ is generated independently,
	\begin{align*}
	\E \left[ \gamma_k \right] & = \prod_{i\in[n]} \E\Brack{\gamma_k^{(i)}} \\
	 & = \prod_{i\in[n]}\left[  (1-p) + p \left( \frac{1}{p}\left(\sqrt{\exp\left(-\frac{1}{n}f_{i}(y_{k+1})+\frac{1}{n}f_{i}(x_{k})\right)}-1\right)+1\right)\right] \\
	 & = \prod_{i\in[n]} \sqrt{\exp \left(-\frac 1 n  f_{i}(y_{k+1})+\frac{1}{n}f_{i}(x_{k})\right)} =   \sqrt{\frac {\exp(-F(y_{k+1}))}{\exp(-F(x_k))}}.
	\end{align*}
	Next, for any $ i\in[n]$, we lower and upper bound $-f_{i}(y_{k+1})+f_{i}(x_{k}) $. First, 
	\begin{align*}
	 -f_{i}(y_{k+1})+f_{i}(x_{k}) & \leq\nabla f_{i}(x_{k})^{\top}\left(x_{k}-y_{k+1}\right)\\
	 &\leq\sqrt{2h}C_{x}\left\Vert \nabla f_{i}(x_{k})\right\Vert _2 \leq\sqrt{2h}C_{x}LR_{x}.
	\end{align*}
	The first inequality followed from convexity of $f_i$, the second from $y_{k + 1} - x_k = \sqrt{2h}\xi_k$ and our assumed bound, and the third from smoothness and $\nabla f(x^*) = 0$. To show a lower bound,
	\begin{align*}
    f_{i}(y_{k+1})-f_{i}(x_{k}) & \leq  \nabla f_{i}(x_{k})^{\top}\left(y_{k+1}-x_{k}\right)+\frac{L}{2}\left\Vert y_{k+1}-x_{k}\right\Vert_2^{2}\\
	& \leq \sqrt{2h}C_{x}LR_{x} + hL C_\xi^2 d.
	\end{align*}
	The first inequality was smoothness. Repeating this argument for each $i \in [n]$ and averaging,
	\begin{equation}\label{eq:upperboundFF}
	-\sqrt{2h}C_{x}LR_{x} - hL C_\xi^2 d \leq -F(y_{k+1}) + F(x_k) \leq \sqrt{2h}C_{x}LR_{x}.
	\end{equation}
	Then, when $h\leq \tfrac 1 {98C_x^2L^2R_x^2 + 7LC_\xi^2 d}$, 
	$$\frac{3}{4}\leq \sqrt{ \frac {\exp(-F(y_{k+1}))}{\exp(-F(x_k))} }\leq   \frac{4}{3}, \text{ and
	for all } i\in [n],\; -f_{i}(y_{k+1})+f_{i}(x_{k})  \leq \frac 1 {4}.$$
	Thus, we can bound each $	\gamma_{k}^{(i)}$:
	\begin{align*}
	\gamma_{k}^{(i)} & \leq\frac{1}{p}\left(\exp\left(\frac{1}{8n}\right)-1\right)+1\leq 1 + \frac{1}{7pn}.
	\end{align*}
	Finally, when $|S_k| \leq 2pn$, $\gamma_k \leq (1 + \frac{1}{7pn})^{2pn} \leq \tfrac{4}{3}$ as desired.
\end{proof}

\begin{lemma}
	\label{lem:algos_close}
	Draw $x_0 \sim \Nor(x^*, \tfrac{1}{L}\id)$. Let $\pih_K$ be the output distribution of the algorithm of Definition~\ref{def:mmhf} for $K$ steps starting from $x_0$, and let $\pi_K$ be the output distribution of Algorithm~\ref{alg:mrw} starting from $x_0$. For any $\delta \in [0, 1]$, let $p = \tfrac {5\log \frac{12K}{\delta} }{n}$ in Algorithm~\ref{alg:mrw}. There exist 
	\[C_\xi = O\left(1 + \sqrt{\frac{\log \frac{ K}{\delta}}{d}}\right),\; C_x = O\left( \sqrt{\log \frac{ nK}{\delta}}\right),\; \text{ and } R_x = O\left(\sqrt{\frac{d\log \frac{\kappa K} {\delta}}{\mu}}\right),\]
	so that when $h \leq \tfrac 1 {98C_x^2L^2R_x^2 + 7LC_\xi^2 d}$, we have $\tvd{\pi_K}{\pih_K} \leq \delta$.
\end{lemma}

\begin{proof}
   By the coupling definition of total variation, it suffices to upper bound the probability that the algorithms' trajectories, sharing all randomness in proposing points $y_{k + 1}$, differ. This can happen for two reasons: either we used an incorrect filtering step (i.e.\ the pair $(x_k, y_{k + 1})$ did not lie in the second case of \eqref{eq:alphaxy}), or we incorrectly rejected in Line 7 of Algorithm~\ref{alg:mrw} because $|S_k| \ge 2pn$. We bound the error due to either happening over any iteration by $\delta$, yielding the conclusion.
   
   \paragraph{Incorrect filtering.} Consider some iteration $k$. Lemma~\ref{lem:con_holds} shows that as long as its three conditions hold in iteration $k$, we are in the second case of \eqref{eq:alphaxy}, so it suffices to show all conditions hold. By Fact~\ref{fact:subgauss} and as $\xi_k$ is independent of all $\{\nabla f_i(x_k)\}_{i \in [n]}$, with probability at least $1 - \tfrac{\delta}{2K}$, both of the conditions $\norm{\xi_k}_2 \leq C_{\xi} \sqrt{d}$ and\footnote{We recall that the distribution of $v^\top \xi$ for $\xi \sim \Nor(0, \id)$ is the one-dimensional $\Nor(0, \norm{v}_2^2)$.} $|\nabla f_i(x_k)^\top \xi_k| \leq C_{x} \norm{\nabla f_i(x_k)}_2 $ for all $i \in [n]$ hold for some
   \[C_{\xi} = O\left(1 + \sqrt{\frac{\log \frac{ K}{\delta}}{d}}\right),\; C_x = O\left( \sqrt{\log \frac{ nK}{\delta}}\right).\]
   
   Next, $x_0 \sim \mathcal{N}(x^*,\tfrac{1}{L}\id)$ is drawn from a $\kappa^{\frac{d}{2}}$ warm start for $\pi$. By Fact~\ref{fact:subgauss}, we have $\norm{x_0 - x^*}_2\leq R_x$ for $x_0$ drawn from $\pi$ with probability at least $1 - \tfrac{\delta}{4K} \cdot \kappa^{-\frac{d}{2}}$, for some 
   \[R_x = O\left(\sqrt{\frac{d\log \frac{\kappa K} {\delta}}{\mu}}\right).\] Since warmness of the exact algorithm of Definition~\ref{def:mmhf} is monotonic, as long as the trajectories have not differed up to iteration $k$, $\norm{x_k - x^*}_2 \leq R_x$ also holds with probability  $\ge 1 - \tfrac {\delta} {4K}$. Inductively, the total variation error caused by incorrect filtering over $K$ steps is at most $\tfrac{3\delta}{4}$.
   
   \paragraph{Error due to large $|S_k|$.} Supposing all the conditions of Lemma~\ref{lem:con_holds} are satisfied in iteration $k$, we show that with high probability, $\IMRW$ and Algorithm~\ref{alg:mrw} make the same accept or reject decision. By Lemma~\ref{lem:con_holds}, $\IMRW$ \eqref{eq:alphaxy} accepts with probability $\alpha'_k =\tfrac 3 4 \sqrt{\tfrac {\exp(-F(y_{k+1}))}{\exp(-F(x_k))} }$. On the other hand, Algorithm~\ref{alg:mrw} accepts with probability 
   \[\alpha_k = \frac 3 4 \E\Brack{\gamma_k \mid |S_k| \le 2pn} \cdot \Pr[|S_k| \le 2pn].\]
   The total variation between the output distributions is $|\alpha_k - \alpha'_k|$. Further, since by Lemma~\ref{lem:con_holds},
   \begin{align*}\alpha'_k &= \frac 3 4  \E\Brack{\gamma_k} \\
   &= \frac 3 4 \Par{\E\Brack{\gamma_k \mid |S_k| \le 2pn} \cdot \Pr[|S_k| \le 2pn] + \E\Brack{\gamma_k \mid |S_k| > 2pn} \cdot \Pr[|S_k| > 2pn]} \\
   &= \alpha_k + \frac 3 4 \E\Brack{\gamma_k \mid |S_k| > 2pn} \cdot \Pr[|S_k| > 2pn],\end{align*}
   it suffices to upper bound this latter quantity. First, by Lemma~\ref{lem:S_size}, when $p = \tfrac {5\log \frac{12K}{\delta} }{n}$, we have $\Pr[|S_k| > 2pn] \le \tfrac{\delta}{12K}$. Finally, since each $i \in S_k$ is generated independently,
    \begin{align*}
   \E\Brack{\gamma_k \mid |S_k| > 2pn} & \leq \max_{S' : |S'| = 2pn} \E \left[ \prod_{i \in [n] } \gamma_k^{(i)}  \mid S'\subseteq S_k \right] \\
   &  \leq 2 \E\left[ \prod_{i \in [n] \setminus S'} \gamma_k^{(i)}\right] = 2 \sqrt{ \prod_{[n] \setminus S'} \exp\Par{-\frac 1 n f_i(y_{k+1})+ \frac 1 n f_i(x_k))}} \leq 4.
   \end{align*}   
   Here, we used Lemma~\ref{lem:con_holds} applied to the set $S'$, and the upper bound \eqref{eq:upperboundFF} we derived earlier. Combining these calculations shows that the total variation distance incurred in any iteration $k$ due to $|S_k|$ being too large is at most $\tfrac{\delta}{4K}$, so the overall contribution over $K$ steps is at most $\tfrac{\delta}{4}$.
\end{proof}

We used the following helper lemma in our analysis. 
\begin{lemma}
	\label{lem:S_size}
	Let $S\subseteq [n]$ be formed by independently including each $i \in [n]$ with probability $p$. Then,
	\[
	\Pr\left[|S|>2pn\right]\leq\exp\left(-\frac{3pn}{14}\right).
	\]
\end{lemma}
\begin{proof}
	For $i\in[n]$, let $\mathbf{1}_{i\in S}$ be the indicator random
	variable of the event $i\in S$, so $\E\left[\mathbf{1}_{i\in S}\right]=p$
	and 
	\[
	\Var\left[\mathbf{1}_{i\in S}-p\right]=p(1-p)^{2}+(1-p)p^{2}\leq2p.
	\]
	By Bernstein's inequality,	
	\[
	\Pr\left[\sum_{i\in[n]}\mathbf{1}_{i\in S}\geq np+r\right]\leq\exp\left(-\frac{\frac{1}{2}r^{2}}{2np+\frac{1}{3}r}\right).
	\]
	In particular, when $r=pn$, we have the desired conclusion.
\end{proof}

\subsection{Conductance analysis}\label{ssec:conductfs}

We next bound the mixing time of $\IMRW$, using the following result from prior work. We remark that (see Section~\ref{ssec:bug}) in our application, the $\log \beta$ term is non-dominant.

\begin{proposition}[Lemma 1, Lemma 2, \cite{ChenDWY19}]
	\label{prop:mixtime}
	Let a random walk with a $\mu$-strongly logconcave stationary distribution $\pi$ on $x \in \R^d$ have transition distributions $\{\tran_x\}_{x \in \R^d}$. For some $\eps \in [0, 1]$, let convex set $\Omega \subseteq \R^d$ have $\pi(\Omega) \ge 1 - \tfrac{\eps^2}{2\beta^2}$. Let $\pistart$ be a $\beta$-warm start for $\pi$, and let the algorithm be initialized at $x_0 \sim \pistart$. Suppose for any $x, x' \in \Omega$ with $\norm{x - x'}_2 \le \Delta$,
	\begin{equation}\label{eq:tvdreq}\tvd{\tran_x}{\tran_{x'}} \le \frac 7 8.\end{equation}
	Then,  the random walk mixes to total variation distance within $\eps$ of $\pi$ in $O(\log \beta + \tfrac{1}{\Delta^2\mu}\log \tfrac{\log\beta}{\eps})$ iterations.
\end{proposition}

 Consider an iteration of $\IMRW$ from $x_k = x$. Let $\prop_x$ be the density of $y_{k + 1}$, and let $\tran_x$ be the density of $x_{k + 1}$ after filtering. Define a convex set $\Omega \subseteq \R^d$ parameterized by $R_\Omega \in \R_{\ge 0}$: 
\begin{align*}
\Omega = \{x \in \R^d: \norm{x-x^*}_2 \leq R_{\Omega} \}.
\end{align*}
We show that for two close points $x, x' \subseteq \Omega $, the total variation between $\tran_x$ and $\tran_{x'}$ is small.

\begin{lemma}
	\label{lem:tv_dist_finite}
	For some $h = O (\tfrac{1}{L^{2}R_{\Omega}^{2} + Ld})$ and $x, x' \subseteq \Omega$ with $\norm{x - x'}_2\leq \tfrac 1 8 \sqrt{h}$,  $\tvd{\tran_x}{\tran_{x'}} \leq \tfrac 7 8$.
\end{lemma}

\begin{proof}
	By the triangle inequality of total variation distance, 
	$$\tvd{\tran_x}{\tran_{x'}} \leq \tvd{\tran_x}{\prop_{x}}  + \tvd{\prop_x}{\prop_{x'}} + \tvd{\tran_{x'}}{\prop_{x'}}.$$
	First, by Pinsker's inequality and the KL divergence between Gaussian distributions, 
	$$
	\left\Vert \prop_{x}-\prop_{x'}\right\Vert _{\textup{TV}}\leq\sqrt{2\textup{KL}(\prop_{x}||\prop_{x'})}=\frac{\left\Vert x-x'\right\Vert_2 }{\sqrt{2h}}.
	$$
	When $\norm{x - x'}_2 \leq \tfrac 1 8 \sqrt{h}$, $ \tvd{\prop_x}{\prop_{x'}}  \leq \tfrac 1 8$. 
	Next, we bound $\tvd{\tran_x}{\prop_{x}} $: by a standard calculation (e.g.\ Lemma D.1 of \cite{LeeST20}), we have
	\begin{align*}
	\tvd{\tran_x}{\prop_{x}}  
	& =1-\frac 3 4 \E_{\xi_{k+1}}\left[\sqrt{\frac{\exp\left(-F(y_{k+1})\right)}{\exp\left(-F(x_{k})\right)}}\right].
	\end{align*}
    We show that $\tvd{\tran_x}{\prop_{x}} \leq\frac{3}{8}$. It suffices to show that
		$
		\E_{\xi_{k+1}}\left[\sqrt{\exp\left(-F(y_{k+1})+F(x_{k})\right)}\right]\geq\frac{5}{6}.
		$
		
	Since $\frac{15}{16}\sqrt{\exp\left(-\frac{1}{16}\right)}\geq\frac{5}{6}$,
	it suffices to show that with probability at least $\tfrac{15}{16}$ over the
	randomness of $\xi_{k+1}$, $-F(y_{k+1})+F(x_{k})\geq-\frac{1}{16}$. As $\xi_{k+1}\sim\mathcal{N}(0,\mathbb{I}_{d})$, by applying Fact~\ref{fact:subgauss} twice,
		\begin{equation}\label{eq:xiconc}\begin{aligned}
		\Pr\left[\text{\ensuremath{\left\Vert \xi_{k+1}\right\Vert }}_2^{2}>36d\right]\leq\exp(-4)\leq\frac{1}{32},\\
		\Pr\left[\left|\nabla F(x_{k})^{\top}\xi_{k+1}\right|^{2}\geq36\norm{\nabla F(x_{k})}_2 ^{2}\right]\leq\frac{1}{32}.
		\end{aligned}\end{equation}
		
	We upper bound the term $F(y_{k+1})-F(x_{k})$ by smoothness and Cauchy-Schwarz:
	\begin{align*}
		F(y_{k+1})-F(x_{k}) & \leq\nabla F(x_{k})^{\top}\left(y_{k+1}-x_{k}\right)+\frac{L}{2}\left\Vert y_{k+1}-x_{k}\right\Vert_2 ^{2}\\
		& \leq\sqrt{2h}\left|\nabla F(x_{k})^{\top}\xi_{k+1}\right|+hL\text{\ensuremath{\left\Vert \xi_{k+1}\right\Vert }}_2^{2}.
	\end{align*}
		
	Then, since $\left\Vert \nabla F(x_{k})\right\Vert \leq LR_{\Omega}$
	when $x\in \Omega$, it is enough to choose $h = O (\tfrac{1}{L^{2}R_{\Omega}^{2}+Ld})$ so that
	$$
		-F(y_{k+1})+F(x_{k})\geq-\frac{1}{16},
	$$
	as long as the events of \eqref{eq:xiconc} hold, which occurs with probability at least $\frac{15}{16}$. Similarly, we can show that $ 	\tvd{\tran_{x'}}{\prop_{x'}} \leq \tfrac 3 8 $.
	Combining the three bounds, we have the desired conclusion.
		
\end{proof}

\restatezerofs*
\begin{proof}
First, $\mathcal{N}(x^*,\tfrac{1}{L}\id)$ yields a $\beta = \kappa^{\frac{d}{2}}$-warm start for $\pi$ (see e.g. \cite{DwivediCW018}). For this value of $\beta$, by Fact~\ref{fact:subgauss} it suffices to choose 
\[R_{\Omega} = \Theta\Par{\sqrt{\frac{d\log\frac{\kappa}{\eps}}{\mu}}}\]
for $\pi(\Omega) \ge 1 - \tfrac{\eps^2}{2\beta^2}$. Letting $\delta = \tfrac \eps 2$, we will choose the step size $h$ and iteration count $K$ so that
\begin{align*}\frac 1 h = \Theta\Par{L\kappa d\log^2\frac{n\kappa d}{\eps}},\; K = \Theta\Par{\kappa^2 d\log^3 \frac{n\kappa d}{\eps}}\end{align*}
have constants compatible with Lemma~\ref{lem:algos_close}.
Note that this choice of $h$ is also sufficiently small to apply Lemma~\ref{lem:tv_dist_finite} for our choice of $R_\Omega$. By applying Proposition~\ref{prop:mixtime} to the algorithm of Definition~\ref{def:mmhf}, and using the bound from Lemma~\ref{lem:tv_dist_finite}, in $K$ iterations $\IMRW$ will mix to total variation distance $\delta$ to $\pi$. Furthermore, applying Lemma~\ref{lem:algos_close}, we conclude that Algorithm~\ref{alg:mrw} has total variation distance at most $2\delta = \eps$ from $\pi$.

It remains to bound the oracle complexity of Algorithm~\ref{alg:mrw}. Note in every iteration, we never compute more than $4pn$ values of $\{f_i\}_{i \in [n]}$, since we always reject if $|S_k| \ge 2pn$, and we only compute values for indices in $S_k$. For the value of $p$ in Lemma~\ref{lem:algos_close}, this amounts to $O(\log\tfrac{n\kappa d}{\eps})$ value queries.
\end{proof} 	
	\subsection*{Acknowledgments}
	YL and RS are supported by NSF awards CCF-1749609, CCF-1740551, DMS-1839116, and DMS-2023166, a Microsoft Research Faculty Fellowship, a Sloan Research Fellowship, and a Packard Fellowship. KT is supported by NSF CAREER Award CCF-1844855 and a PayPal research gift.
	
	RS and KT would like to thank Sinho Chewi for his extremely generous help, in particular insightful conversations which led to our discovery of the gap in Section 3, as well as his suggested fix.

	\bibliographystyle{alpha}	
	\bibliography{structured-sampling}

\newcommand{\etalchar}[1]{$^{#1}$}
\begin{thebibliography}{MMW{\etalchar{+}}19}

\bibitem[AH16]{AbernethyH16}
Jacob~D. Abernethy and Elad Hazan.
\newblock Faster convex optimization: Simulated annealing with an efficient
  universal barrier.
\newblock In {\em Proceedings of the 33nd International Conference on Machine
  Learning, {ICML} 2016, New York City, NY, USA, June 19-24, 2016}, pages
  2520--2528, 2016.

\bibitem[All17]{Allen-Zhu17}
Zeyuan Allen{-}Zhu.
\newblock Katyusha: The first direct acceleration of stochastic gradient
  methods.
\newblock {\em J. Mach. Learn. Res.}, 18:221:1--221:51, 2017.

\bibitem[BCM{\etalchar{+}}18]{barkhagen2018stochastic}
M~Barkhagen, NH~Chau, {\'E}~Moulines, M~R{\'a}sonyi, S~Sabanis, and Y~Zhang.
\newblock On stochastic gradient langevin dynamics with dependent data streams
  in the logconcave case.
\newblock {\em arXiv preprint arXiv:1812.02709}, 2018.

\bibitem[BDMP17]{BrosseDMP17}
Nicolas Brosse, Alain Durmus, Eric Moulines, and Marcelo Pereyra.
\newblock Sampling from a log-concave distribution with compact support with
  proximal langevin monte carlo.
\newblock In {\em Proceedings of the 30th Conference on Learning Theory, {COLT}
  2017, Amsterdam, The Netherlands, 7-10 July 2017}, pages 319--342, 2017.

\bibitem[BEL18]{BubeckEL18}
S{\'{e}}bastien Bubeck, Ronen Eldan, and Joseph Lehec.
\newblock Sampling from a log-concave distribution with projected langevin
  monte carlo.
\newblock {\em Discret. Comput. Geom.}, 59(4):757--783, 2018.

\bibitem[Ber18]{Bernton18}
Espen Bernton.
\newblock Langevin monte carlo and {JKO} splitting.
\newblock In {\em Conference On Learning Theory, {COLT} 2018, Stockholm,
  Sweden, 6-9 July 2018}, pages 1777--1798, 2018.

\bibitem[BFFN19]{baker2019control}
Jack Baker, Paul Fearnhead, Emily~B Fox, and Christopher Nemeth.
\newblock Control variates for stochastic gradient mcmc.
\newblock {\em Statistics and Computing}, 29(3):599--615, 2019.

\bibitem[BFR{\etalchar{+}}19]{bierkens2019zig}
Joris Bierkens, Paul Fearnhead, Gareth Roberts, et~al.
\newblock The zig-zag process and super-efficient sampling for bayesian
  analysis of big data.
\newblock {\em The Annals of Statistics}, 47(3):1288--1320, 2019.

\bibitem[BT09]{BeckT09}
Amir Beck and Marc Teboulle.
\newblock A fast iterative shrinkage-thresholding algorithm for linear inverse
  problems.
\newblock {\em {SIAM} J. Imaging Sciences}, 2(1):183--202, 2009.

\bibitem[BV04]{BertsimasV04}
Dimitris Bertsimas and Santosh~S. Vempala.
\newblock Solving convex programs by random walks.
\newblock {\em J. {ACM}}, 51(4):540--556, 2004.

\bibitem[CCBJ18]{ChengCBJ18}
Xiang Cheng, Niladri~S. Chatterji, Peter~L. Bartlett, and Michael~I. Jordan.
\newblock Underdamped langevin {MCMC:} {A} non-asymptotic analysis.
\newblock In {\em Conference On Learning Theory, {COLT} 2018, Stockholm,
  Sweden, 6-9 July 2018}, pages 300--323, 2018.

\bibitem[CDWY19]{ChenDWY19}
Yuansi Chen, Raaz Dwivedi, Martin~J. Wainwright, and Bin Yu.
\newblock Fast mixing of metropolized hamiltonian monte carlo: Benefits of
  multi-step gradients.
\newblock {\em CoRR}, abs/1905.12247, 2019.

\bibitem[CFM{\etalchar{+}}18]{chatterji2018theory}
Niladri Chatterji, Nicolas Flammarion, Yian Ma, Peter Bartlett, and Michael
  Jordan.
\newblock On the theory of variance reduction for stochastic gradient monte
  carlo.
\newblock In {\em International Conference on Machine Learning}, pages
  764--773, 2018.

\bibitem[CV15]{cousins2015bypassing}
Benjamin Cousins and Santosh Vempala.
\newblock Bypassing kls: Gaussian cooling and an o\^{}*(n3) volume algorithm.
\newblock In {\em Proceedings of the forty-seventh annual ACM symposium on
  Theory of computing}, pages 539--548, 2015.

\bibitem[CV18]{CousinsV18}
Ben Cousins and Santosh~S. Vempala.
\newblock Gaussian cooling and o\({}^{\mbox{*(n\({}^{\mbox{3)}}\)}}\)
  algorithms for volume and gaussian volume.
\newblock {\em {SIAM} J. Comput.}, 47(3):1237--1273, 2018.

\bibitem[CV19]{ChenV19}
Zongchen Chen and Santosh~S. Vempala.
\newblock Optimal convergence rate of hamiltonian monte carlo for strongly
  logconcave distributions.
\newblock In {\em Approximation, Randomization, and Combinatorial Optimization.
  Algorithms and Techniques, {APPROX/RANDOM} 2019, September 20-22, 2019,
  Massachusetts Institute of Technology, Cambridge, MA, {USA}}, pages
  64:1--64:12, 2019.

\bibitem[CWZ{\etalchar{+}}17]{chen2017convergence}
Changyou Chen, Wenlin Wang, Yizhe Zhang, Qinliang Su, and Lawrence Carin.
\newblock A convergence analysis for a class of practical variance-reduction
  stochastic gradient mcmc.
\newblock {\em arXiv preprint arXiv:1709.01180}, 2017.

\bibitem[Dal17]{dalalyan2017theoretical}
Arnak~S Dalalyan.
\newblock Theoretical guarantees for approximate sampling from smooth and
  log-concave densities.
\newblock {\em Journal of the Royal Statistical Society: Series B (Statistical
  Methodology)}, 3(79):651--676, 2017.

\bibitem[DBL14]{DefazioBL14}
Aaron Defazio, Francis~R. Bach, and Simon Lacoste{-}Julien.
\newblock {SAGA:} {A} fast incremental gradient method with support for
  non-strongly convex composite objectives.
\newblock In {\em Advances in Neural Information Processing Systems 27: Annual
  Conference on Neural Information Processing Systems 2014, December 8-13 2014,
  Montreal, Quebec, Canada}, pages 1646--1654, 2014.

\bibitem[DCWY18]{DwivediCW018}
Raaz Dwivedi, Yuansi Chen, Martin~J. Wainwright, and Bin Yu.
\newblock Log-concave sampling: Metropolis-hastings algorithms are fast!
\newblock In {\em Conference On Learning Theory, {COLT} 2018, Stockholm,
  Sweden, 6-9 July 2018}, pages 793--797, 2018.

\bibitem[DFK91]{DyerFK91}
Martin~E. Dyer, Alan~M. Frieze, and Ravi Kannan.
\newblock A random polynomial time algorithm for approximating the volume of
  convex bodies.
\newblock {\em J. {ACM}}, 38(1):1--17, 1991.

\bibitem[DK19]{dalalyan2019user}
Arnak~S Dalalyan and Avetik Karagulyan.
\newblock User-friendly guarantees for the langevin monte carlo with inaccurate
  gradient.
\newblock {\em Stochastic Processes and their Applications},
  129(12):5278--5311, 2019.

\bibitem[DM19a]{DurmusM19}
Alain Durmus and \'Eric Moulines.
\newblock High-dimensional bayesian inference via the unadjusted langevin
  algorithm.
\newblock {\em Bernoulli}, 25(4A):2854--2882, 2019.

\bibitem[DM19b]{durmus2019high}
Alain Durmus and Eric Moulines.
\newblock High-dimensional bayesian inference via the unadjusted langevin
  algorithm.
\newblock {\em Bernoulli}, 25(4A):2854--2882, 2019.

\bibitem[DMM19]{DurmusMM19}
Alain Durmus, Szymon Majewski, and Blazej Miasojedow.
\newblock Analysis of langevin monte carlo via convex optimization.
\newblock {\em J. Mach. Learn. Res.}, 20:73:1--73:46, 2019.

\bibitem[DR18]{DalalyanR18}
Arnak~S. Dalalyan and Lionel Riou{-}Durand.
\newblock On sampling from a log-concave density using kinetic langevin
  diffusions.
\newblock {\em CoRR}, abs/1807.09382, 2018.

\bibitem[DRW{\etalchar{+}}16]{dubey2016variance}
Kumar~Avinava Dubey, Sashank~J Reddi, Sinead~A Williamson, Barnabas Poczos,
  Alexander~J Smola, and Eric~P Xing.
\newblock Variance reduction in stochastic gradient langevin dynamics.
\newblock In {\em Advances in neural information processing systems}, pages
  1154--1162, 2016.

\bibitem[DSM{\etalchar{+}}16]{durmus2016stochastic}
Alain Durmus, Umut Simsekli, Eric Moulines, Roland Badeau, and Ga{\"e}l
  Richard.
\newblock Stochastic gradient richardson-romberg markov chain monte carlo.
\newblock In {\em Advances in Neural Information Processing Systems}, pages
  2047--2055, 2016.

\bibitem[FGKS15]{FrostigGKS15}
Roy Frostig, Rong Ge, Sham~M. Kakade, and Aaron Sidford.
\newblock Un-regularizing: approximate proximal point and faster stochastic
  algorithms for empirical risk minimization.
\newblock In {\em Proceedings of the 32nd International Conference on Machine
  Learning, {ICML} 2015, Lille, France, 6-11 July 2015}, pages 2540--2548,
  2015.

\bibitem[GGZ18]{gao2018global}
Xuefeng Gao, Mert G{\"u}rb{\"u}zbalaban, and Lingjiong Zhu.
\newblock Global convergence of stochastic gradient hamiltonian monte carlo for
  non-convex stochastic optimization: Non-asymptotic performance bounds and
  momentum-based acceleration.
\newblock {\em arXiv preprint arXiv:1809.04618}, 2018.

\bibitem[Gul92]{Guler92}
Osman Guler.
\newblock New proximal point algorithms for convex minimization.
\newblock {\em SIAM Journal on Optimization}, 2(4):649--664, 1992.

\bibitem[Har04]{Harge04}
Gilles Harg\'e.
\newblock A convex/log-concave correlation inequality for gaussian measure and
  an application to abstract wiener spaces.
\newblock {\em Probability theory and related fields}, 130(3):415--440, 2004.

\bibitem[HKRC18]{HsiehKRC18}
Ya{-}Ping Hsieh, Ali Kavis, Paul Rolland, and Volkan Cevher.
\newblock Mirrored langevin dynamics.
\newblock In {\em Advances in Neural Information Processing Systems 31: Annual
  Conference on Neural Information Processing Systems 2018, NeurIPS 2018, 3-8
  December 2018, Montr{\'{e}}al, Canada}, pages 2883--2892, 2018.

\bibitem[JZ13]{Johnson013}
Rie Johnson and Tong Zhang.
\newblock Accelerating stochastic gradient descent using predictive variance
  reduction.
\newblock In {\em Advances in Neural Information Processing Systems 26: 27th
  Annual Conference on Neural Information Processing Systems 2013. Proceedings
  of a meeting held December 5-8, 2013, Lake Tahoe, Nevada, United States},
  pages 315--323, 2013.

\bibitem[Kra40]{Kramers40}
Hendrik~Anthony Kramers.
\newblock Brownian motion in a field of force and the diffusion model of
  chemical reactions.
\newblock {\em Physica}, 7(4):284--304, 1940.

\bibitem[Lee18]{Lee18}
Yin~Tat Lee.
\newblock Lecture 8: Stochastic methods and applications.
\newblock Class notes, UW CSE 599: Interplay between Convex Optimization and
  Geometry, 2018.

\bibitem[LK99]{LovaszK99}
L{\'{a}}szl{\'{o}} Lov{\'{a}}sz and Ravi Kannan.
\newblock Faster mixing via average conductance.
\newblock In {\em Proceedings of the Thirty-First Annual {ACM} Symposium on
  Theory of Computing, May 1-4, 1999, Atlanta, Georgia, {USA}}, pages 282--287,
  1999.

\bibitem[LMH15]{LinMH15}
Hongzhou Lin, Julien Mairal, and Za{\"{\i}}d Harchaoui.
\newblock A universal catalyst for first-order optimization.
\newblock In {\em Advances in Neural Information Processing Systems 28: Annual
  Conference on Neural Information Processing Systems 2015, December 7-12,
  2015, Montreal, Quebec, Canada}, pages 3384--3392, 2015.

\bibitem[LPW09]{LevinPW09}
David~Asher Levin, Yuval Peres, and Elizabeth Wilmer.
\newblock {\em Markov Chains and Mixing Times}.
\newblock American Mathematical Society, 2009.

\bibitem[LST20]{LeeST20}
Yin~Tat Lee, Ruoqi Shen, and Kevin Tian.
\newblock Logsmooth gradient concentration and tighter runtimes for
  metropolized hamiltonian monte carlo.
\newblock In {\em Conference on Learning Theory, {COLT} 2020}, 2020.

\bibitem[LSV18]{LeeSV18}
Yin~Tat Lee, Zhao Song, and Santosh~S. Vempala.
\newblock Algorithmic theory of odes and sampling from well-conditioned
  logconcave densities.
\newblock {\em CoRR}, abs/1812.06243, 2018.

\bibitem[LV06a]{lovasz2006simulated}
L{\'a}szl{\'o} Lov{\'a}sz and Santosh Vempala.
\newblock Simulated annealing in convex bodies and an o*(n4) volume algorithm.
\newblock {\em Journal of Computer and System Sciences}, 72(2):392--417, 2006.

\bibitem[LV06b]{LovaszV06b}
L{\'{a}}szl{\'{o}} Lov{\'{a}}sz and Santosh~S. Vempala.
\newblock Fast algorithms for logconcave functions: Sampling, rounding,
  integration and optimization.
\newblock In {\em 47th Annual {IEEE} Symposium on Foundations of Computer
  Science {(FOCS} 2006), 21-24 October 2006, Berkeley, California, USA,
  Proceedings}, pages 57--68, 2006.

\bibitem[LV06c]{LovaszV06a}
L{\'{a}}szl{\'{o}} Lov{\'{a}}sz and Santosh~S. Vempala.
\newblock Hit-and-run from a corner.
\newblock {\em {SIAM} J. Comput.}, 35(4):985--1005, 2006.

\bibitem[MFWB19]{MouFWB19}
Wenlong Mou, Nicolas Flammarion, Martin~J. Wainwright, and Peter~L. Bartlett.
\newblock An efficient sampling algorithm for non-smooth composite potentials.
\newblock {\em CoRR}, abs/1910.00551, 2019.

\bibitem[MMW{\etalchar{+}}19]{MouMWBJ19}
Wenlong Mou, Yi{-}An Ma, Martin~J. Wainwright, Peter~L. Bartlett, and
  Michael~I. Jordan.
\newblock High-order langevin diffusion yields an accelerated {MCMC} algorithm.
\newblock {\em CoRR}, abs/1908.10859, 2019.

\bibitem[NDH{\etalchar{+}}17]{nagapetyan2017true}
Tigran Nagapetyan, Andrew~B Duncan, Leonard Hasenclever, Sebastian~J Vollmer,
  Lukasz Szpruch, and Konstantinos Zygalakis.
\newblock The true cost of stochastic gradient langevin dynamics.
\newblock {\em arXiv preprint arXiv:1706.02692}, 2017.

\bibitem[Nea11]{Neal11}
Radford~M Neal.
\newblock Mcmc using hamiltonian dynamics.
\newblock {\em Handbook of Markov chain Monte Carlo}, 2(11):2, 2011.

\bibitem[Nes83]{Nesterov83}
Yurii Nesterov.
\newblock A method for solving a convex programming problem with convergence
  rate $o(1/k^2)$.
\newblock {\em Doklady AN SSSR}, 269:543--547, 1983.

\bibitem[NF19]{nemeth2019stochastic}
Christopher Nemeth and Paul Fearnhead.
\newblock Stochastic gradient markov chain monte carlo.
\newblock {\em arXiv preprint arXiv:1907.06986}, 2019.

\bibitem[OV00]{otto2000generalization}
Felix Otto and C{\'e}dric Villani.
\newblock Generalization of an inequality by talagrand and links with the
  logarithmic sobolev inequality.
\newblock {\em Journal of Functional Analysis}, 173(2):361--400, 2000.

\bibitem[PB14]{ParikhB14}
Neal Parikh and Stephen~P. Boyd.
\newblock Proximal algorithms.
\newblock {\em Found. Trends Optim.}, 1(3):127--239, 2014.

\bibitem[Per16]{Pereyra16}
Marcelo Pereyra.
\newblock Proximal markov chain monte carlo algorithms.
\newblock {\em Stat. Comput.}, 26(4):745--760, 2016.

\bibitem[Roc76]{Rockafellar76}
R~Tyrell Rockafellar.
\newblock Monotone operators and the proximal point algorithm.
\newblock {\em SIAM journal on control and optimization}, 14(5):877--898, 1976.

\bibitem[SKR19]{salim2019stochastic}
Adil Salim, Dmitry Koralev, and Peter Richt{\'a}rik.
\newblock Stochastic proximal langevin algorithm: Potential splitting and
  nonasymptotic rates.
\newblock In {\em Advances in Neural Information Processing Systems}, pages
  6653--6664, 2019.

\bibitem[SL19]{shen2019randomized}
Ruoqi Shen and Yin~Tat Lee.
\newblock The randomized midpoint method for log-concave sampling.
\newblock In {\em Advances in Neural Information Processing Systems}, pages
  2100--2111, 2019.

\bibitem[SRB17]{SchmidtRB17}
Mark Schmidt, Nicolas~Le Roux, and Francis~R. Bach.
\newblock Minimizing finite sums with the stochastic average gradient.
\newblock {\em Math. Program.}, 162(1-2):83--112, 2017.

\bibitem[STL20]{ShenTL20}
Ruoqi Shen, Kevin Tian, and Yin~Tat Lee.
\newblock Composite logconcave sampling with a restricted gaussian oracle.
\newblock {\em CoRR}, abs/2006.05976, 2020.

\bibitem[Tib96]{Tibshirani96}
Robert Tibshirani.
\newblock Regression shrinkage and selection via the lasso.
\newblock {\em Journal of the Royal Statistical Society, Series B
  (Methodological)}, 58(1):267--288, 1996.

\bibitem[Vem05]{Vempala05}
Santosh Vempala.
\newblock Geometric random walks: A survey.
\newblock {\em MSRI Combinatorial and Computational Geometry}, 52:573--612,
  2005.

\bibitem[Wib19]{Wibisono19}
Andre Wibisono.
\newblock Proximal langevin algorithm: Rapid convergence under isoperimetry.
\newblock {\em CoRR}, abs/1911.01469, 2019.

\bibitem[WS16]{WoodworthS16}
Blake~E. Woodworth and Nati Srebro.
\newblock Tight complexity bounds for optimizing composite objectives.
\newblock In {\em Advances in Neural Information Processing Systems 29: Annual
  Conference on Neural Information Processing Systems 2016, December 5-10,
  2016, Barcelona, Spain}, pages 3639--3647, 2016.

\bibitem[WT11]{welling2011bayesian}
Max Welling and Yee~W Teh.
\newblock Bayesian learning via stochastic gradient langevin dynamics.
\newblock In {\em Proceedings of the 28th international conference on machine
  learning (ICML-11)}, pages 681--688, 2011.

\bibitem[ZH05]{ZouH05}
Hui Zou and Trevor Hastie.
\newblock Regularization and variable section via the elastic net.
\newblock {\em Journal of the Royal Statistical Society, Series B
  (Methodological)}, 67(2):301--320, 2005.

\bibitem[ZXG18]{zou2018subsampled}
Difan Zou, Pan Xu, and Quanquan Gu.
\newblock Subsampled stochastic variance-reduced gradient langevin dynamics.
\newblock In {\em International Conference on Uncertainty in Artificial
  Intelligence}, 2018.

\end{thebibliography}
	\newpage
	
	\addtocontents{toc}{\protect\setcounter{tocdepth}{1}}
	\begin{appendix}

\section{Discussion of inexactness tolerance}
\label{app:xstar}

We briefly discuss the tolerance of our algorithm to approximation error in two places: computation of minimizers, and implementation of RGOs in the methods of Sections~\ref{sec:framework} and~\ref{sec:composite}. 

\paragraph{Inexact minimization.} For all function classes considered in this work, there exist efficient optimization methods converging to a minimizer with logarithmic dependence on the target accuracy. 

Specifically, for negative log-densities with condition number $\kappa$, accelerated gradient descent \cite{Nesterov83} converges at a rate $O(\sqrt{\kappa})$ with logarithmic dependence on initial error and target accuracy (we implicitly assumed in stating our runtimes that one can attain initial error polynomial in problem parameters for negative log-densities; otherwise, there is additional logarithmic overhead in the quality of the initial point to optimization procedures). For composite functions $\fwc + \fcomp$ where $\fwc$ has condition number $\kappa$, the FISTA method of \cite{BeckT09} converges at the same rate with each iteration querying $\nabla \fwc$ and a proximal oracle for $\fcomp$ once; typically, access to a proximal oracle is a weaker assumption than access to a restricted Gaussian oracle, so this is not restrictive. Finally, for minimizing finite sums with condition number $\kappa$, the algorithm of \cite{Allen-Zhu17} obtains a convergence rate linearly dependent on $n + \sqrt{n\kappa} \le n + \kappa$; alternatively, \cite{Johnson013} has a dependence on $n + \kappa$. In all our final runtimes, these optimization rates do not constitute the bottleneck for oracle complexities.

The only additional difficulty our algorithms may present is if the function requiring minimization, say of the form $\fcomp(x) + \tfrac{1}{2\eta}\norm{x - y}_2^2$ for some $y \in \R^d$ where we have computed the minimizer $x^*$ to $\fcomp$, has $\norm{y - x^*}_2^2$ very large (so the initial function error is bad). However, in all our settings $y$ is drawn from a distribution with sub-Gaussian tails, so $\norm{y - x^*}_2^2$ decays exponentially (whereas the complexity of first-order methods increases only logarithmically), negligibly affecting the expected oracle query complexity for our methods.

Finally, by solving the relevant optimization problems to high accuracy as a subroutine in each of our methods, and adjusting various distance bounds to the minimizer by constants (e.g.\ by expanding the radius in the definition of the sets $\Omega$ in Algorithm~\ref{alg:cssm} and Section~\ref{ssec:conductfs}), this accomodates tolerance to inexact minimization and only affects all bounds throughout the paper by constants. The only other place that $x^*$ is used in our algorithms is in initializing warm starts; tolerance to inexactness in our warmness calculations follows essentially identically to Section 3.2.1 of \cite{DwivediCW018}.

\paragraph{Inexact oracle implementation.} Our algorithms based on restricted Gaussian oracle access are tolerant to total variation error inverse polynomial in problem parameters for the restricted Gaussian oracle for $g$. We discussed this at the end of Section~\ref{sec:framework}, in the case of RGO use for our reduction framework. To see this in the case of the composite sampler in Section~\ref{sec:composite}, we pessimistically handled the case where the sampler $\yor$ for a quadratic restriction of $f$ resulted in total variation error in the proof of Proposition~\ref{prop:sjdguarantee}, assuming that the error was incurred in every iteration. By accounting for similar amounts of error in calls to $\oracle$ (on the order of $\tfrac{\eps}{T}$, where $T$ is the number of times an RGO was used), the bounds in our algorithm are only affected by constants.

\section{Deferred proofs from Section~\ref{sec:composite}}
\label{app:deferred}

\subsection{Deferred proofs from Section~\ref{ssec:outerloop}}

\subsubsection{Approximate rejection sampling}
\label{sssec:approxreject}
We first define the rejection sampling framework we will use, and prove various properties.

\begin{definition}[Approximate rejection sampling]
	\label{def:reject_approx}
	Let $\pi$ be a distribution, with $\tfrac{d\pi}{dx}(x) \propto p(x)$. Suppose set $\Omega$ has $\pi(\Omega) = 1 - \eps'$, and distribution $\pih$ with $\frac{d\pih}{dx}(x) \propto \hp(x)$ has for some $C \ge 1$,
	\[\frac{p(x)}{\hp(x)} \le C \text{ for all } x \in \Omega, \text{ and } \frac{\int \hp(x) dx}{\int p(x) dx} \le 1.\]
	Suppose there is an algorithm $\alg$ which draws samples from a distribution $\pih'$, such that $\tvd{\pih'}{\pih} \le 1 - \delta$. We call the following scheme approximate rejection sampling: repeat independent runs of the following procedure until a point is outputted.
	\begin{enumerate}
		\item Draw $x$ via $\alg$ until $x \in \Omega$.
		\item With probability $\tfrac{p(x)}{C\hat{p}(x)}$, output $x$.
	\end{enumerate}
\end{definition}

\begin{lemma}\label{lem:reject_approx_proof}
	Consider an approximate rejection sampling scheme with relevant parameters defined as in Definition~\ref{def:reject_approx}, with $2\delta \le \tfrac{1 - \eps'}{C}$. The algorithm terminates in at most
	\begin{equation}\label{eq:algcalls}\frac{1}{\frac{1 - \eps'}{C} - 2\delta}\end{equation}
	calls to $\alg$ in expectation, and outputs a point from a distribution $\pi'$ with $\tvd{\pi'}{\pi} \le \eps' + \frac{2\delta C}{1 - \eps'}$.
\end{lemma}
\begin{proof}
	Define for notational simplicity normalization constants $Z \defeq \int p(x) dx$ and $\hat{Z} \defeq \int \hp(x) dx$. First, we bound the probability any particular call to $\alg$ returns in the scheme:
	\begin{equation}\label{eq:returnprob}\begin{aligned}\int_{x \in \Omega} \frac{p(x)}{C\hp(x)}d\pih'(x) &\ge \int_{x \in \Omega} \frac{p(x)}{C\hp(x)}d\pih(x) - \left|\int_{x \in \Omega}\frac{p(x)}{C\hp(x)}(d\pih'(x) - d\pih(x))\right|\\
	&= \int_{x \in \Omega} \frac{Z}{C\hat{Z}} d\pi(x) - \left|\int_{x \in \Omega}\frac{p(x)}{C\hp(x)}(d\pih'(x) - d\pih(x))\right|\\
	&\ge \frac{1 - \eps'}{C} - \int_{x \in \Omega}|d\pih'(x) - d\pih(x)| \ge \frac{1 - \eps'}{C} - 2\delta.\end{aligned}\end{equation}
	The second line followed by the definitions of $Z$ and $\hat{Z}$, and the third followed by triangle inequality, the assumed lower bound on $Z/\hat{Z}$, and the total variation distance between $\pih'$ and $\pih$. By linearity of expectation and independence, this proves the first claim.
	
	Next, we claim the output distribution is close in total variation distance to the conditional distribution of $\pi$ restricted to $\Omega$. The derivation of \eqref{eq:returnprob} implies
	\begin{equation}\label{eq:rationotbad}\begin{aligned} \int_{x \in \Omega} \frac{p(x)}{C\hp(x)}d\pih(x)\ge \frac{1 - \eps'}{C},\; \left|\int_{x \in \Omega}\frac{p(x)}{C\hp(x)}(d\pih'(x) - d\pih(x))\right| \le 2\delta,
	\\
	\implies 1 - \frac{2\delta C}{1 - \eps'} \le \frac{\int_{x \in \Omega} \frac{p(x)}{C\hp(x)}d\pih'(x)}{\int_{x \in \Omega} \frac{p(x)}{C\hp(x)}d\pih(x)} \le 1 + \frac{2\delta C}{1 - \eps'}.\end{aligned}\end{equation}
	Thus, the total variation of the true output distribution from $\pi$ restricted to $\Omega$ is
	\begin{align*}
	\half \int_{x \in \Omega} \left|\frac{d\pi(x)}{1 - \eps'} - \frac{\frac{p(x)}{C\hp(x)}d\pih'(x)}{\int_{x \in \Omega} \frac{p(x)}{C\hp(x)}d\pih'(x)}\right| \\
	\le \half \int_{x \in \Omega} \left|\frac{d\pi(x)}{1 - \eps'} - \frac{\frac{p(x)}{C\hp(x)}d\pih'(x)}{\int_{x \in \Omega} \frac{p(x)}{C\hp(x)}d\pih(x)}\right| + \half \int_{x \in \Omega} \left|\frac{\frac{p(x)}{C\hp(x)}d\pih'(x)}{\int_{x \in \Omega} \frac{p(x)}{C\hp(x)}d\pih(x)} - \frac{\frac{p(x)}{C\hp(x)}d\pih'(x)}{\int_{x \in \Omega} \frac{p(x)}{C\hp(x)}d\pih'(x)}\right| \\
	\le \half \int_{x \in \Omega} \left|\frac{d\pi(x)}{1 - \eps'} - \frac{\frac{p(x)}{C\hp(x)}d\pih'(x)}{\int_{x \in \Omega} \frac{p(x)}{C\hp(x)}d\pih(x)}\right|+ \frac{\delta C}{1 - \eps'} = \half \int_{x \in \Omega} \frac{d\pi(x)}{1 - \eps'}\left|1 - \frac{d\pih'}{d\pih}(x) \right| + \frac{\delta C}{1 - \eps'}.
	\end{align*}
	The first inequality was triangle inequality, and we bounded the second term by \eqref{eq:rationotbad}. To obtain the final equality, we used
	\begin{align*}\int_{x \in \Omega} \frac{p(x)}{C\hp(x)}d\pih(x) = \int_{x \in \Omega} \frac{Z}{C\hat{Z}}d\pi(x) = \frac{(1 - \eps')Z}{C\hat{Z}} \\
	\implies \frac{\frac{p(x)}{C\hp(x)}d\pih'(x)}{\int_{x \in \Omega} \frac{p(x)}{C\hp(x)}d\pih(x)} = \frac{p(x)}{Z} \cdot \frac{\hat{Z}}{\hp(x)} \cdot \frac{1}{1 - \eps'} \cdot d\pih'(x) = \frac{d\pi(x)}{1 - \eps'} \cdot \frac{d\pih'}{d\pih}(x).\end{align*}
	We now bound this final term. Observe that the given conditions imply that $\tfrac{d\pi}{d\pih}(x)$ is bounded by $C$ everywhere in $\Omega$. Thus, expanding we have
	\[\half \int_{x \in \Omega} \frac{d\pi(x)}{1 - \eps'}\left|1 - \frac{d\pih'}{d\pih}(x) \right|  \le \frac{C}{2(1 - \eps')} \int_{x \in \Omega} |d\pih(x) - d\pih'(x)| \le \frac{\delta C}{1 - \eps'}.\]
	Finally, combining these guarantees, and the fact that restricting $\pi$ to $\Omega$ loses $\eps'$ in total variation distance, yields the desired conclusion by triangle inequality.
\end{proof}

\begin{corollary}
	\label{corr:unbiased_reject_approx}
	Let $\hat{\theta}(x)$ be an unbiased estimator for $\tfrac{p(x)}{\hp(x)}$, and suppose $\hat{\theta}(x) \le C$ with probability 1 for all $x \in \Omega$. Then, implementing the procedure of Definition~\ref{def:reject_approx} with acceptance probability $\tfrac{\hat{\theta}(x)}{C}$ has the same runtime bound and total variation guarantee as given by Lemma~\ref{lem:reject_approx_proof}.
\end{corollary}
\begin{proof}
	It suffices to take expectations over the randomness of $\hat{\theta}$ everywhere in the proof of Lemma~\ref{lem:reject_approx_proof}.
\end{proof}

\subsubsection{Distribution ratio bounds}
\label{sssec:pipih}

We next show two bounds relating the densities of distributions $\pi$ and $\pih$.
We first define the normalization constants of \eqref{eq:pidef}, \eqref{eq:pihdef} for shorthand, and then tightly bound their ratio.
\begin{definition}[Normalization constants]
	\label{def:normconst}
	We denote normalization constants of $\pi$ and $\pih$ by
	\begin{align*}Z_\pi &\defeq \int_x \exp\left(-f(x) - g(x)\right) dx,\\
	Z_{\pih} &\defeq \int_{x, y} \exp\left(-f(y) - g(x) - \frac{1}{2\eta}\norm{y - x}_2^2 - \frac{\eta L^2}{2}\norm{x - x^*}_2^2\right)dxdy.\end{align*}
\end{definition}

\begin{lemma}[Normalization constant bounds]
	\label{lem:normratiobound}
	Let $Z_\pi$ and $Z_{\pih}$ be as in Definition \ref{def:normconst}.
	Then,
	\[\left(\frac{2\pi\eta}{1+\eta L}\right)^{\frac{d}{2}} \left(1+\frac{\eta L^{2}}{\mu}\right)^{-\frac{d}{2}} \le \frac{Z_{\pih}}{Z_{\pi}} \le (2\pi\eta)^{\frac{d}{2}}.\]
\end{lemma}
\begin{proof}
	For each $x$, by convexity we have
	\begin{equation}
	\label{eq:usefulbound}
	\begin{aligned}\int_y \exp\left(-f(y) - g(x) - \frac{1}{2\eta}\norm{y - x}_2^2 - \frac{\eta L^2}{2}\norm{x - x^*}_2^2\right)dy \\
	\le \exp\left(-g(x) - \frac{\eta L^2}{2}\norm{x - x^*}_2^2\right)\int_y\exp\left(-f(x)- \inprod{\nabla f(x)}{y - x} - \frac{1}{2\eta}\norm{y - x}_2^2\right)dy\\
	= \exp\left(-f(x) - g(x) - \frac{\eta L^2}{2}\norm{x - x^*}_2^2\right)\int_y \exp\left(\frac{\eta}{2}\norm{\nabla f(x)}_2^2 - \frac{1}{2\eta}\norm{y - x + \eta \nabla f(x)}_2^2\right) dy \\
	=(2\pi\eta)^{\frac{d}{2}}\exp\left(-f(x) - g(x)\right)\exp\left(\frac{\eta}{2}\norm{\nabla f(x)}_2^2 - \frac{\eta L^2}{2}\norm{x - x^*}_2^2\right) \\
	\le (2\pi\eta)^{\frac{d}{2}}\exp\left(-f(x) - g(x)\right).\end{aligned}
	\end{equation}
	Integrating both sides over $x$ yields the upper bound on $\tfrac{Z_{\pih}}{Z_{\pi}}$. Next, for the lower bound we have a similar derivation. For each $x$, by smoothness
	\begin{align*}
	\int_{y}\exp\left(-f(y)-g(x)-\frac{1}{2\eta}\left\Vert y-x\right\Vert_2 ^{2}-\frac{\eta L^{2}}{2}\left\Vert x-x^{*}\right\Vert_2^{2}\right)dy\\
	\geq \exp\left(-f(x)-g(x)-\frac{\eta L^{2}}{2}\left\Vert x-x^{*}\right\Vert_2^{2}\right)\int_{y}\exp\left(\left\langle \nabla f(x),x-y\right\rangle -\frac{1+\eta L}{2\eta}\left\Vert y-x\right\Vert_2^{2}\right)dy\\
	= \exp\left(-f(x)-g(x)-\frac{\eta L^{2}}{2}\left\Vert x-x^{*}\right\Vert ^{2} + \frac{\eta}{2(1+\eta L)}\left\Vert \nabla f(x)\right\Vert ^{2}\right)\left(\frac{2\pi\eta}{1+\eta L}\right)^{\frac{d}{2}}\\
	\geq  \exp\left(-f(x)-g(x)-\frac{\eta L^{2}}{2}\left\Vert x-x^{*}\right\Vert_2 ^{2}\right)\left(\frac{2\pi\eta}{1+\eta L}\right)^{\frac{d}{2}}.
	\end{align*}
	Integrating both sides over $x$ yields
	\begin{align*}
	\frac{Z_{\pih}}{Z_\pi} \ge \left(\frac{2\pi\eta}{1+\eta L}\right)^{\frac{d}{2}} \frac{\int_{x}\exp\left(-f(x)-g(x)-\frac{\eta L^{2}}{2}\left\Vert x-x^{*}\right\Vert_2 ^{2}\right)dx}{\int_{x}\exp\left(-f(x)-g(x)\right)dx}
	\geq \left(\frac{2\pi\eta}{1+\eta L}\right)^{\frac{d}{2}} \left(1+\frac{\eta L^{2}}{\mu}\right)^{-\frac{d}{2}}.
	\end{align*}
	The last inequality followed from Proposition~\ref{prop:normalizationratio}, where we used $f + g$ is $\mu$-strongly convex.
\end{proof}

\begin{lemma}[Relative density bounds]
	\label{lem:densityratio}
	Let $\eta = \tfrac{1}{32L\kappa d\log(288\kappa/\eps)}$. For all $x \in \Omega$, as defined in \eqref{eq:omegadef},  $\frac{d\pi}{d\pih}(x) \le 2$. Here, $\tfrac{d\pih}{dx}(x)$ denotes the marginal density of $\pih$. Moreover, for all $x \in \R^d$, $\frac{d\pi}{d\pih}(x) \ge \thalf$.
\end{lemma}
\begin{proof}
	We first show the upper bound. By Lemma~\ref{lem:normratiobound},
	\begin{equation}\label{eq:upperpartial}
	\begin{aligned}\frac{d\pi}{d\pih}(x) &= \frac{\exp\left(-f(x) - g(x)\right)}{\int_y\exp\left(-f(y) - g(x) - \frac{1}{2\eta}\norm{y - x}_2^2 - \frac{\eta L^2}{2}\norm{x - x^*}_2^2\right) dy} \cdot \frac{Z_{\pih}}{Z_\pi} \\
	&\le \frac{\exp\left(-f(x) - g(x)\right)}{\int_y\exp\left(-f(y) - g(x) - \frac{1}{2\eta}\norm{y - x}_2^2 - \frac{\eta L^2}{2}\norm{x - x^*}_2^2\right) dy} \cdot (2\pi\eta)^{\frac{d}{2}}.\end{aligned}\end{equation}
	We now bound the first term, for $x \in \Omega$. By smoothness, we have
	\[\frac{\exp\left(-f(y) - g(x)\right)}{\exp\left(-f(x) - g(x)\right)} \ge \exp\left(\inprod{\nabla f(x)}{x - y} - \frac{L}{2}\norm{y - x}_2^2\right), \]
	so applying this for each $y$,
	\begin{align*}
	\frac{\int_y \exp\left(-f(y) - g(x) - \frac{1}{2\eta}\norm{y - x}_2^2-\frac{\eta L^2}{2}\norm{x - x^*}_2^2\right)dy}{\exp\left(-f(x) - g(x)\right)} \\
	\ge \exp\left(-\frac{\eta L^2}{2}\norm{x - x^*}_2^2\right)\int_y \exp\left(\inprod{\nabla f(x)}{x - y} - \frac{1 + \eta L}{2\eta}\norm{y - x}_2^2\right)dy \\
	= \exp\left(-\frac{\eta L^2}{2}\norm{x - x^*}_2^2 + \frac{\eta}{2(1 + \eta L)}\norm{\nabla f(x)}_2^2\right)\int_y \exp\left(-\frac{1 + \eta L}{2\eta}\norm{x - y - \frac{\eta}{1 + \eta L}\nabla f(x)}_2^2\right)dy \\
	\ge \exp\left(-\frac{\eta L^2}{2}\cdot \frac{16d\log(288\kappa/\eps)}{\mu}\right)\left(\frac{2\pi\eta}{1 + \eta L}\right)^{\frac{d}{2}} \ge \frac{3}{4}\left(\frac{2\pi\eta}{1 + \eta L}\right)^{\frac{d}{2}}.
	\end{align*}
	In the last line, we used that $x \in \Omega$ implies $\norm{x - x^*}_2^2 \le \tfrac{16d\log(288\kappa/\eps)}{\mu}$, and the definition of $\eta$. Combining this bound with \eqref{eq:upperpartial}, we have the desired
	\[\frac{d\pi}{d\pih}(x) \le \frac{4}{3}\left(1 + \eta L\right)^{\frac{d}{2}} \le 2.\] Next, we consider the lower bound. By combining \eqref{eq:usefulbound} with Lemma~\ref{lem:normratiobound}, we have the desired
	\begin{align*}\frac{d\pi}{d\pih}(x) &= \frac{\exp\left(-f(x) - g(x)\right)}{\int_y\exp\left(-f(y) - g(x) - \frac{1}{2\eta}\norm{y - x}_2^2 - \frac{\eta L^2}{2}\norm{x - x^*}_2^2\right) dy} \cdot \frac{Z_{\pih}}{Z_\pi} \\
	&\ge (2\pi\eta)^{-\frac{d}{2}} \cdot \left(\frac{2\pi\eta}{1+\eta L}\right)^{\frac{d}{2}} \left(1+\frac{\eta L^{2}}{\mu}\right)^{-\frac{d}{2}} = \left(\frac{1}{1 + \eta L}\right)^{\frac{d}{2}}\left(1 + \eta L \kappa\right)^{-\frac{d}{2}} \ge \half.\end{align*}
\end{proof}

\subsubsection{Correctness of $\cssm$}
\restatecssmcorrectness*
\begin{proof}
	We remark that $\eta = \tfrac{1}{32 L\kappa d\log(288\kappa/\eps)}$ is precisely the choice of $\eta$ in $\sjd$ where $\delta = \eps/18$, as in $\cssm$. First, we may apply Fact~\ref{fact:subgauss} to conclude that the measure of set $\Omega$ with respect to the $\mu$-strongly logconcave density $\pi$ is at least $1 - \eps/3$. The conclusion of correctness will follow from an appeal to Corollary~\ref{corr:unbiased_reject_approx}, with parameters
	\[C = 4,\; \eps' = \frac{\eps}{3},\; \delta = \frac{\eps}{18}.\]
	Note that indeed we have $\eps' + \tfrac{2\delta C}{1 - \eps'}$ is bounded by $\eps$, as $1 - \eps' \ge \tfrac{2}{3}$. Moreover, the expected number of calls \eqref{eq:algcalls} is clearly bounded by a constant as well. 
	
	We now show that these parameters satisfy the requirements of Corollary~\ref{corr:unbiased_reject_approx}. Define the functions
	\begin{align*}p(x) &\defeq \exp(-f(x) - g(x)),\\ \hp(x) &\defeq (2\pi\eta)^{-\frac{d}{2}}\int_y \exp\left(-f(y) - g(x) - \frac{1}{2\eta}\norm{y - x}_2^2 - \frac{\eta L^2}{2}\norm{x - x^*}_2^2\right)dy,\end{align*}
	and observe that clearly the densities of $\pi$ and $\pih$ are respectively proportional to $p$ and $\hp$. Moreover, define $Z = \int p(x) dx$ and $\hat{Z} = \int \hp(x) dx$. By comparing these definitions with Lemma~\ref{lem:normratiobound}, we have $Z = Z_\pi$ and $\hat{Z} = (2\pi\eta)^{-\frac{d}{2}}Z_{\pih}$, so by the upper bound in Lemma~\ref{lem:normratiobound}, $\hat{Z}/Z \le 1$. Next, we claim that the following procedure produces an unbiased estimator for $\tfrac{p(x)}{\hp(x)}$.
	\begin{enumerate}
		\item Sample $y \sim \pi_x$, where $\tfrac{d\pi_x(y)}{dy} \propto \exp\left(-f(y) - \tfrac{1}{2\eta}\norm{y - x}_2^2\right)$
		\item $\alpha \gets \exp\left(f(y) - \inprod{\nabla f(x)}{y - x} - \tfrac{L}{2}\norm{y - x}_2^2 + g(x) + \tfrac{\eta L^2}{2}\norm{x - x^*}_2^2\right)$
		\item Output $\hat{\theta}(x) \gets \exp\left(-f(x) - g(x) + \tfrac{\eta}{2(1 + \eta L)}\norm{\nabla f(x)}_2^2\right)(1 + \eta L)^{\frac{d}{2}}\alpha$
	\end{enumerate}
	To prove correctness of this estimator $\hat{\theta}$, define for simplicity
	\[Z_x \defeq \int_y \exp\left(-f(y) - g(x) - \frac{1}{2\eta}\norm{y - x}_2^2 - \frac{\eta L^2}{2}\norm{x - x^*}_2^2\right)dy.\]
	We compute, using $\tfrac{d\pi_x(y)}{dy} = \tfrac{\exp(-f(y) - g(x) - \frac{1}{2\eta}\norm{y - x}_2^2 - \frac{\eta L^2}{2}\norm{x - x^*}_2^2)}{Z_x}$, that
	\begin{align*}
	\E_{\pi_x}\left[\alpha\right] &= \int_y \exp\left(f(y) - \inprod{\nabla f(x)}{y - x} - \frac{L}{2}\norm{y - x}_2^2 + g(x) + \frac{\eta L^2}{2}\norm{x - x^*}_2^2\right)d\pi_x(y) \\
	&= \frac{1}{Z_x} \int_y \exp\left(- \inprod{\nabla f(x)}{y - x} - \frac{L}{2}\norm{y - x}_2^2 - \frac{1}{2\eta}\norm{y - x}_2^2\right) dy \\
	&= \frac{1}{Z_x}\exp\left(-\frac{\eta}{2(1 + \eta L)}\norm{\nabla f(x)}_2^2\right)\left(\frac{2\pi \eta}{1 + \eta L}\right)^{\frac{d}{2}}.
	\end{align*}
	This implies that the output quantity
	\[\hat{\theta}(x) = \exp\left(-f(x) - g(x) + \frac{\eta}{2(1 + \eta L)}\norm{\nabla f(x)}_2^2\right)(1 + \eta L)^{\frac{d}{2}}\alpha\]
	is unbiased for $\tfrac{p(x)}{\hp(x)} = \exp(-f(x) - g(x)) Z_x^{-1}(2\pi\eta)^{\frac{d}{2}}$. Finally, note that for any $y$ used in the definition of $\hat{\theta}(x)$, by using $f(y) - f(x) - \inprod{\nabla f(x)}{y - x} - \tfrac{L}{2}\norm{y - x}_2^2 \le 0$ via smoothness, we have
	\begin{align*}\hat{\theta}(x) &=  \exp\left(-f(x) - g(x) + \frac{\eta}{2(1 + \eta L)}\norm{\nabla f(x)}_2^2\right)(1 + \eta L)^{\frac{d}{2}}\alpha \\
	&\le (1 + \eta L)^{\frac{d}{2}}\exp\left(\frac{\eta}{2(1 + \eta L)}\norm{\nabla f(x)}_2^2 + \frac{\eta L^2}{2}\norm{x - x^*}_2^2\right) \\
	&\le (1 + \eta L)^{\frac{d}{2}}\exp\left(\eta L^2\norm{x - x^*}_2^2\right) \le 4.
	\end{align*}
	Here, we used the definition of $\eta$ and $L^2\norm{x - x^*}_2^2 \le 16L\kappa d\log(288\kappa/\eps)$ by the definition of $\Omega$. 
\end{proof}

\subsection{Deferred proofs from Section~\ref{ssec:alternate}}

Throughout this section, for error tolerance $\delta \in [0, 1]$ which parameterizes $\sjd$, we denote for shorthand a high-probability region $\Omega_\delta$ and its radius $R_\delta$ by
\begin{equation}\label{eq:omegadelta}\Omega_\delta \defeq \left\{x \mid \norm{x - x^*}_2 \le R_\delta\right\},\text{ for } R_\delta \defeq 4\sqrt{\frac{d\log(16\kappa/\delta)}{\mu}}.\end{equation}
The following density ratio bounds hold within this region, by simply modifying Lemma~\ref{lem:densityratio}.
\begin{corollary}
	\label{corr:densityratiodelta}
	Let $\eta = \tfrac{1}{32L\kappa d\log(16\kappa/\delta)}$, and let $\pih$ be parameterized by this choice of $\eta$ in \eqref{eq:pihdef}. For all $x \in \Omega_\delta$, as defined in \eqref{eq:omegadelta}, $\frac{d\pi}{d\pih}(x) \le 2$. Moreover, for all $x \in \R^d$, $\frac{d\pi}{d\pih}(x) \ge \half$.
\end{corollary}
The following claim follows immediately from applying Fact~\ref{fact:subgauss}.
\begin{lemma}\label{lem:omegahighprob}
	With probability at least $1 - \tfrac{\delta^2}{8(1 + \kappa)^d}$, $x \sim \pih$ lies in $\Omega_\delta$.
\end{lemma}
Finally, when clear from context, we overload $\pih$ as a distribution on $x\in\R^d$ to be the $x$ component marginal of the distribution \eqref{eq:pihdef}, i.e. with density
\[ \frac{d\pih}{dx}(x) \propto \int_y\exp\left(-f(y) - g(x) - \frac{1}{2\eta}\norm{y - x}_2^2 - \frac{\eta L^2}{2}\norm{x - x^*}_2^2\right) dy.
\]

We first note that $\pih$ is stationary for $\sjd$; this follows immediately from Lemma~\ref{lem:alternate_exact}. In Section~\ref{sssec:sjdconduct}, we bound the \emph{conductance} of the walk. We then use this bound in Section~\ref{sssec:runtime} to bound the mixing time and overall complexity of $\sjd$. 

\subsubsection{Conductance of $\sjd$}
\label{sssec:sjdconduct}
We bound the conductance of this random walk, as a process on the iterates $\{x_k\}$, to show the final point has distribution close to the marginal of $\pih$ on $x$. To do so, we break Proposition~\ref{prop:mixtime} into two pieces, which we will use in a more white-box manner to prove our conductance bound.

\begin{definition}[Restricted conductance]
	Let a random walk with stationary distribution $\pih$ on $x \in \R^d$ have transition densities $\tran_x$, and let $\Omega \subseteq \R^d$. The $\Omega$-restricted conductance, for $v \in (0, \thalf\pih(\Omega))$, is
	\[\Phi_{\Omega}(v) = \inf_{\pih(S \cap \Omega)\in(0, v]} \frac{\tran_S(S^c)}{\pih(S \cap \Omega)}, \text{ where } \tran_S(S^c) \defeq \int_{x \in S}\int_{x'\in S^c}\tran_x(x') d\pih(x)dx'.\]
\end{definition}

\begin{proposition}[Lemma 1, \cite{ChenDWY19}]
	\label{prop:mixviaconduct}
	Let $\pistart$ be a $\beta$-warm start for $\pih$, and let $x_0 \sim \pistart$. For some $\delta > 0$, let $\Omega \subseteq \R^d$ have $\pih(\Omega) \ge 1 - \tfrac{\delta^2}{2\beta^2}$. Suppose that a random walk with stationary distribution $\pih$ satisfies the $\Omega$-restricted conductance bound
	\[\Phi_{\Omega}(v) \ge \sqrt{B\log\left(\frac{1}{v}\right)},\text{ for all } v \in \left[\frac{4}{\beta},\half\right].\]
	Let $x_K$ be the result of $K$ steps of this random walk, starting from $x_0$. Then, for 
	\[K \ge \frac{64}{B}\log\left(\frac{\log\beta}{2\delta}\right),\]
	the resulting distribution of $x_K$ has total variation at most $\tfrac{\delta}{2}$ from $\pih$.
\end{proposition}

We state a well-known strategy for lower bounding conductance, via showing the stationary distribution has good \emph{isoperimetry} and that transition distributions of nearby points have large overlap.

\begin{proposition}[Lemma 2, \cite{ChenDWY19}]
	\label{prop:conductviaisotv}
	Let a random walk with stationary distribution $\pih$ on $x \in \R^d$ have transition distribution densities $\tran_x$, and let $\Omega \subseteq \R^d$, and let $\pih_\Omega$ be the conditional distribution of $\pih$ on $\Omega$. Suppose for any $x, x' \in \Omega$ with $\norm{x - x'}_2 \le \Delta$,
	\[\tvd{\tran_x}{\tran_{x'}} \le \half.\]
	Also, suppose $\pih_\Omega$ satisfies, for any partition $S_1$, $S_2$, $S_3$ of $\Omega$, where $d(S_1, S_2)$ is the minimum Euclidean distance between points in $S_1$, $S_2$, the log-isoperimetric inequality
	\begin{equation}\label{eq:logiso}\pih_\Omega(S_3) \ge \frac{1}{2\psi}d(S_1, S_2) \cdot \min\left(\pih_\Omega(S_1), \pih_\Omega(S_2)\right) \cdot \sqrt{\log\left(1 + \frac{1}{\min\left(\pih_\Omega(S_1), \pih_\Omega(S_2)\right)}\right)}.\end{equation}
	Then, we have the bound for all $v \in (0, \thalf]$
	\[\Phi_{\Omega}(v) \ge \min\Par{1, \frac{\Delta}{128\psi}\sqrt{\log\left(\frac{1}{v}\right)}}.\]
\end{proposition}

To utilize Propositions~\ref{prop:mixviaconduct} and~\ref{prop:conductviaisotv}, we prove the following bounds in Appendices~\ref{ssec:warmstart},~\ref{ssec:tvclose}, and~\ref{ssec:isoperimetry}.

\begin{restatable}[Warm start]{lemma}{restatewarmstart}\label{lem:warmstart}
	For $\eta \le \tfrac{1}{L\kappa d}$, $\pistart$ defined in \eqref{eq:pistartdef} is a $2(1 + \kappa)^{\frac{d}{2}}$-warm start for $\pih$.
\end{restatable}

\begin{restatable}[Transitions of nearby points]{lemma}{restatetrantv}\label{lem:tv_closepts}
	Suppose $\eta L \le 1$, $\eta L^2R_{\delta}^2 \le \thalf$, and $400d^2\eta\le R_\delta^2$. For a point $x$, let $\tran_x$ be the density of $x_k$ after sampling according to Lines 6 and 7 of Algorithm~\ref{alg:sjd} from $x_{k - 1} = x$. For $x, x' \in \Omega_\delta$ with $\norm{x - x'}_2 \le \tfrac{\sqrt{\eta}}{10}$, for $\Omega_\delta$ defined in \eqref{eq:omegadelta}, we have $\tvd{\tran_x}{\tran_{x'}} \le \thalf$.
\end{restatable}

\begin{restatable}[Isoperimetry]{lemma}{restateiso}\label{lem:iso}
	Density $\pih$ and set $\Omega_\delta$ defined in \eqref{eq:pihdef}, \eqref{eq:omegadelta} satisfy \eqref{eq:logiso} with $\psi = 8\mu^{-\half}$.
\end{restatable}

We note that the parameters of Algorithm~\ref{alg:sjd} and the set $\Omega_\delta$ in \eqref{eq:omegadelta} satisfy all assumptions of Lemmas~\ref{lem:warmstart},~\ref{lem:tv_closepts}, and~\ref{lem:iso}. By combining these results in the context of Proposition~\ref{prop:conductviaisotv}, we see that the random walk satisfies the bound for all $v \in (0, \thalf]$:
\[\Phi_{\Omega_\delta}(v) \ge \sqrt{\frac{\eta\mu}{2^{20}\cdot100}\cdot\log\left(\frac{1}{v}\right)}.\]
Plugging this conductance lower bound, the high-probability guarantee of $\Omega_\delta$ by Lemma~\ref{lem:omegahighprob}, and the warm start bound of Lemma~\ref{lem:warmstart} into Proposition~\ref{prop:mixviaconduct}, we have the following conclusion.

\begin{corollary}[Mixing time of ideal $\sjd$]\label{corr:sjdmix}
	Assume that calls to $\yor$ are exact in the implementation of $\sjd$. Then, for any error parameter $\delta$, and
	\[K \defeq \frac{2^{26}\cdot100}{\eta\mu}\log\left(\frac{d\log(16\kappa)}{4\delta}\right),\]
	the distribution of $x_K$ has total variation at most $\tfrac{\delta}{2}$ from $\pih$.
\end{corollary}

\subsubsection{Complexity of $\sjd$}
\label{sssec:runtime}

We first state a guarantee on the subroutine $\yor$, which we prove in Appendix~\ref{ssec:yoracle}.

\begin{restatable}[$\yor$ guarantee]{lemma}{restateyor}\label{lem:yor} For $\delta \in [0, 1]$, define $R_\delta$ as in \eqref{eq:omegadelta}, and let $\eta = \tfrac{1}{32L\kappa d\log(16\kappa/\delta)}$. For any $x$ with $\norm{x - x^*}_2 \le \sqrt{\kappa d\log(16\kappa/\delta)}\cdot R_\delta$, Algorithm~\ref{alg:yor} ($\yor$) draws an exact sample $y$ from the density proportional to $\exp\left(-f(y) - \tfrac{1}{2\eta}\norm{y - x}_2^2\right)$ in an expected $2$ iterations.
\end{restatable}

We also state a result due to \cite{ChenDWY19}, which bounds the mixing time of 1-step Metropolized HMC for well-conditioned distributions; this handles the case when $\norm{x - x^*}_2$ is large in Algorithm~\ref{alg:yor}.

\begin{proposition}[Theorem 1, \cite{ChenDWY19}]\label{prop:atmostd}
	Let $\pi$ be a distribution on $\R^d$ whose negative log-density is convex and has condition number bounded by a constant. Then, Metropolized HMC from an explicit starting distribution mixes to total variation $\delta$ to the distribution $\pi$ in $O(d\log(\tfrac{d}{\delta}))$ iterations.
\end{proposition}

\restatesjdguarantee*
\begin{proof}
	Under an exact $\yor$, Corollary~\ref{corr:sjdmix} shows the output distribution of $\sjd$ has total variation at most $\tfrac{\delta}{2}$ from $\pih$. Next, the resulting distribution of the subroutine $\yor$ is never larger than $\delta/(2Kd\log(\frac{d\kappa}{\delta}))$ in total variation distance away from an exact sampler. By running for $K$ steps, and using the coupling characterization of total variation, it follows that this can only incur additional error $\delta/(2d\log(\frac{d\kappa}{\delta}))$, proving correctness (in fact, the distribution is always at most $O((d\log(d\kappa/\delta))^{-1})$ away in total variation from an exact $\yor$).
	
	Next, we prove the guarantee on the expected gradient evaluations per iteration. Lemma~\ref{lem:yor} shows whenever the current iterate $x_k$ has $\norm{x - x^*}_2 \le \sqrt{\kappa d\log(16\kappa/\delta)} \cdot R_{\delta}$, the expected number of gradient evaluations is constant, and moreover Proposition~\ref{prop:atmostd} shows that the number of gradient evaluations is never larger than $O(d\log(\tfrac{d\kappa}{\delta}))$, where we use that the condition number of the log-density in \eqref{eq:pixdef} is bounded by a constant. Therefore, it suffices to show in every iteration $0 \le k \le K$, the probability $\norm{x_k - x^*}_2 > \sqrt{\kappa d\log(16\kappa/\delta)} \cdot R_{\delta}$ is $O((d\log(d\kappa/\delta))^{-1})$. By the warmness assumption in Lemma~\ref{lem:warmstart}, and the concentration bound in Fact~\ref{fact:subgauss}, the probability $x_0$ does not satisfy this bound is negligible (inverse exponential in $\kappa d^2\log(\kappa/\delta)$). Since warmness is monotonically decreasing with an exact sampler,\footnote{This fact is well-known in the literature, and a simple proof is that if a distribution is warm, then taking one step of the Markov chain induces a convex combination of warm point masses, and is thus also warm.} and the accumulated error due to inexactness of $\yor$ is at most $O((d\log(d\kappa/\delta))^{-1})$ through the whole algorithm, this holds for all iterations.
\end{proof} 	% 	%

\section{Mixing time ingredients}
\label{sec:ingredients}

We now prove facts which are used in the mixing time analysis of $\sjd$. Throughout this section, as in the specification of $\sjd$, $f$ and $g$ are functions with properties as in \eqref{eq:pidef}, and share a minimizer $x^*$.

\subsection{Warm start}
\label{ssec:warmstart}
We show that we obtain a warm start for the distribution $\pih$ in algorithm $\sjd$ via one call to the restricted Gaussian oracle for $g$, by proving Lemma~\ref{lem:warmstart}.

\restatewarmstart*
\begin{proof}
	By the definitions of $\pih$ and $\pistart$ in \eqref{eq:pihdef}, \eqref{eq:pistartdef}, we wish to bound everywhere the quantity
	\begin{equation}\label{eq:warmness}\frac{d\pistart}{d\pih}(x) = \frac{Z_{\pih}}{\zstart} \cdot  \frac{\exp\left(-\frac{L}{2}\norm{x - x^*}_2^2-\frac{\eta L^2}{2}\norm{x - x^*}_2^2 - g(x)\right)}{\int_y\exp\left(-f(y) - g(x) - \frac{1}{2\eta}\norm{y - x}_2^2 - \frac{\eta L^2}{2}\norm{x - x^*}_2^2\right) dy}. \end{equation}
	Here, $Z_{\pih}$ is as in Definition~\ref{def:normconst}, and we let $\zstart$ denote the normalization constant of $\pistart$, i.e.
	\[\zstart \defeq  \int_x \exp\left(-\frac{L}{2}\norm{x - x^*}_2^2 - \frac{\eta L^2}{2}\norm{x - x^*}_2^2- g(x)\right) dx.\]
	Regarding the first term of \eqref{eq:warmness}, 
	the earlier derivation \eqref{eq:usefulbound} showed
	\[
	\int_y \exp\left(-f(y) - g(x) - \frac{1}{2\eta}\norm{y - x}_2^2 - \frac{\eta L^2}{2}\norm{x - x^*}_2^2\right)dy 
	\le (2\pi\eta)^{\frac{d}{2}}\exp\left(-f(x) - g(x)\right).
	\]
	Then, integrating, we can bound the ratio of the normalization constants
	\begin{equation}
	\label{eq:zhzstart}
	\begin{aligned}
	\frac{Z_{\pih}}{Z_{\pistart}}  &\leq \frac{ \int_x(2\pi\eta)^{\frac{d}{2}}\exp\left(-f(x) - g(x)\right)dx}{ \int_x \exp\left(-\frac{L}{2}\norm{x - x^*}_2^2 - \frac{\eta L^2}{2}\norm{x - x^*}_2^2- g(x)\right) dx}  \\
	&\leq \frac{ \int_x(2\pi\eta)^{\frac{d}{2}}\exp\left(-f(x^*)-\frac{\mu}{2}\norm{x - x^*}_2^2 - g(x)\right)dx}{ \int_x \exp\left(-\frac{L}{2}\norm{x - x^*}_2^2 - \frac{\mu}{2}\norm{x - x^*}_2^2- g(x)\right) dx} \\
	&\leq (2\pi\eta)^{\frac{d}{2}}\exp\left(-f(x^*)\right) \left( 1 + \frac{L}{\mu}\right)^{\frac{d}{2}}.
	\end{aligned}
	\end{equation}
	The second inequality followed from $f$ is $\mu$-strongly convex and $\eta L^2 \leq \mu$ by assumption. The last inequality followed from Proposition \ref{prop:normalizationratio}, where we used  $\frac{\mu}{2}\norm{x - x^*}_2^2 + g(x)$ is $\mu$-strongly convex. Next, to bound the second term of \eqref{eq:warmness}, notice first that 
	\begin{align*}
	\frac{\exp\left(-\frac{L}{2}\norm{x - x^*}_2^2-\frac{\eta L^2}{2}\norm{x - x^*}_2^2 - g(x)\right)}{\int_y\exp\left(-f(y) - g(x) - \frac{1}{2\eta}\norm{y - x}_2^2 - \frac{\eta L^2}{2}\norm{x - x^*}_2^2\right) dy} 
	= \frac{\exp \left(-\frac{L}{2}\norm{x - x^*}_2^2\right) }{\int_y\exp\left(-f(y)  - \frac{1}{2\eta}\norm{y - x}_2^2\right) dy}.
	\end{align*}
	It thus suffices to lower bound $\exp \left(\frac{L}{2}\norm{x - x^*}_2^2\right)\int_y\exp\left(-f(y)  - \frac{1}{2\eta}\norm{y - x}_2^2 \right) dy$. We have
	\begin{equation}
	\label{eq:warmnessotherterm}
	\begin{aligned}
	\exp \left(\frac{L}{2}\norm{x - x^*}_2^2\right)\int_y\exp\left(-f(y)  - \frac{1}{2\eta}\norm{y - x}_2^2 \right) dy \\
	\geq \exp\left(-f(x)+\frac{ L}{2}\norm{x - x^*}_2^2  \right)\int_y\exp\left(- \langle \nabla f(x),y-x\rangle  - \left(\frac{1}{2\eta}+\frac{L}{2}\right)\norm{y-x}_2^2\right) dy \\
	= \exp\left(-f(x)+\frac{ L}{2}\norm{x - x^*}_2^2  \right) \left( \frac{2\pi\eta }{1+L\eta}\right)^{\frac{d}{2}} \exp\left(\frac{\eta}{2(1+L\eta )}\norm{\nabla f(x)}_2^2\right) \\
	\geq \exp(-f(x^*)) \left( \frac{2\pi\eta }{1+L\eta}\right)^{\frac{d}{2}}
	\end{aligned}
	\end{equation}
	The first and third steps followed from $L$-smoothness of $f$, and the second applied the Gaussian integral (Fact~\ref{fact:gaussz}). Combining the bounds in \eqref{eq:zhzstart} and \eqref{eq:warmnessotherterm}, \eqref{eq:warmness} becomes
	\begin{align*}
	\frac{d\pistart}{d\pih}(x) \leq  \left( 1 + \frac{L}{\mu}\right)^{\frac{d}{2}} \left( 1 + L\eta\right)^{\frac{d}{2}} \leq 2(1+\kappa)^{\frac{d}{2}},
	\end{align*}
	where $x \in \R^d$ was arbitrary, which completes the proof.
\end{proof}

\subsection{Transitions of nearby points}
\label{ssec:tvclose}

Here, we prove Lemma~\ref{lem:tv_closepts}. Throughout this section, $\tran_x$ is the density of $x_k$, according to the steps in Lines 6 and 7 of $\sjd$ (Algorithm~\ref{alg:sjd}) starting at $x_{k - 1} = x$. We also define $\prop_x$ to be the density of $y_k$, by just the step in Line 6. We first make a simplifying observation: by Observation~\ref{observe:pvst}, for any two points $x$, $x'$, we have
	\[\tvd{\tran_x}{\tran_{x'}} \le \tvd{\prop_x}{\prop_{x'}}.\]
Thus, it suffices to understand $\tvd{\prop_x}{\prop_{x'}}$ for nearby $x, x' \in \Omega_\delta$. Our proof of Lemma~\ref{lem:tv_closepts} combines two pieces: (1) bounding the ratio of normalization constants $Z_x$, $Z_{x'}$ of $\prop_x$ and $\prop_{x'}$ for nearby $x$, $x'$ in Lemma~\ref{lem:norm_ratio_closepts} and (2) the structural result Proposition~\ref{prop:min_perturb}. 
To bound the normalization constant ratio, we state two helper lemmas. Lemma~\ref{lem:ycloser} characterizes facts about the minimizer of 
\begin{equation}\label{eq:marginalfunc}f(y) + \frac{1}{2\eta}\norm{y - x}_2^2. \end{equation} 
\begin{lemma}
	\label{lem:ycloser}
	Let $f$ be convex with minimizer $x^*$, and $y_x$ minimize \eqref{eq:marginalfunc} for a given $x$. Then,
	\begin{enumerate}
		\item $\norm{y_x - y_{x'}}_2 \le \norm{x - x'}_2$.
		\item For any $x$, $\norm{y_x - x^*}_2 \le \norm{x - x^*}_2$.
		\item For any $x$ with $\norm{x - x^*}_2 \le R$, $\norm{x - y_x}_2 \le \eta LR$.
	\end{enumerate}
\end{lemma}
\begin{proof}By optimality conditions in the definition of $y_x$,
	\[\eta\nabla f(y_x) = x - y_x. \]
	Fix two points $x$, $x'$, and let $x_t \defeq (1 - t)x + tx'$. Letting $\jac_x(y_x)$ be the Jacobian matrix of $y_x$,
	\begin{align*}\frac{d}{dt} \eta\nabla f(y_{x_t}) = \frac{d}{dt}\left(x_t - y_{x_t}\right) &\implies \eta \nabla^2 f(y_{x_t}) \jac_x(y_{x_t}) (x' - x) = (\id - \jac_x(y_{x_t}))(x' - x)\\
	&\implies \jac_x(y_{x_t}) (x' - x) = (\id + \eta \nabla^2 f(y_{x_t}))^{-1}(x' - x).\end{align*}
	We can then compute
	\[y_{x'} - y_x = \int_0^1 \frac{d}{dt}y_{x_t} dt = \int_0^1 \jac_x(y_{x_t})(x' - x)dt = \int_0^1 (\id + \eta \nabla^2 f(y_{x_t}))^{-1}(x' - x)dt. \]
	By triangle inequality and convexity of $f$, the first claim follows:
	\[\norm{y_{x'} - y_x}_2 \le \int_0^1 \norm{(\id + \eta\nabla^2 f(y_{x_t}))^{-1}}_2 \norm{x' - x}_2 dt \le \norm{x' - x}_2. \]
	The second claim follows from the first by $y_{x^*} = x^*$. The third claim follows from the second via
	\[\norm{x - y_x}_2 = \eta\norm{\nabla f(y_x)}_2 \le \eta L\norm{y_x - x^*}_2 \le \eta LR.\]
\end{proof}
Next, Lemma~\ref{lem:gauint_in_closepts} states well-known bounds on the integral of a well-conditioned function $h$.
\begin{lemma}
	\label{lem:gauint_in_closepts}
	Let $h$ be a $L_h$-smooth, $\mu_h$-strongly convex function and let $y^*_h$ be its minimizer. Then
	\[\left(2\pi L_h^{-1}\right)^{\frac{d}{2}}\exp\left(-h(y^*_h)\right) \le \int_y \exp\left(-h(y)\right) \le \left(2\pi \mu_h^{-1}\right)^{\frac{d}{2}}\exp\left(-h(y^*_h)\right).\]
\end{lemma}
\begin{proof}
	By smoothness and strong convexity,
	\[\exp\left(-h(y^*_h) - \frac{L_h}{2}\norm{y - y^*_h}_2^2\right) \le \exp(-h(y)) \le \exp\left(-h(y^*_h) - \frac{\mu_h}{2}\norm{y - y^*_h}_2^2\right).\]
	The result follows by Gaussian integrals, i.e.\ Fact~\ref{fact:gaussz}.
\end{proof}

We now define the normalization constants of $\prop_x$ and $\prop_{x'}$:
\begin{equation}\label{eq:zxzxpdef}\begin{aligned}Z_x = \int_y \exp\left(-f(y) - \frac{1}{2\eta}\norm{y - x}_2^2\right)dy,\\ Z_{x'} = \int_y \exp\left(-f(y) - \frac{1}{2\eta}\norm{y - x'}_2^2\right)dy. \end{aligned}\end{equation}

We apply Lemma~\ref{lem:ycloser} and Lemma~\ref{lem:gauint_in_closepts} to bound the ratio of  $Z_x$ and $Z_{x'}$.

\begin{lemma}
	\label{lem:norm_ratio_closepts}
	Let $f$ be $\mu$-strongly convex and $L$-smooth. Let $x, x' \in \Omega_\delta$, for $\Omega_\delta$ defined in \eqref{eq:omegadelta}, and let $\norm{x-x'}_2\leq \Delta$. Then, the normalization constants $Z_x$ and $Z_{x'}$ in \eqref{eq:zxzxpdef} satisfy
	\[
	\frac{Z_x}{Z_{x'}} \leq1.05\exp \left(3LR\Delta + \frac{L\Delta^2}{2}\right).
	\]
\end{lemma}

\begin{proof}
	First, applying Lemma~\ref{lem:gauint_in_closepts} to $Z_x$ and $Z_{x'}$ yields that the ratio is bounded by
	\begin{align*}\frac{Z_x}{Z_{x'}} &\le \frac{\exp\left(-f(y_x) - \frac{1}{2\eta}\norm{y_x - x}_2^2 \right)\left(2\pi\left(\mu + \frac{1}{\eta}\right)^{-1}\right)^{\frac{d}{2}}}{\exp\left(-f(y_{x'}) - \frac{1}{2\eta}\norm{y_{x'} - x}_2^2\right)\left(2\pi\left(L + \frac{1}{\eta}\right)^{-1}\right)^{\frac{d}{2}}} \\
	&\le 1.05\exp\left(f(y_{x'}) - f(y_x) + \frac{1}{2\eta}\left(\norm{y_{x'} - x'}_2^2 - \norm{y_x - x}_2^2\right)\right). \end{align*}
	Here, we used the bound for $\eta^{-1} \ge 32Ld$ that
	\[\left(\frac{L + \frac{1}{\eta}}{\mu + \frac{1}{\eta} }\right)^{d/2} \le 1.05.\]
	Regarding the remaining term, recall $x$, $x'$ both belong to $\Omega_\delta$, and $\norm{x - x'}_2 \le \Delta$. We have 
	\begin{align*}f(y_{x'}) - f(y_x) + \frac{1}{2\eta}\left(\norm{y_{x'} - x'}_2^2 - \norm{y_x - x}_2^2\right) \\
	\le \inprod{\nabla f(y_x)}{y_{x'} - y_x} + \frac{L}{2}\norm{y_{x'} - y_x}_2^2
	+ \frac{1}{2\eta}\inprod{y_{x'} - x' + y_x - x}{y_{x'} - y_x + x - x'} \\
	\le LR\Delta + \frac{L\Delta^2}{2} + \frac{1}{2\eta}\left(\norm{y_x - x}_2 + \norm{y_{x'} - x'}_2\right)\left(\norm{y_{x'} - y_x}_2 + \norm{x' - x}_2\right)\\
	\le LR\Delta + \frac{L\Delta^2}{2} +  \frac{2\eta LR}{2\eta}\left(\norm{y_{x'} - y_x}_2 + \norm{x' - x}_2\right) \le 3LR\Delta + \frac{L\Delta^2}{2}.\end{align*}
	The first inequality was smoothness and expanding the difference of quadratics. The second was by $\norm{\nabla f(y_x)}_2 \le L\norm{y_x - x^*}_2 \le LR$ and $\norm{y_{x'} - y_x}_2 \le \Delta$, where we used the first and second parts of Lemma~\ref{lem:ycloser}; we also applied Cauchy-Schwarz and triangle inequality. The third used the third part of Lemma~\ref{lem:ycloser}. Finally, the last inequality was by the first part of Lemma~\ref{lem:ycloser} and $\norm{x' - x}_2 \le \Delta$.
\end{proof}
We now are ready to prove Lemma~\ref{lem:tv_closepts}.
\restatetrantv*
\begin{proof}
	First, by Observation~\ref{observe:pvst}, it suffices to show $\tvd{\prop_x}{\prop_{x'}} \le \thalf$. Pinsker's inequality states
	\[
	\tvd{\prop_x }{ \prop_{x'}} \leq\sqrt{\frac{1}{2}d_{\text{KL}}\left(\prop_x,\prop_{x'}\right)},
	\]
	where $d_{\text{KL}}$ is KL-divergence, so it is enough to show $d_{\text{KL}}\left(\prop_x,\prop_{x'}\right) \le \thalf$. Notice that
	\begin{align*}
	d_{\text{KL}}\left(\prop_x,\prop_{x'}\right) = \log\left(\frac{Z_{x'}}{Z_x}\right) 
	+ \int_y \prop_x(y) \log\left(\frac{\exp\left(-f(y) - \frac{1}{2\eta}\norm{y - x}_2^2\right)}{\exp\left(-f(y) - \frac{1}{2\eta}\norm{y - x'}_2^2\right)}\right) dy.
	\end{align*}
	By Lemma~\ref{lem:norm_ratio_closepts}, the first term satisfies, for $\Delta \defeq \tfrac{\sqrt{\eta}}{10}$,
	\[
	\log\left(\frac{Z_{x'}}{Z_x}\right) \leq 3LR\Delta + \frac{L\Delta^2}{2} + \log(1.05).
	\]
	To bound the second term, we have 
	\begin{align*} \int_y \prop_x(y) \log\left(\frac{\exp\left(-f(y) - \frac{1}{2\eta}\norm{y - x}_2^2\right)}{\exp\left(-f(y) - \frac{1}{2\eta}\norm{y - x'}_2^2\right)}\right) dy &= \frac{1}{2\eta}\int_y \prop_x(y) \left(\norm{y - x'}_2^2 - \norm{y - x}_2^2\right)dy\\
	&= \frac{1}{2\eta}\int_y \prop_x(y)\inprod{x - x'}{2\left(y - x\right) + \left(x - x'\right)} dy\\
	&\le \frac{\Delta^{2}}{2\eta}+\frac{\Delta}{\eta}\left\Vert \int_{y}y \mathcal{P}_{x}(y)dy-x\right\Vert_2.
	\end{align*}
	Here, the second line was by expanding and the third line was by $\norm{x - x'}_2 \le \Delta$ and Cauchy-Schwarz. By Proposition \ref{prop:min_perturb}, $\left\Vert \int_{y}y\mathcal{P}_{x}(y)dy-x\right\Vert_2 \leq 2\eta LR$, where by assumption the parameters satisfy the conditions of Proposition~\ref{prop:min_perturb}. Then, combining the two bounds, we have
	\[
	d_{\text{KL}}\left(\prop_x,\prop_{x'}\right) \leq 3LR\Delta + \frac{L\Delta^2}{2} +\frac{\Delta^{2}}{2\eta} +2LR\Delta + \log(1.05)= 5LR\Delta + \frac{L\Delta^2}{2} +\frac{\Delta^{2}}{2\eta} + \log(1.05).
	\]
	When $\Delta = \tfrac{\sqrt{\eta}}{10}$, $\eta L \le 1$, and $\eta L^2R^2 \leq \thalf$, we have the desired
	\[d_{\text{KL}}\left(\prop_x,\prop_{x'}\right)  \le \frac{\sqrt{\eta}LR}{2} + \frac{L\eta}{200} + \frac{1}{200} + \log(1.05)\le \half.\]
\end{proof}

\subsection{Isoperimetry}
\label{ssec:isoperimetry}

In this section, we prove Lemma~\ref{lem:iso}, which asks to show that $\pih_{\Omega_\delta}$ satisfies a log-isoperimetric inequality \eqref{eq:logiso}. Here, we define $\pih_{\Omega_\delta}$ to be the conditional distribution of the $\pih$ $x$-marginal on set $\Omega_\delta$. We recall this means that for any partition $S_1$, $S_2$, $S_3$ of $\Omega_\delta$,
\[\pih_{\Omega_\delta}(S_3) \ge \frac{1}{2\psi}d(S_1, S_2) \cdot \min\left(\pih_{\Omega_\delta}(S_1), \pih_{\Omega_\delta}(S_2)\right) \cdot \sqrt{\log\left(1 + \frac{1}{\min\left(\pih_{\Omega_\delta}(S_1), \pih_{\Omega_\delta}(S_2)\right)}\right)}.\]
The following fact was shown in \cite{ChenDWY19}.
\begin{lemma}[\cite{ChenDWY19}, Lemma 11]\label{lem:logisostrongly}
	Any $\mu$-strongly logconcave distribution $\pi$ satisfies the log-isoperimetric inequality \eqref{eq:logiso} with $\psi = \mu^{-\half}$.
\end{lemma}
Observe that $\pi_{\Omega_\delta}$, the restriction of $\pi$ to the convex set $\Omega_\delta$, is $\mu$-strongly logconcave by the definition of $\pi$ \eqref{eq:pidef}, so it satisfies a log-isoperimetric inequality. We now combine this fact with the relative density bounds Lemma~\ref{lem:densityratio} to prove Lemma~\ref{lem:iso}.

\restateiso*
\begin{proof}
	Fix some partition $S_1$, $S_2$, $S_3$ of $\Omega_\delta$, and without loss of generality let $\pih_{\Omega_\delta}(S_1) \le \pih_{\Omega_\delta}(S_2)$. First, by applying Corollary~\ref{corr:densityratiodelta}, which shows $\tfrac{d\pi}{d\pih}(x) \in [\thalf, 2]$ everywhere in $\Omega_\delta$, we have the bounds
	\[\frac{1}{2}\pi_{\Omega_\delta}(S_1) \leq \pih_{\Omega_\delta}(S_1) \leq 2\pi_{\Omega_\delta}(S_1), \; \frac{1}{2}\pi_{\Omega_\delta}(S_2) \leq \pih_{\Omega_\delta}(S_2) \leq 2\pi_{\Omega_\delta}(S_2),\;\text{and} \;\pih_{\Omega_\delta}(S_3) \geq \frac{1}{2}\pi_{\Omega_\delta}(S_3).\] 
	Therefore, we have the sequence of conclusions
	\begin{align*}
	\pih_{\Omega_\delta}(S_3) &\ge \frac{1}{2}\pi_{\Omega_\delta}(S_3) \\
	&\ge \frac{d(S_1, S_2)\sqrt{\mu}}{4} \cdot \min\left(\pi_{\Omega_\delta}(S_1), \pi_{\Omega_\delta}(S_2)\right) \cdot \sqrt{\log\left(1 + \frac{1}{\min\left(\pi_{\Omega_\delta}(S_1), \pi_{\Omega_\delta}(S_2)\right)}\right)}\\
	&\ge \frac{d(S_1, S_2)\sqrt{\mu}}{8} \cdot \pih_{\Omega_\delta}(S_1) \cdot \sqrt{\log\left(1 + \frac{1}{2\pih_{\Omega_\delta}(S_1)}\right)} \\
	&\ge \frac{d(S_1, S_2)\sqrt{\mu}}{16}\cdot \pih_{\Omega_\delta}(S_1) \cdot \sqrt{\log\left(1 + \frac{1}{\pih_{\Omega_\delta}(S_1)}\right)}.
	\end{align*}
	Here, the second line was by applying Lemma~\ref{lem:logisostrongly} to the $\mu$-strongly logconcave distribution $\pi_{\Omega_\delta}$, and the final line used $\sqrt{\log(1 + \alpha)} \le 2\sqrt{\log(1 + \tfrac{\alpha}{2})}$ for all $\alpha > 0$.
\end{proof}

\subsection{Correctness of $\yor$}
\label{ssec:yoracle}
In this section, we show how we can sample $y$ efficiently in the alternating scheme of the algorithm \sjd, within an extremely high probability region. Specifically, for any $x$ with $\norm{x - x^*}_2 \le \sqrt{\kappa d\log(16\kappa/\delta)} \cdot R_\delta$, where $R_\delta$ is defined in \eqref{eq:omegadelta}, we give a method for implementing 
\[\text{draw } y \propto \exp\left(-f(y) - \frac{1}{2\eta}\norm{y - x}_2^2\right)dy. \]
The algorithm is Algorithm~\ref{alg:yor}, which is a simple rejection sampling scheme.

\begin{algorithm}[ht!]\caption{$\yor(f, x, \eta, \delta)$}
	\label{alg:yor}
	\textbf{Input:} $L$-smooth, $\mu$-strongly convex $f: \R^d \to \R$ with minimizer $x^*$, $\eta > 0$, $\delta \in [0, 1]$, $x \in \R^d$.\\
	\textbf{Output:} If $\norm{x - x^*}_2 \le \sqrt{\kappa d\log(16\kappa/\delta)}\cdot R_\delta$, return exact sample from distribution with density $\propto \exp(-f(y) - \tfrac{1}{2\eta}\norm{y - x}_2^2)$ (see \eqref{eq:omegadelta} for definition of $R_{\delta}$). Otherwise, return sample within $\delta$ TV from distribution with density $\propto \exp(-f(y) - \tfrac{1}{2\eta}\norm{y - x}_2^2)$.
	\begin{algorithmic}[1]
		\If{$\norm{x - x^*}_2 \le \sqrt{\kappa d\log(16\kappa/\delta)}\cdot R_\delta$}
		\While{\textbf{true}}
		\State Draw $y \sim \Nor(x - \eta \nabla f(x), \eta \id)$
		\State $\tau \sim \text{Unif}[0, 1]$
		\If{$\tau \le  \exp(f(x) + \inprod{\nabla f(x)}{y - x} - f(y))$}
		\State \Return $y$
		\EndIf
		\EndWhile
		\EndIf
		\State \Return Sample $x$ within TV $\delta$ from density $\propto \exp(-f(y) - \tfrac{1}{2\eta}\norm{y - x}_2^2)$ using \cite{ChenDWY19}
	\end{algorithmic}
\end{algorithm}
We recall that we gave guarantees on rejection sampling procedures in Lemma~\ref{lem:reject} (an ``exact'' version of Lemma~\ref{lem:reject_approx_proof} and Corollary~\ref{corr:unbiased_reject_approx}). We now prove Lemma~\ref{lem:yor} via a direct application of Lemma~\ref{lem:reject}.

\restateyor*
\begin{proof}
	For $\norm{x - x^*}_2 \le \sqrt{\kappa d\log(16\kappa/\delta)}\cdot R_\delta$, $\yor$ is a rejection sampling scheme with 
	\[p(y) = \exp\left(-f(y) - \frac{1}{2\eta}\norm{y - x}_2^2\right),\; \hp(y) = \exp\left(-f(x) - \inprod{\nabla f(x)}{y - x} -\frac{1}{2\eta}\norm{y - x}_2^2\right).\]
	It is clear that $p(y) \le \hp(y)$ everywhere by convexity of $f$, so we may choose $C = 1$. To bound the expected number of iterations and obtain the desired conclusion, Lemma~\ref{lem:reject} requires a bound on
	\begin{equation}\label{eq:phatpratio}\frac{\int_y \exp\left(-f(x) - \inprod{\nabla f(x)}{y - x} -\frac{1}{2\eta}\norm{y - x}_2^2\right)dy}{\int_y \exp\left(-f(y) - \frac{1}{2\eta}\norm{y - x}_2^2\right)dy},\end{equation}
	the ratio of the normalization constants of $\hp$ and $p$. First, by Fact~\ref{fact:gaussz},
	\[\int_y \exp\left(-f(x) - \inprod{\nabla f(x)}{y - x} -\frac{1}{2\eta}\norm{y - x}_2^2\right)dy = \exp\left(-f(x) + \frac{\eta}{2}\norm{\nabla f(x)}_2^2\right)(2\pi\eta)^{\frac{d}{2}}.\]
	Next, by smoothness and Fact~\ref{fact:gaussz} once more,
	\begin{align*}
	\int_y \exp\left(-f(y) - \frac{1}{2\eta}\norm{y - x}_2^2\right)dy 
	&\ge \int_y \exp\left(-f(x) - \inprod{\nabla f(x)}{y - x} - \frac{1 + \eta L}{2\eta}\norm{y - x}_2^2\right) dy \\
	&= \exp\left(-f(x) + \frac{\eta}{2(1 + \eta L)}\norm{\nabla f(x)}_2^2\right)\left(\frac{2\pi \eta}{1 + \eta L}\right)^{\frac{d}{2}}.
	\end{align*}
	Taking a ratio, the quantity in \eqref{eq:phatpratio} is bounded above by
	\begin{align*}\exp\left(\left(\frac{\eta}{2} - \frac{\eta}{2(1 + \eta L)}\right)\norm{\nabla f(x)}_2^2\right)\left(1 + \eta L\right)^{\frac{d}{2}} &\le 1.5\exp\left(\frac{\eta^2 L}{2(1 + \eta L)}\norm{\nabla f(x)}_2^2\right)\\
	&\le 1.5\exp\left(\frac{\eta^2 L^3}{2}\cdot \left(\frac{16\kappa d^2\log^2(16\kappa/\delta)}{\mu}\right)\right) \le 2.\end{align*}
	The first inequality was $(1 + \eta L)^{\frac{d}{2}} \le 1.5$, the second used smoothness and the assumed bound on $\norm{x - x^*}_2$, and the third again used our choice of $\eta$.
\end{proof}  	%

\section{Structural results}
\label{sec:structural}

Here, we prove two structural results about distributions whose negative log-densities are small perturbations of a quadratic, which obtain tighter concentration guarantees compared to naive bounds on strongly logconcave distributions. They are used in obtaining our bounds in Section~\ref{sec:ingredients} (and for the warm start bounds in Section~\ref{sec:improve}), but we hope both the statements and proof techniques are of independent interest to the community. Our first structural result is a bound on normalization constant ratios, used throughout the paper. 

\begin{proposition}
\label{prop:normalizationratio}
	Let $f: \R^d \rightarrow \R$ be  $\mu$-strongly convex with minimizer $x^*$, and let $\lambda > 0$. Then,
	\[\frac{\int \exp(-f(x)) dx}{\int \exp\left(-f(x) - \frac{1}{2\lambda}\norm{x - x^*}_2^2\right)dx} \le \left(1 + \frac{1}{\mu\lambda}\right)^{\frac{d}{2}}. \]
\end{proposition}
\begin{proof}
	Define the function 
	\[R(\alpha) \defeq \frac{\int \exp\left(-f(x) - \frac{1}{2\lambda\alpha}\norm{x - x^*}_2^2\right) dx}{\int \exp\left(-f(x) - \frac{1}{2\lambda}\norm{x - x^*}_2^2\right)dx}. \]
	Let $d\pi_\alpha(x)$ be the density proportional to $\exp\left(-f(x) - \tfrac{1}{2\lambda\alpha}\norm{x - x^*}_2^2\right)dx$. We compute
	\begin{align*}\frac{d}{d\alpha} R(\alpha) &= \int \frac{\exp\left(-f(x) - \frac{1}{2\lambda\alpha}\norm{x - x^*}_2^2\right)}{\int \exp\left(-f(x) - \frac{1}{2\lambda}\norm{x - x^*}_2^2\right)dx} \frac{1}{2\lambda\alpha^2} \norm{x - x^*}_2^2 dx \\
	&= \frac{R(\alpha)}{2\lambda\alpha^2} \int \frac{\exp\left(-f(x) - \frac{1}{2\lambda\alpha}\norm{x - x^*}_2^2\right)\norm{x - x^*}_2^2}{\int \exp\left(-f(x) - \frac{1}{2\lambda\alpha}\norm{x - x^*}_2^2\right)dx}   dx \\
	&= \frac{R(\alpha)}{2\lambda\alpha^2}\int \norm{x - x^*}_2^2 d\pi_\alpha(x) \le \frac{R(\alpha)}{2\alpha} \cdot \frac{d}{\mu\lambda\alpha + 1}. \end{align*}
	Here, the last inequality was by Fact~\ref{fact:distxstar}, using the fact that the function $f(x) + \tfrac{1}{2\lambda\alpha}\norm{x - x^*}_2^2$ is $\mu + \tfrac{1}{\lambda\alpha}$-strongly convex. Moreover, note that $R(1) = 1$, and
	\[\frac{d}{d\alpha}\log\left(\frac{\alpha}{\mu\lambda\alpha + 1}\right) = \frac{1}{\alpha} - \frac{\mu\lambda}{\mu\lambda\alpha + 1}=  \frac{1}{\mu\lambda\alpha^2 + \alpha}.\]
	Solving the differential inequality
	\[\frac{d}{d\alpha} \log(R(\alpha)) = \frac{dR(\alpha)}{d\alpha} \cdot \frac{1}{R(\alpha)} \le \frac{d}{2} \cdot \frac{1}{\mu\lambda\alpha^2 + \alpha}, \]
	we obtain the bound for any $\alpha \ge 1$ (since $\log(R(1)) = 0$)
	\[\log(R(\alpha)) \le \frac{d}{2}\log\left(\frac{\mu\lambda\alpha + \alpha}{\mu\lambda\alpha + 1}\right) \implies R(\alpha) \le \left(\frac{\mu\lambda\alpha + \alpha}{\mu\lambda\alpha + 1}\right)^{\frac{d}{2}} \le \left(1 + \frac{1}{\mu\lambda}\right)^{\frac{d}{2}}.\]
	Taking a limit $\alpha \rightarrow \infty$ yields the conclusion.
\end{proof}

Our second structural result uses a similar proof technique to show that the mean of a bounded perturbation $f$ of a Gaussian is not far from its mode, as long as the gradient of the mode is small. We remark that one may directly apply strong logconcavity, i.e.\ a variant of Fact~\ref{fact:distxstar}, to obtain a weaker bound by roughly a $\sqrt{d}$ factor, which would result in a loss of $\Omega(d)$ in the guarantees of Theorem~\ref{thm:mainclaim}. This tighter analysis is crucial in our improved mixing time result. 

Before stating the bound, we apply Fact~\ref{fact:convexshrink} to the convex functions $h(x) = (\theta^\top x)^2$ and $h(x) = \norm{x}_2^4$ to obtain the following conclusions which will be used in the proof of Proposition~\ref{prop:min_perturb}.

\begin{corollary}
	\label{corr:slcmomentbounds}
	Let $\pi$ be a $\mu$-strongly logconcave density. Then,
	\begin{enumerate}
		\item $\E_{\pi}[(\theta^\top(x - \E_\pi[x]))^2] \le \mu^{-1}$, for all unit vectors $\theta$.
		\item $\E_{\pi}[\norm{x - \E_\pi[x]}_2^4] \le 3d^2\mu^{-2}$.
	\end{enumerate}
\end{corollary}

\begin{proposition}
	\label{prop:min_perturb}
	Let $f: \R^d \rightarrow \R$ be $L$-smooth and convex with minimizer $x^*$, let $x \in \R^d$ with $\norm{x - x^*}_2 \le R$, and let $d\pi_\eta(y)$ be the density proportional to $\exp\left(-f(y) - \tfrac{1}{2\eta}\norm{y - x}_2^2\right)dy$. Suppose that $\eta \leq \min\left(\tfrac{1}{2L^2R^2},\tfrac{ R^2}{400d^2}\right)$. Then,
	\[\norm{\E_{\pi_\eta}[y] - x}_2 \le 2\eta LR. \]
\end{proposition}
\begin{proof}
	Define a family of distributions $\pi^\alpha$ for $\alpha \in [0, 1]$, with
	\[d\pi^\alpha(y) \propto \exp\left(-\alpha\left(f(y) - f(x) - \inprod{\nabla f(x)}{y - x}\right) - f(x) - \inprod{\nabla f(x)}{y - x} - \frac{1}{2\eta}\norm{y - x}_2^2\right)dy. \]
	In particular, $\pi^1 = \pi_\eta$, and $\pi^0$ is a Gaussian with mean $x - \eta\nabla f(x)$. We define $\bya \defeq \E_{\pi_\alpha}[y]$, and
	\[\sya \defeq \argmin_y \left\{\alpha\left(f(y) - f(x) - \inprod{\nabla f(x)}{y - x}\right) + f(x) + \inprod{\nabla f(x)}{y - x} + \frac{1}{2\eta}\norm{y - x}_2^2\right\}.\]
	Define the function $D(\alpha) \defeq \norm{\bya - x}_2$, such that we wish to bound $D(1)$. First, by smoothness
	\[D(0) = \norm{\E_{\pi_0}[y] - x}_2 = \norm{\eta\nabla f(x)}_2 \le \eta LR.\]
	Next, we observe
	\[\frac{d}{d\alpha}D(\alpha) = \inprod{\frac{\bya - x}{\norm{\bya - x}_2}}{\frac{d\bya}{d\alpha}} \le \norm{\frac{d\bya}{d\alpha}}_2. \]
	In order to bound $\norm{\tfrac{d\bya}{d\alpha}}_2$, fix a unit vector $\theta$. We have
	\begin{equation}
	\label{eq:bigderivbound}
	\begin{aligned}\inprod{\frac{d\bya}{d\alpha}}{\theta} &= \frac{d}{d\alpha}\inprod{\int (y - x) d\pi^\alpha(y)}{\theta} \\
	&= \int \inprod{y - x}{\theta} (f(x) + \inprod{\nabla f(x)}{y - x} - f(y)) d\pi^\alpha(y) \\
	&\le \sqrt{\int (\inprod{y - x}{\theta})^2 d\pi^\alpha(y)}\sqrt{\int (f(x) + \inprod{\nabla f(x)}{y - x} - f(y))^2 d\pi^\alpha(y)}\\
	&\le \sqrt{\int (\inprod{y - x}{\theta})^2 d\pi^\alpha(y)}\sqrt{\int \frac{L^2}{4}\norm{y - x}_2^4 d\pi^\alpha(y)}.
	\end{aligned}
	\end{equation}
	The third line was Cauchy-Schwarz and the last line used smoothness and convexity, i.e.
	\begin{align*}
	-\frac{L}{2}\norm{y - x}_2^2 \le f(x) + \inprod{\nabla f(x)}{y - x} - f(y) \le 0.
	\end{align*}
	We now bound these terms. First,
	\begin{equation}
	\label{eq:firsttermcs}
	\begin{aligned}
	\int (\inprod{y - x}{\theta})^2d\pi^\alpha(y) &\le 2\int(\inprod{y -\bya}{\theta})^2d\pi^\alpha(y) +    2\int(\inprod{\bya - x}{\theta})^2d\pi^\alpha(y) \\
	&\le 2\eta + 2\norm{\bya - x}_2^2 = 2\eta + 2D(\alpha)^2.
	\end{aligned}
	\end{equation}
	Here, we applied the first part of  Corollary~\ref{corr:slcmomentbounds}, as $\pi^\alpha$ is $\eta^{-1}$-strongly logconcave, and the definition of $D(\alpha)$. Next, using for any $a, b \in \R^d$, $\norm{a + b}_2^4 \le (\norm{a}_2 + \norm{b}_2)^4 \le 16\norm{a}_2^4 + 16\norm{b}_2^4$, we have
	\begin{equation}
	\label{eq:secondtermcs}
	\begin{aligned}
	\int \frac{L^2}{4}\norm{y - x}_2^4 d\pi^\alpha(y) &\le \int 4L^2\norm{y - \bya}_2^4 d\pi^\alpha(y) + \int 4L^2\norm{x - \bya}_2^4 d\pi^\alpha(y) \\
	&\le 12L^2 d^2 \eta^2 + 4L^2 D(\alpha)^4.
	\end{aligned}
	\end{equation}
	Here, we used the second part of Corollary~\ref{corr:slcmomentbounds}. Maximizing \eqref{eq:bigderivbound} over $\theta$, and applying \eqref{eq:firsttermcs}, \eqref{eq:secondtermcs},
	\begin{align}
	\label{eq:deri_D_bound}
	\frac{d}{d\alpha}D(\alpha) \le \norm{\frac{d\bya}{d\alpha}}_2 &\le \sqrt{8L^2(\eta + D(\alpha)^2)(3d^2\eta^2 + D(\alpha)^4)} \notag \\
	& \leq 4L(\sqrt{\eta} +D(\alpha))\cdot \max(2\eta d, D(\alpha)^2).
	\end{align}
	Assume for contradiction that $D(1) > 2\eta LR$, violating the conclusion of the proposition. By continuity of $D$, there must have been some $\bal \in (0, 1)$ where $D(\bal) = 2\eta LR$, and for all $0 \leq \alpha < \bal$, $D(\alpha) < 2\eta L R$. By the mean value theorem, there then exists $0 \le \hal \le \bal$ such that 
	\[
	\frac{dD(\hal)}{d\alpha} = \frac{D(\bal) - D(0)}{\bal} > \eta LR.
	\]
	On the other hand, by our assumption that $2\eta L^2R^2\leq 1$, for any $d \ge 1$ it follows that
	\[2\eta d \geq 4\eta^2 L^2R^2 > D(\hal)^2,\; \sqrt{2\eta} \geq 2\eta L R > D(\hal).\]
	Then, plugging these bounds into \eqref{eq:deri_D_bound} and using $\sqrt{\eta} + D(\hal) \le \tfrac{5}{2}\sqrt{\eta}$ as $\sqrt{2} \le \tfrac{3}{2}$,
   \[
   \frac{d}{d\alpha}D(\hal) \le 4L \cdot \frac 5 2\sqrt{\eta} \cdot 2\eta d = 20\sqrt{\eta}\frac{d}{R} \cdot \eta LR \leq  \eta LR.
   \]
   We used $\eta \leq \tfrac{R^2}{400d^2}$ in the last inequality. This is a contradiction, implying $D(1) \le 2\eta LR$.
\end{proof} 	\end{appendix}

\end{document}